\documentclass[a4paper,12pt,oneside,reqno]{amsart}

\usepackage[foot]{amsaddr}

\usepackage[latin1]{inputenc}
\usepackage{hyperref}
\usepackage[headinclude,DIV13]{typearea}
\areaset{16cm}{24.7cm}
\usepackage{txfonts,amssymb,amsmath,amsthm,bbm,wasysym}
\usepackage[pdftex]{graphicx}
\usepackage[margin=2.7cm]{geometry}
\usepackage[framemethod=tikz]{mdframed}
\usepackage{cancel} 
\usepackage{xcolor}
\usepackage[markup=nocolor,deletedmarkup=xout]{changes}

\usepackage{graphicx}
\graphicspath{ {figures/} }
\usepackage{caption}
\usepackage{subcaption}
\usepackage{mathrsfs}
  
\usepackage{xcolor}
\usepackage{enumerate}
\usepackage{nicefrac}
\theoremstyle{plain}
\newtheorem{theorem}{Theorem}[section]
\newtheorem{lemma}[theorem]{Lemma}

\newtheorem{corollary}[theorem]{Corollary}

\theoremstyle{definition}
\newtheorem{definition}[theorem]{Definition}
\newtheorem{assumption}{Assumption}

{\itshape}{\rmfamily}
{\itshape}{\rmfamily}

\theoremstyle{remark}
\newtheorem{remark}{Remark}

\def\paragraph#1{\noindent \textbf{#1}}

\numberwithin{equation}{section}

\def\dd{\mathrm{d}}

\def\a{\alpha}
\def\b{\beta}
\def\d{\delta}
\def\ve{\varepsilon}
\def\g{\gamma}
\def\l{\lambda}
\def\s{\sigma}
\def\t{\tau}

\def\O{\Omega}
\def\S{\Sigma}
\def\R{{\Bbb R}} 
\def\N{{\Bbb N}}
\def\P{{\Bbb P}}

\def\E{{\Bbb E}}

\def\T{{\Bbb T}}

\let\cal=\mathcal
\def\LL{{\cal L}}
\def\MM{{\cal M}}

 \def \h {{\eta}}

\def\eps{\varepsilon}
\def\v{\mathbf{v}}

\def\ee{\mathrm{e}}

\def\lalpha{{\lfloor\alpha\rfloor}}

\newcommand{\ifct}[1]{\mathbbm{1}_{ {#1} }}
\newcommand{\Exd}[1]{\mathbb{E}\left[ {#1}\right]}
\newcommand{\Prob}[1]{\mathbb{P}\left({#1}\right)}
\newcommand{\abs}[1]{\left\lvert#1\right\rvert}
\newcommand{\norm}[1]{\left\lVert#1\right\rVert}
\newcommand{\dset}[1]{\left\{ {#1} \right\}}
\newcommand{\gauss}[1]{{\left\lfloor{#1}\right\rfloor}}
\newcommand{\lK}{{\lambda_K}}
\newcommand{\ext}{{\mathrm{ext}}}
\newcommand{\av}{{\mathrm{av}}}
\newcommand{\dangle}[1]{{\left\langle {#1} \right\rangle}}

\begin{document}

\title{Effective growth rates in a periodically changing environment: From mutation to invasion}
\author[1]{Manuel Esser$^1$}
\author[2]{Anna Kraut$^{2,3}$}
\address[1]{Institut f\"ur Angewandte Mathematik, Rheinische Friedrich-Wilhelms-Universit\"at, Endenicher Allee 60, 53115 Bonn, Germany}
\email{manuel.esser@uni-bonn.de}
\address[2]{School of Mathematics, University of Minnesota - Twin Cities, 206 Church St SE, Minneapolis, MN 55455, USA}
\address[3]{Department of Mathematics, Statistics, and Computer Science, St.\ Olaf College, 1520 St.\ Olaf Avenue, Northfield, MN 55057, USA}
\email{kraut1@stolaf.edu}


\thanks{This work was partially supported by the Deutsche Forschungsgemeinschaft (DFG, German Research Foundation) under Germany's Excellence Strategy GZ 2047/1, Projekt-ID 390685813 and through Project-ID 211504053 - SFB 1060. The authors thank Anton Bovier for stimulating discussion, feedback on the manuscript, and for facilitating visits of A.\ Kraut to Bonn to work on this project. The authors thank the anonymous referees for their comments and questions that helped to improve the manuscript.}


\begin{abstract}
We consider a stochastic individual-based model of adaptive dynamics for an asexually reproducing population with mutation, with linear birth and death rates, as well as a density-dependent competition.  To depict repeating changes of the environment, all of these parameters vary over time as piecewise constant and periodic functions, on an intermediate time-scale between those of stabilization of the resident population (fast) and exponential growth of mutants (slow). Studying the growth of emergent mutants and their invasion of the resident population in the limit of small mutation rates for a simultaneously diverging population size, we are able to determine their effective growth rates. We describe this growth as a mesoscopic scaling-limit of the orders of population sizes, where we observe an averaging effect of the invasion fitness.  Moreover, we prove a limit result for the sequence of consecutive macroscopic resident traits that is similar to the so-called trait-substitution-sequence.
\end{abstract}

\maketitle

\section{Introduction}
	
	Mathematical approaches to understanding the long-term evolution of populations have a long history and can even be traced back to ideas of Malthus in 1798 \cite{Malt1798}. The study of heterogeneous populations is of particular interest as it allows to analyse the diversity and the interaction of species as they adapt over time.
	The driving mechanisms, ecology and evolution, which are addressed by models of adaptive dynamics, may strongly depend on the environment a population is living in. Since realistically this environment cannot be assumed to stay constant over time, we study the effects of periodic changes of the model parameters in this paper. From an application point of view, this is for example motivated by the seasonal changes of temperature, humidity, accessibility of nutrition and other resources, which may effect the fertility of individuals and thus directly have an impact on the population's growth \cite{Ewing16,LitKla01}. Another interesting example is pulsed drug-based therapy for infectious diseases or cancer. Dependent on the treatment protocol, the concentrations of drugs may vary over time. Assuming a regular supply of drug, this can be described by periodic changes, leading to varying reproduction rates of the pathogen.
	
	While in both of these cases the population's dynamics are directly affected by the environmental changes on a short time-scale, it is reasonable to expect some averaging and thus a macroscopic trend of growth or shrinking on a larger time-scale. This averaging effect is what we study rigorously in the present paper. Starting from a model that describes the individuals' dynamics on a microscopic level, i.e.\ taking into account interactions between single individuals, we derive results for the effective mesoscopic growth rates of subpopulations of intermediate size, i.e.\ that consist of a larger number of individuals but are still negligible with respect to the total population size. Moreover, we give an explicit description of the macroscopic evolution of the whole population process, tracing the evolution of the dominating trait, in the large population limit. A crucial aspect in this is to understand under which conditions we can observe the emergence and growth of new types or even the replacement of resident traits by fitter mutants.
	
	We consider a variation of the stochastic individual-based model of adaptive dynamics that has been introduced by Fournier and Méléard \cite{FoMe04} and since then was studied for a broad spectrum of scaling limits and model extensions (see e.g.\ \cite{Bov21} for an overview and \cite{BaBo18,Sma17,NeuBo17,BoCoNeu18,BlaTob20}). Its aim is to study the interplay of ecology and evolution, i.e.\ both the short-term effects of competitive interactions of different subpopulations and the long-term effects of occurrence and fixation of new mutant species. Since our interest lies in analysing the effects of time-dependent changes of ecological parameters on the long-term evolution of a population, these models of adaptive dynamics are naturally helpful.
	
	As one of the first results on the individual-based model, Champagnat was able to show that certain assumptions on the scaling of large populations and very rare mutations lead to a separation of the time-scales of ecology and evolution, which is a fundamental principle of adaptive dynamics. Under the aforementioned assumptions, Champagnat derived convergence to the trait-substitution-sequence (TSS) \cite{Cha06} and, together with M\'el\'eard, the polymorphic-evolution-sequence (PES) \cite{ChMe11}. On an accel\-erated time-scale, these sequences describe how the macroscopic population essentially jumps between (monomorphic or polymorphic) equilibria of different Lotka-Volterra systems.
	A broader spectrum of more frequent mutations was investigated by Bovier, Coquille, and Smadi for a simple trait space with a valley in the fitness landscape \cite{BoCoSm19}. This work laid the basis for the more general study of moderately rare mutations in \cite{CoKrSm21}, under collaboration of Kraut. The latter provides both the description of a macroscopic limit process, which consists of (deterministic) jumps between equilibrium states, as well as a mesoscopic limit result for the growth and decline of all subpopulations, observable on a logarithmic time-scale.
	
	Despite the variety of different scenarios that have already been analysed, all of these previous works ask for the parameters of the population process to be constant over time. In the present paper, we break with this assumption and allow for periodic parameter changes. As before, we study the limit of a diverging carrying capacity $K\nearrow\infty$, which scales the order of the total population's size, and choose moderate mutation probabilities $\mu_K\searrow0$. In addition, we introduce a finite number of parameter constellations, which repeat periodically on time intervals (phases) with fixed length of order $\lK$, to model a changing environment.  These parameter constellations vary the individual birth, death, and competition rates, which in particular determine the fitness, or growth rate, of the different subpopulations. Consequently, both the sign of the fitness, resulting in growth or shrinkage of the subpopulations, and the fitness relations between different types may change between phases.
	
	We choose an intermediate time-scale of $1\ll\lK\ll\ln K$ for these environmental changes. As a result, the environment stays stable enough for the macroscopic resident population to adapt to it in between parameter changes, but changes occur fast enough to influence the growth of the micro- and mesoscopic subpopulations in between invasions. Under these assumptions, we can observe an averaging effect on the level of mutant growth. Similarly to \cite{CoKrSm21}, we prove a mesoscopic convergence result for the orders of population sizes $K^\b$ of all subpopulations. In the limit of $K\to \infty$, on the time-scale $\ln K$ of exponential growth, these exponents $\b$ converge to deterministic, piecewise affine functions $\bar{\b}$ that can be described by a recursive algorithm. The slopes of these functions are determined by the effective (time-average) fitness of the subpopulations. Based on this mesoscopic characterisation, we further derive a substitution-sequence on the same time-scale, des\-cribing the macroscopic jumps of the population between successive resident traits.
	
	The fact that the environmental parameters now change on an intermediate time-scale at a first glance seems to be only a small variation of the former models. However, a couple of non-trivial difficulties arise in all parts of the established proof strategies:
	First of all, since time spans of order $\ln K$ consist of asymptotically infinitely many parameter phases that need to be concatenated, the way in which a large deviation principle is usually applied for these types of processes to ensure stability of the resident population in between invasions (see e.g.\ \cite{Cha06}) is not sufficient. To obtain a quantification of the speed at which the probability of exit from a domain within a $\l_K$-time span tends to zero, we instead use potential theoretic arguments similar to Baar, Bovier, and Champagnat \cite{BaBoCh17}.  Moreover, to take care of the short $O(1)$ times of re-equilibration after a parameter change, we study the speed of convergence in the standard convergence result of Ethier and Kurtz \cite{EtKu86}.
	
	Second, we need to extend the general growth results for branching processes (with immigration) of Champagnat, M\'el\'eard and Tran \cite{ChMeTr19} to periodically changing para\-meters. This in particular requires a careful consideration of small populations, i.e.\ in the case of extinction or a newly emerging mutant population. Here we study the distribution function of the extinction time, making use of estimates on the probability generating function of Galton-Watson processes.
	
	Finally, we need to carefully consider the event of a mutant population becoming macroscopic. Here we need to choose the stopping time, after which we start the comparison to the deterministic system, such that the invading mutant is guaranteed to be in a phase with positive invasion fitness and successfully fixate as the new resident trait within a time of order 1.
	
	Changing environmental parameters have been previously introduced to a number of mathematical population models. While we cannot give an extensive review here, we want to mention a few examples. One popular scenario is that of a shifting optimal trait (mostly in deterministic ODE or PDE models), where the fitness of all other traits depends only on the distance to the current optimum \cite{JaDe18,RoPaBoMa20,GaCo23}. 
	A common observation is the importance of the relation between the time-scales of environmental shifts and trait changes (mutation speed and step size), which determines whether the population can successfully adapt or not.
	A scenario similar to the one of the present paper is that of periodically changing environments, in both deterministic and stochastic models \cite{MuYo20,TaWeAsMo20,BuFr22}. Here, previous studies have focussed on the dynamics of a fixed (usually small) number of competing traits without mutation.
	 As we observe in this paper, time-scales again play a crucial role, where all of the above works find that sufficiently fast fluctuations lead to the population evolving according to time-averaged effective parameters. 
	 Other questions have been addressed by various authors, for example the dynamics of phenotypic switching and dormancy for non-competitive multi-type systems in more or less randomly fluctuating environment \cite{KuLei05,MaSm08,DoMaBa11,JoWa14,BlHeSl21}. To the best of our knowledge, the dynamics of stochastic models with periodically changing environment, general fitness landscapes, and newly emerging mutant types are still an open problem.
	 
	 Similar to some of these approaches it will be interesting to extend the adaptive dynamics model of this paper to more generally changing parameters. Modelling the environmental parameters as continuous functions or as a stochastic process itself, where jump times and transitions are random, can allow for a more realistic depiction of biological scenarios that either only change gradually or less regularly. Our results in this manuscript are meant as a first step to establish techniques of how to study this new class of models and lay the basis for future research that is already in progress. 
	
	The remainder of this paper is structured as follows. In Section \ref{sec:2.1_model}, the individual-based model for a population in a time-dependent environment is introduced rigorously. We point out some key quantities, like equilibrium states and invasion fitness, in Section \ref{sec:2.2_quantities}. Finally, we describe the behaviour of the limit process in terms of an inductive algorithm in Section \ref{sec:2.3_results} and state our main convergence results. Chapter \ref{sec:3_heuristics_discussion} provides a discussion of the general heuristics and the necessity of some of our assumptions. Moreover, we give an outlook on possible extensions of this approach. The proofs of the main results of this paper can be found in Chapter \ref{sec:4_proofs}. The technical results on birth death processes, which lay the basis for these proofs, are discussed in Appendices \ref{app:A_attraction_equil}, \ref{app:B_branching_processes}, and \ref{app:C}.


\section{Model and Main Results}

\subsection{Individual-based model in a time-dependent evironment}
\label{sec:2.1_model}
We consider a popu\-lation  that is composed of a finite number of asexual reproducing individuals. Each of them is characterized by a genotypic trait, taken from a finite trait space that is given by a (possibly directed) graph $G=(V,E)$. Here, the set of vertices $V$ represents the possible traits that individuals can carry. The set of edges $E$ marks the possibility of mutation between traits. We start out with a microscopic, individual-based model with logistic growth.

To extend the basic model to one with a periodically changing environment, we consider a finite number $\ell\in\N$ of phases. For each phase $i=1,\cdots,\ell$ and all traits $v,w\in V$, we  introduce the following biological parameters:
\begin{itemize}
	\item $b^i_v\in\R_+$, the \textit{birth rate} of an individual of trait $v$ during phase $i$,
	\item $d^i_v\in\R_+$, the \textit{(natural) death rate} of an individual of trait $v$ during phase $i$,
	\item $c^i_{v,w}\in\R_+$, the \textit{competition} imposed by an individual of trait $w$  onto an individual of trait $v$ during phase $i$,
	\item $K\in\N,$ the \textit{carrying capacity} that scales the environment's capacity to support life,
	\item $\mu_K\in[0,1]$, the \textit{probability of mutation} at a birth event (phase-independent),
	\item $m_{v,\cdot}\in\MM_p(V)$, the \textit{law of the trait of a mutant} offspring produced by an individual of trait $v$ (phase-independent).
\end{itemize}

To ensure logistic growth and ensure the possibility of mutation only along the edges of our trait graph, we make the following assumptions on our parameters.
\begin{assumption}\label{ass:selfcomp}
	\begin{enumerate}[(a)]
		\item For every $v\in V$ and $i=1,\cdots\ell$, $c^i_{v,v}>0$. 
		\item $m_{v,v}=0$, for all $v\in V$, and $m_{v,w}>0$ if and only if $(v,w)\in E$. 
	\end{enumerate}
\end{assumption}

Rescaling the competition by $K$ (cf.\ \eqref{eq:time_dep_generator} below) leads to a total population size of order $K$. We adapt the following terminology: As $K\to\infty$, subpopulations of certain traits are referred to as
\begin{itemize}
\item \textit{microscopic} if they are of order 1,
\item \textit{macroscopic} if they are of order $K$,
\item \textit{mesoscopic} if they are of order strictly between 1 and $K$.
\end{itemize}

For a new mutant, reaching a macroscopic population size through exponential growth takes a time of order $\ln K$.
For a resident population, it takes a time of order $1$ to reach a small neighbourhood of its new equilibrium after an environmental change. 
In order for environmental changes to happen slow enough such that the resident populations can adapt, but fast enough such that they influence the growth of mutants, we choose
\begin{align}
	1\ll\lK\ll\ln K
\end{align}
as an intermediate time-scale for the length of the $\ell$ phases. For each $i=1,\ldots,\ell$, we assume that the $i$-th phase has length $T_i\lK$, where $T_i>0$. To refer  to the endpoints of the phases, we define $T^\Sigma_j:=\sum_{i=1}^{j}T_i$.

Building on this, we define the time-dependent birth-, death-, and competition rates as the periodic extension of
\begin{align}
	b_v^K(t)=\sum_{i=1}^{\ell}\ifct{t\in [T^\Sigma_{i-1}\lK,T^\Sigma_i\lK)}b^i_v,
\end{align}
and analogously for $d_v^K(t)$ and $c_{v,w}^K(t)$. Note that $b^i_v$ and $b^K_v$ are very similar in notation. To make the distinction clear, we always use the upper index $i$ to refer to the constant parameter in phase $i$ and the index $K$ to refer to the time-dependent parameter function for carrying capacity $K$, and use the same convention in comparable cases.

For any $K$, the evolution of the population over time is described by a Markov process $N^K$ with values in $\mathbb{D}(\R_+,\N_0^V)$. $N^K_v(t)$ denotes the number of individuals of trait $v\in V$ that are alive at time $t\geq 0$. The process is characterised by its infinitesimal generator
\begin{align}\label{eq:time_dep_generator}
	\left(\LL_t^K\phi\right)(N)=&\sum_{v\in V}(\phi(N+\d_v)-\phi(N))\left(N_vb_v^K(t)(1-\mu_K)+\sum_{w\in V}N_wb_w^K(t)\mu_Km_{w,v}\right)\notag\\
	&+\sum_{v\in V}(\phi(N-\d_v)-\phi(N))N_v\left(d_v^K(t)+\sum_{w\in V}\frac{c_{v,w}^K(t)}{K}N_w\right),
\end{align}
where $\phi:\N_0^V\to\R$ is measurable and bounded and $\d_v$ denotes the unit vector at $v\in V$. The process can be constructed algorithmically following a Gillespie algorithm \cite{Gill76}. Alternatively, the process can be represented via Poisson measures (see \cite{FoMe04}), a representation that is used in the proofs of this paper.

\subsection{Important quantities}
\label{sec:2.2_quantities}
In this paper we study the typical behaviour of the processes $(N^K,K\in\N)$ for large populations, i.e.\ as $K\to\infty$. A classical law of large numbers \cite{EtKu86} states that the rescaled processes $N^K/K$ converge on finite time intervals to the solution of a system of Lotka-Volterra equations. The study of these equations is central to determine the short term evolution of the processes $N^K$.

\begin{definition}[Lotka-Volterra system]
	For a phase $i\in\dset{1,\dots\ell}$ and a subset $\v\subseteq V$, we refer to the following differential equations as the corresponding \textit{Lotka-Volterra system}:
	\begin{align}\label{LV_system}
		\dot{n}_v(t)= \left( b^i_v - d^i_v - \sum_{w \in \v} c^i_{v,w} n_w(t) \right) n_v(t), \quad v\in\v,\ t\geq0
	\end{align}
\end{definition}

In this work, we focus on the case of a sequence of monomorphic resident traits, meaning that, apart from the invasion phases, only one single (fit) subpopulation is of macroscopic size and fluctuates around its equilibrium size. This monomorphism is ensured by a termination criterion in the construction of the limiting process for our main Theorem \ref{thm:conv_beta} (criterion (d), see also Remark \ref{rem:terminationcrit}). Taking into account the phase-dependent parameters, we denote these \textit{monomorphic equilibria} by $\bar{n}^i_v:=(b^i_v-d^i_v)/c^i_{v,v}$.

Talking about evolution, the most important quantity is fitness. For the individual-based model of adaptive dynamics, the notion of invasion fitness has been shown to be useful. It describes the approximate growth rate of a small population of trait $w$ in a bulk population of trait $v$ in the mutation-free system. To adapt it to the present setting we have to include the phase-dependence.
\begin{definition}[Invasion fitness]
	For each phase $i\in\dset{1,\cdots,\ell}$ and for all traits $v,w\in V$ such that the equilibrium size of $\bar{n}^i_v$ is positive, we denote by
	\begin{align}
		f^i_{w,v}:=b^i_w-d^i_w-c^i_{w,v}\bar{n}^i_v
	\end{align}
	the \textit{invasion fitness} of trait $w$ with respect to the monomorphic resident $v$ in the $i$-th phase. Moreover, we define the time-dependent fitness and the average fitness by the periodic extension of
	\begin{align}
		f^K_{w,v}(t):=\sum_{i=1}^{\ell}\ifct{t\in [T^\Sigma_{i-1}\lK,T^\Sigma_i\lK)}f^i_{w,v} \quad\text{and}\quad
		f^{\text{av}}_{w,v}:=\frac{\sum_{i=1}^\ell T_i f^i_{w,v}}{T^\S_\ell}.
	\end{align}
\end{definition}

Let us now consider multi-step mutations arising along paths within the trait graph $G=(V,E)$. We introduce the graph distance between two vertices $v,w\in V$ as the length of the shortest (directed) connecting path
\begin{align}
	d(v,w):=\min_{\g:v\to w}\abs{\g},
\end{align}
where we use the convention that the minimum over an empty set is $\infty$. Note that $d(v,w)$ is not a distance in the classical sense, as it may not be symmetric in the case of a directed graph.

Since a single birth event causes a mutation with probability $\mu_K$, a macroscopic trait $v$ (size of order $K$) produces subpopulations of a size of order $K\mu_K$ of its neighbouring traits. These traits themselves produce subpopulations of a size of order $K\mu_K^2$ of second order neighbours of $v$. In general, $v$ induces mutant populations of trait $w$ of size of order $K\mu_K^{d(v,w)}$. We study mutation probabilities of the form 
\begin{align}
	\mu_K=K^{-1/\a},\qquad\a\in\R_{>0}\setminus\N.
\end{align}
As a consequence all traits at a distance of at most $\lalpha$ have a size that is non-vanishing for increasing $K$, which means that they can survive. For technical reasons, we exclude $\a\in\N$, see the discussion in Section \ref{sec:3_heuristics_discussion}.

\subsection{Results}
\label{sec:2.3_results}
	The main result of this paper gives a precise description of the orders of the different subpopulation sizes as $K$ tends to infinity. It is convenient to describe the population size of a certain trait $v\in V$ at time $t$ by its $K$-exponent, which is given by the following definition.
	
	\begin{definition}[Order of the population size]
		For all $v\in V$ and all $t\geq 0$, we set
		\begin{align}\label{eq:def_betaK}
			\b^K_v(s):=\frac{\ln(N_v^K(s\ln K)+1)}{\ln K}.
		\end{align}
	\end{definition}
	Note that adding one inside the logarithm is only done to ensure that $\b^K_v(s)=0$ is equivalent to $N^K_v(s\ln K)=0$. Before we state the result below, let us describe the limiting functions $(\bar{\b}_v,v\in V)$. We can define these trajectories up to a stopping time $T_0$ by the following inductive procedure:
	
	Let $v_0\in V$ be the initial macroscopic trait. For simplicity, we assume that the initial orders of population sizes converge in probability to $\bar{\b}_w(0)$ satisfying the constraints
		\begin{align}
			\bar{\b}_w(0)&=\max_{\substack{u\in V}}\left[\bar{\b}_u(0)-\frac{d(u,w)}{\alpha}\right]\lor 0,\label{eq:result_initialcondition}\\
			\bar{\b}_w(0)&=1\ \Leftrightarrow\ w=v_0\label{eq:result_initialcondition_mono}.
		\end{align}
	The increasing sequence of invasion times is denoted by $(s_k)_{k\geq0}$, where $s_0:=0$ and, for $k\geq1$,
	\begin{align}\label{eq:def_s_k}
		s_k:=\inf\{t>s_{k-1}:\exists\ w\in V\backslash v_{k-1}:\bar{\b}_w(t)=1\}.
	\end{align}

	Moreover, we set $v_k$ to be the trait $w\in V\backslash v_{k-1}$ that satisfies $\bar{\b}_w(s_k)=1$, which we assume to be unique in order to proceed (cf. termination criteria below).
	
	For $s_{k-1}\leq t\leq s_k$, for any $w\in V$, $\bar{\b}_w(t)$ is defined by
	\begin{align}\label{eq:def_betabar}
		\bar{\b}_w(t):=\max_{\substack{u\in V}}\left[\bar{\b}_u(s_{k-1})+(t-t_{u,k})_+f^\av_{u,v_{k-1}}-\frac{d(u,w)}{\alpha}\right]\lor 0,
	\end{align}
	where, for any $w \in V$,
	\begin{align} \label{eq:def_t_wk}
		t_{w,k}:=\begin{cases}\inf\{t\geq s_{k-1}:\exists\ u\in V: d(u,w)=1, \bar{\b}_u(t)=\frac{1}{\alpha}\}
			&\text{if }\bar{\b}_w(s_{k-1})=0\\s_{k-1}&\text{else}\end{cases}
	\end{align}
	is the first time in $[s_{k-1},s_k]$ when this trait arises.
	
\begin{remark}\label{rem:betabar}
The formula in \eqref{eq:def_betabar} can heuristically be explained as follows: From time $s_{k-1}$ on, on the $\ln K$-time-scale, every living trait $u$ (i.e.\ traits such that $t\geq t_{u,k}$) grows/shrinks at least at the rate of its own fitness $f^\av_{u,v_{k-1}}$, which would yield $\bar{\beta}_u(t)\approx \bar{\b}_u(s_{k-1})+(t-s_{k-1})f^\av_{u,v_k}$. On top of this, every living trait spreads a $\mu_K=K^{-1/\alpha}$ portion of its population size to its neighbouring traits through mutation. These then pass on a $\mu_K^2$ portion to the second order neighbours and so on. Overall, a trait $w\in V$ receives a $\mu_K^{d(u,w)}=K^{-d(u,w)/\alpha}$ portion of incoming mutants from all living traits $u$, and its actual population size can hence be determined by taking the leading order term, i.e.\ the maximum of all these exponents $\bar{\b}_u(s_{k-1})+(t-s_{k-1})f^\av_{u,v_k}-d(u,w)/\alpha$.
\end{remark}
	
	The stopping time $T_0$, that terminates the inductive construction of the limiting trajectories, is set to $s_k$ if
	\begin{enumerate}[(a)]
		\item there is more than one $w\in V\backslash v_{k-1}$ such that $\bar{\b}_w(s_k)=1$;
		\item there exists $w\in V\backslash v_{k-1}$ such that $\bar{\b}_w(s_k)=0$ and $\bar{\b}_w(s_k-\eps)>0$ for all $\eps>0$ small enough;
		\item there exists $w\in V\backslash v_{k-1}$ such that $s_k=t_{w,k}$;
		\item there exists an $i\in\dset{1,\ldots,\ell}$ such that either $f^i_{v_{k-1},v_k}\geq 0$ or $f^i_{v_k,v_{k-1}}=0$;
		\item there exists an $i\in\dset{1,\ldots,\ell}$ such that  $b^i_{v_k}-d^i_{v_k}\leq 0$.
	\end{enumerate}
	
	\begin{remark}\label{rem:terminationcrit}
		Note that conditions (a)-(c) are purely technical (cf.\ \cite{CoKrSm21}).
		The first part of condition (d) is a sufficient criterion to ensure the principle of \textit{invasion implies fixation}, i.e.\ any mutant trait that reaches a macroscopic populations size replaces the former resident trait and there is no coexistence. The criterion is not necessary and there are other possible scenarios where the invading mutant replaces the resident population (see discussion in Chapter \ref{sec:3_heuristics_discussion}). The second part is again technical and ensures that the $\ve K$-threshold (needed for the approximations by birth death processes) is reached at a time when invasion can take place in finite time, i.e.\ the comparision to the deterministic Lotka-Volterra system is possible (cf.\ the classical result in \cite{EtKu86}) The last condition (e) ensures that the new resident possesses a strictly positive monomorphic equilibrium $\bar{n}^i_{v_k}$ in all phases.
	\end{remark}
	
	\begin{theorem}[Convergence of $\b$]
		\label{thm:conv_beta}
		Let a finite graph $G=(V,E)$ and $\a\in\R_{>0}\setminus\N$ be given and consider the model defined by \eqref{eq:time_dep_generator}. Let $v_0\in V$ and
		assume that, for every $w\in V$, $\b^K_w(0)\to \bar{\b}_w(0)$ in probability, as $K\to\infty$, where the limits satisfy \eqref{eq:result_initialcondition} and \eqref{eq:result_initialcondition_mono}.
		Then, for all fixed $0\leq T\leq T_0$, the following convergence holds in probability, with respect to the $L^\infty([0,T],\R^V_{\geq 0})$ norm
		\begin{align}
			(\b_w^K(t), w\in V)_{t\in [0,T]}\overset{K\to\infty}{\longrightarrow}(\bar{\b}_w(t), w\in V)_{t\in [0,T]},
		\end{align}
		 where $\bar{\b}_w$ are the deterministic, piecewise affine, continuous functions defined in \eqref{eq:def_betabar}.
	\end{theorem}

	\begin{remark}
		Note that we only assume \eqref{eq:result_initialcondition} to ensure convergence at $t=0$. If the $\b^K_w(0)$ converge to initial conditions $\hat{\b}_w(0)$ that do not satisfy this constraint, the orders of the population sizes stabilize in a time of order 1 at $\bar\b_w(0):=\max_{\substack{u\in V}}\left[\hat{\b}_u(0)-\frac{d(u,w)}{\alpha}\right]\lor 0$. These new orders satisfy \eqref{eq:result_initialcondition}. Because the $\beta^K$ describe the population on a $\ln K$-time-scale, this means that we still get convergence on the half-open interval $(0,T]$.
	\end{remark}

	Building upon this detailed description of growth of all living traits, it is natural to ask for the ``visible'' evolution of the population process, i.e.\ the progression of macroscopic traits that dominate the whole system.
	
	\begin{corollary}[Sequence of resident traits]
		\label{cor:seq_residents}
		Let
		\begin{align}
			\nu^K_\ve(t):=\sum_{w\in V: N^K_w(t)>\ve K}\d_w\quad\text{ and }\quad\nu(s):=\sum_{k\in\N_0} \ifct{s_k\leq s<s_{k+1}}\d_{v_k}	
		\end{align}
		denote point measures having support on the macroscopic traits. Then, under the assumptions of Theorem \ref{thm:conv_beta}, there exists an $\ve_0>0$ such that, for all $0<\ve<\ve_0$ and all $p\in[1,\infty)$, the following convergence holds in probability, with respect to the $L^p([0,T],\mathcal{M}(V))$ norm
		\begin{align}
			\left(\nu^K_\ve(s\ln K)\right)_{s\in [0,T]}\overset{K\to\infty}{\longrightarrow} 	\left(\nu(s)\right)_{s\in [0,T]},
		\end{align}
		where $\mathcal{M}(V)$ denotes the set of finite, non-negative point measures on $V$ equipped with the weak topology.
	\end{corollary}


\section{Heuristics and discussion}
\label{sec:3_heuristics_discussion}

In this chapter, we give a heuristic idea of the proof strategy for the main Theorem \ref{thm:conv_beta} and discuss the necessity of some of the assumptions that are made. Moreover, we present some specific examples for fitness landscapes that do not satisfy these assumptions, which can still be treated with similar techniques. 

\subsection{Heuristics of the proof of Theorem \ref{thm:conv_beta}}
	As it is usually the case for adaptive dynamics models, the analysis of the limiting dynamics is split into approximations for the resident and the mutant populations.

	First, in Section \ref{sec:4.1_proof_stability}, we prove that - as long as the mutant populations stay below a certain small $\eps K$-threshold - the resident population also only deviates from its equilibrium state by an amount of order $\eps K$. In previous papers, this is often done by applying large deviation results that guarantee for the stochastic process to stay close to an attractive equilibrium for an exponential time in $K$. In our case, to bound the probability of failure (i.e.\ deviating too far from the equilibrium), we need to concatenate these results for an order of $\ln K/\l_K$ many $\lambda_K$-phases that are necessary to observe mutant growth on the $\ln K$-time-scale.
	By conditioning on not deviating too much during the previous phases, we can write the overall probability of failure as the sum of the probabilities to deviate during specific phases. We hence need the latter probabilities to converge to 0 faster than $\l_K/\ln K$. 
	In previous works, this probability of exit from a domain was bounded through a large deviation principle that ensures a vanishing probability of deviating within an exponential time as $K\to\infty$, but does not specify the exact speed of convergence (see e.g.\ \cite{Cha06}). In the present paper, we instead apply a potential theoretic approach similar to \cite{BaBoCh17} to study the embedded discrete time Markov chain and bound the probability of deviation during a $\lambda_K$-phase in $o(\lambda_K/\ln K)$. We combine this with a revised version of the standard convergence result to the deterministic system of \cite[Ch.\ 11.2]{EtKu86} to address the short time spans of order 1 at the beginning of each phase, where the resident population attains its new equilibrium. We prove convergence in probability instead of almost surely but can again quantify the convergence speed and bound the probability of failure in $o(\lambda_K/\ln K)$ in return. Overall, concatenating these two results, which are derived in Appendix \ref{app:A_attraction_equil}, for $O(\ln K/\lambda_K)$ many phases yields a vanishing probability for the resident population to stray from its respective equilibria.
	
	With these bounds on the resident population, in Section \ref{sec:4.2_convergence_beta} we can couple the mutant populations to simpler birth death processes (with immigration) to estimate their growth. In \cite{ChMeTr19} we find a collection of general results on the growth of birth death processes (with immigration) which were formerly used to study similar coupling processes. These results however only cover processes with constant parameters. In Appendix \ref{app:B_branching_processes}, we argue that we can essentially work with the time-average fitness $f^\av_{w,v}$ as the mutants growth rate since $\lambda_K\ll\ln K$, i.e.\ the parameter fluctuations occur on a faster time-scale than the growth of the mutants. This requires a careful rerun of the proofs in \cite{ChMeTr19} to keep track of the error stemming from this averaging approximation.
	
	Finally, based on the results on the coupling processes, we can derive the piecewise affine growth of the orders of population sizes $\beta^K_w$ of the mutant populations as in \eqref{eq:def_betabar}. The equations for $\bar{\b}_w$ combine the growth of a mutant at the rate of its own fitness $f^\av_{w,v}$ with the growth due to incoming mutants from other traits $u$ (which themselves grow at least at rate $f^\av_{u,v}$).

\subsection{Discussion of assumptions}
In the following, we address the necessity of some of our assumptions and discuss possible extensions to more general cases.

\subsubsection{Mutation kernel and probability}
	In Section \ref{sec:2.1_model} we choose the mutation kernel (or law of the trait of a mutant) $(m_{v,w})_{v,w\in V}$ to be independent of the phases $i=1,...,\ell$ and the carrying capacity $K$. Moreover, in Section \ref{sec:2.2_quantities} we choose the probability of mutation at birth as
\begin{align}
	\mu_K=K^{-1/\a},\qquad\a\in\R_{>0}\setminus\N,
\end{align}
which is independent of the phases and traits and depends on $K$ in a very specific way.  Both of these assumptions are not necessary and purely made to simplify notation. 

	The important part is to ensure that, for each $v,w\in V$, $m^{(i,K)}_{v,w}>0$ either during all or during none of the phases $i$ (i.e.\ the mutation graph $G$ does not depend on the phase), that $\mu^{(i,v)}_K>0$ for all phases and traits and that the (additional) dependence on $K$ does not influence the order of the population sizes. Overall, we can allow for dependences of the form
\begin{align}
	m^{(i,K)}_{v,w}\mu^{(i,v)}_K = h(v,w,i,K)K^{-1/\a},
\end{align}
where, for each $(v,w)\in E$,
\begin{align}
	h(v,w,i,K)>0 \text{ and } \abs{\frac{\ln h(v,w,i,K)}{\ln K}}=o(1)
\end{align}
and, for each $(v,w)\notin E$, $h(v,w,i,K)\equiv 0$.

	Under these assumptions, the mutation kernel and probability only contribute a (varying) multiplicative lower order constant to the mutant population sizes (beyond the $K^{-1/\a}$) and do not affect the traits' fitnesses. As a consequence, neither the asymptotic growth of the order of the population size $\beta^K$, which determines the next invading trait, nor the outcome of the invasion according to the Lotka-Volterra dynamics are affected. Therefore, the limiting processes $\bar{\beta}$ and $\nu$ would remain unchanged.
	
The technical assumption of $\alpha\notin\N$ allows for the dichotomy that either $K\mu_K^{d(v,w)}\gg 1$, if $d(v,w)<\a$, or $K\mu_K^{d(v,w)}\ll 1$, if $d(v,w)>\a$. Hence, one can decide whether a $w$-population gets founded by mutation or not. However, we think that the critical case where $w$-mutants arrive at a rate of order $K^0=1$ and may go extinct due to stochastic fluctuations can be handled as well. On any diverging time-scale ($\ll\lK$) we see an infinite number of mutations and corresponding attempts to survive and fixate in the population, where survival is also decided on the same short time-scale. As a result, one can heuristically expect fixation with probability converging to one. This type of argument is part of ongoing research. Its details are quite involved and diverge too far from the core of this paper. Therefore, we exclude these cases in our present results.

\subsubsection{Monomorphic resident population}
	Through termination criterion (d), we ensure that $f^i_{v_{k-1},v_k}<0$, for all phases $i=1,...,\ell$. This is a sufficient criterion to imply a monomorphic equilibrium of trait $v_k$ as the outcome of the Lotka-Volterra dynamics involving the former resident trait $v_{k-1}$ and the newly macroscopic mutant $v_k$. While this criterion is not necessary (as shown in the first example below), we do want to guarantee a monomorphic resident population at all times.

	The reason for this is that our potential theoretic approach to proving good bounds on the resident's population size in Appendix \ref{app:A_attraction_equil} (which is made use of in Section \ref{sec:4.1_proof_stability} to derive the bounding functions $\phi^{(K,\eps,\pm)}$),  relies on estimating the influence of variations in the absolute value of the population size. In the case of a monomorphic population, larger variations can be attributed to either a severe over- or undershoot of the equilibrium population size.  In the case of polymorphy however, variations could stem from either of the resident subpopulations or even a mixture of those, which makes the same estimates no longer useful.

	We expect that these problems are more of a technical nature and an extension to polymorphic resident populations is part of our ongoing research.

\subsection{Examples}
In this section, we present two examples that provide some more insight into the assumptions made to ensure a monomorphic resident population.

\subsubsection{Ensured monomorphism despite temporarily fit resident trait}\label{sec:3_ex1}
	We consider the example of an invasion step (i.e.\ the last Lotka-Volterra like step of an invasion, where an already macroscopic mutant population takes over a former resident population) for $\ell=3$ phases, with resident trait $v$ and mutant trait $w$, where termination criterion (d) is triggered, i.e.\ $f^i_{v,w}\geq0$ for one of the phases $i\in\{1,2,3\}$. We impose the following fitness landscape:
\begin{align}
	f^\av_{w,v}>0,\ f^\av_{v,w}<0,\label{ex1:fitness1}\\
	f^1_{w,v}>0, \ f^1_{v,w}<0;\quad f^2_{w,v}<0,\ f^2_{v,w}<0;\quad f^3_{w,v}<0,\ f^3_{v,w}>0\label{ex1:fitness2}\\
	T_2 f^2_{v,w}+T_3 f^3_{v,w}<0\label{ex1:fitness3}
\end{align}

	The first part of \eqref{ex1:fitness1} ensures that the mutant $w$ population reaches a macroscopic population size of order $K$ in the presence of a resident $v$ population in the first place. Next, the conditions on $f^i_{w,v}$ in \eqref{ex1:fitness2} imply that trait $w$ can only invade the resident population in phase 1 and $f^1_{v,w}<0$ guarantees that a monomorphic equilibrium of trait $w$ is obtained. Moreover, the population of trait $v$ drops to a mesoscopic size (strictly smaller order than $K$) by the end of phase 1. In phase 2, due to the respective negative fitnesses, trait $w$ stays in its monomorphic equilibrium while the population size of $v$ shrinks further with rate $\ f^2_{v,w}<0$. In phase 3, trait $v$ is indeed fit and can grow again (triggering termination criterion (d)). However, \eqref{ex1:fitness3} together with the precise approximations in Appendix \ref{app:C} ensures that this growth does not make up for the decrease in population size during phase 2 and hence $v$ will not reach the $\eps K$-threshold again. Finally, the second part of \eqref{ex1:fitness1} implies that the $v$ population shrinks overall and becomes microscopic on the $\ln K$-time-scale.
	
	To summarise, we have shown that termination criterion (d) is not necessary to guarantee a monomorphic resident population. However, formulating a sharp criterion is much more complicated.

\subsubsection{Different possible outcomes for two-phase cycles}
	As a toy example, to motivate our assumptions/termination criteria for the fitness landscape, we consider all possible behaviours during an invasion step for $\ell=2$ phases with resident trait $v$ and macroscopic mutant $w$. For the latter to be able to reach a macroscopic size, we need $f^\av_{w,v}>0$, which implies that there is at least one phase during which $f^i_{w,v}>0$. Hence, excluding cases of fitness 0, there are seven possible scenarios (up to exchangeability of phases):
\begin{center}
	\begin{tabular}{l|c|c|c|c|c|c|c}
		scenario & 1 & 2 & 3 & 4 & 5 & 6 & 7\\
		\hline
		$f^1_{v,w}$ and $f^2_{v,w}$ & -/- & -/- & +/+ & +/+ & +/- & +/- & +/-\\
		\hline
		$f^1_{w,v}$ and $f^2_{w,v}$ & +/+ & +/- & +/+ & +/- & +/+ & +/- & -/+
	\end{tabular}
\end{center}
The analysis of the different scenarios again makes use of the estimates in Appendix \ref{app:C} and we only present the heuristics here.

	Scenarios 1 and 2 are covered by our results and lead to a new monomorphic resident population of trait $w$.

	Scenario 3 yields a polymorphic resident population of coexisting traits $v$ and $w$. This is because in both phases the respective positive invasion fitnesses imply that there is a unique stable equilibrium point of the two-dimensional Lotka-Volterra system with both components being strictly positive (see e.g.\ \cite[Ch.\ 2.4.3]{Istas05} for a discussion of stable equilibria for two-dimensional Lotka-Volterra systems). 
	
	Scenario 4 does not lead to a fixed resident population. During phase 1, traits $v$ and $w$ obtain a polymorphic coexistence equilibrium as in scenario 3. In phase 2 however, a monomorphic $v$ population is the only stable equilibrium state and the $w$ population starts to shrink again. Since $f^\av_{w,v}>0$ is assumed, we have $T_1f^1_{w,v}+T_2f^2_{w,v}>0$, which implies that $w$ recovers from this decline in the next phase 1 and hence the system keeps switching between a coexistence equilibrium of both traits and a monomorphic equilibrium of trait $v$. (Note that in the beginning of each phase 1, there is a time of order $\lambda_K$ during which $v$ is still the monomorphic resident trait, before $w$ reaches a critical size to trigger the Lotka-Volterra dynamics again).
	
	Scenario 5 is in some sense the flipped scenario 4. However, we do not have information about the sign of $f^\av_{v,w}$. If $f^\av_{v,w}>0$, this is indeed the opposite version and the resident population switches back and forth between a coexistence equilibrium of $v$ and $w$ and a monomorphic equilibrium of trait $w$ only. If $f^\av_{v,w}<0$, then $T_1f^1_{v,w}+T_2f^2_{v,w}<0$. Hence, once the pure $w$ equilibrium is obtained during a phase 1, at the beginning of the next phase 2, $v$ starts to decrease in size. This decline can not be made up by its growth in phase 1 and hence the $v$ population becomes mesoscopic and shrinks on the $\ln K$-time-scale, making $w$ the new monomorphic resident trait. Note that this is an even smaller example for the phenomenon described in Section \ref{sec:3_ex1}, i.e.\ an ensured monomorphic new resident population of the mutant trait despite the former resident trait being fit during some phase. However, it is a little more complicated to describe the exact population sizes here (they depend on whether $w$ first becomes macroscopic during phase 1 or 2, where the previous example guarantees invasion during phase 1). Hence we present both examples and treat this one with less detail.
	
	Scenario 6 gets more complicated. In phase 1, the only stable equilibrium is the polymorphic state involving both $v$ and $w$. In phase 2 however, both monomorphic equilibria of $v$ or $w$ are stable. Hence the dynamics depend on the relation between the coexistence equilibrium state and the regions of attraction for the two monomorphic states. If the former is attracted to the equilibrium of $v$, the $w$ population shrinks during phase 2 but can recover to re-attain the coexistence state in phase 1 (i.e.\ the resident population switches between coexistence and trait $v$ alone). If the coexistence state is attracted to the equilibrium of $w$ in phase 2, the outcome again depends on the fitness $f^\av_{v,w}$ and whether trait $v$ can make up for its decrease in phase 2 by its growth in phase 1, similarly to scenario 5.
	
	Scenario 7 again has two possible outcomes. The Lotka-Volterra dynamics would lead to monomorphic equilibria of $v$ in phase 1 and of $w$ in phase 2.  It depends on the average fitnesses though whether the respective invading traits reach the critical threshold size to trigger these dynamics within these phases. In case of trait $w$, this is guaranteed by $f^\av_{w,v}>0$.  If also  $f^\av_{v,w}>0$, both traits grow faster during their respective fit phases than they shrink during their unfit phases. As a consequence, we observe a switching back and forth between the monomorphic equilibria (not necessarily synced up with the phase changes but delayed by a $\lambda_K$ time, as above). If $f^\av_{v,w}<0$, trait $v$ cannot recover during phase 1 and the system stays in the monomorphic $w$ equilibrium.
	
	Overall, already in this minimal example of two phases, we can observe a variety of different behaviours, ranging from monomorphic equilibrium states to coexistence and switching between those. Which of these is the case not only depends on the invasion fitnesses of the single traits but also on the precise relation between them and the timing of the different phases.


\section{Proofs}
\label{sec:4_proofs}

In this chapter, we conduct the proofs of the main results of this paper, i.e.\ Theorem \ref{thm:conv_beta} and Corollary \ref{cor:seq_residents}. We utilise a number of technical results on birth death processes with self-competition or immigration. To maintain a better readability of the main proofs, these technical results are stated and proved in the appendices.

This chapter is divided into several sections. In Section \ref{sec:4.1_proof_stability}, we discuss the stability of the resident trait during the mutants' growth phase. In Section \ref{sec:4.2_convergence_beta} we prove Theorem \ref{thm:conv_beta}, i.e.\ the convergence of the exponents $\beta^K$. Finally, in Section \ref{sec:4.3_seq_residents}, we conclude the result on the sequence of resident traits of Corollary \ref{cor:seq_residents}.

\subsection{Stability of the resident trait}
\label{sec:4.1_proof_stability}

Since we are working in a regime of periodically changing parameters, we cannot expect the resident population's size to stay close to one fixed value. Instead, the population size is attracted to the respective equilibrium sizes of the different phases. Since the population needs a short time to adapt to the new equilibrium after a change in parameters, we define two functions $\phi_v^{(K,\eps,+)}$ and $\phi_v^{(K,\eps,-)}$ that bound the population size and take into account these short transition phases of length $T_\eps$. We can then prove that, as long as the mutant populations stay small and as $K\to\infty$, the resident's population size stays between these bounding functions for a time of order $\ln K$ with high probability.

We define
\begin{align}
	\begin{aligned}
	\phi_v^{(K,\eps,+)}(t)=\begin{cases}
	\max\{\bar{n}^{i-1}_v,\bar{n}^i_v\}+M\eps&\text{, if }t\in(T^\Sigma_{i-1}\lambda_K,T^\Sigma_{i-1}\lambda_K+T_\eps)\\
	\bar{n}^i_v+M\eps&\text{, if }t\in[T^\Sigma_{i-1}\lambda_K+T_\eps,T^\Sigma_i\lambda_K]
	\end{cases}\\
	\phi_v^{(K,\eps,-)}(t)=\begin{cases}
	\min\{\bar{n}^{i-1}_v,\bar{n}^i_v\}-M\eps&\text{, if }t\in(T^\Sigma_{i-1}\lambda_K,T^\Sigma_{i-1}\lambda_K+T_\eps)\\
	\bar{n}^i_v-M\eps&\text{, if }t\in[T^\Sigma_{i-1}\lambda_K+T_\eps,T^\Sigma_i\lambda_K]
	\end{cases}
	\end{aligned}
\end{align}
with periodic extension, where $\bar{n}^{i-1}_v:=\bar{n}_v^\ell$ for $i=1$, and $\phi_v^{(K,\eps,\pm)}(0)=\bar{n}^\ell_v\pm M\eps$. Note that these functions also depend on the choices of $M$ and $T_\eps$. To simplify notation, we however do not include those parameters in the functions' names.
To mark the time at which the mutant populations become too large and start to significantly perturb the system, we introduce the stopping time
\begin{align}
S^{(K,\eps)}_v:=\inf\left\{t\geq0:\sum_{w\neq v}N^{K}_w(t)\geq \eps K\right\}.
\end{align}

With this notation, the resident's stability result can be stated as follows.

\begin{theorem}\label{Thm_phi}
There exists a uniform $M<\infty$ and, for all $\eps>0$ small enough, there exists a deterministic $T_\eps<\infty$ such that, for all traits $v\in V$ such that $b^i_v>d^i_v$, $1\leq i\leq \ell$, and for all $T<\infty$,
\begin{align}
\lim_{K\to\infty}\P\Bigg(\exists\ t\in[0,T\ln K\land S^{(K,\eps)}_v]:\frac{N_v^{K}(t)}{K}\notin[\phi_v^{(K,\eps,-)}(t),\phi_v^{(K,\eps,+)}(t)]\bigg|&\nonumber\\
\frac{N_v^{K}(0)}{K}\in[\phi_v^{(K,\eps,-)}(0)+\eps,\phi_v^{(K,\eps,+)}(0)-\eps]&\Bigg)=0.
\end{align}
\end{theorem}

\begin{proof}
The proof is based on couplings with bounding single-trait birth death processes with self-competition.  We proceed in several steps:
\begin{itemize}
\item[1)] For a fixed $i$ phase, prove that $N^{K}_v/K$ gets $\eps$-close to the new equilibrium $\bar{n}^i_v$ within a finite time $T_\eps$ and stays bounded until then.
\item[2)] For a fixed $i$ phase, prove that $N^{K}_v/K$ stays $\eps$-close to its equilibrium $\bar{n}^i_v$ after $T_\eps$ until the end of the phase.
\item[3)] Use the strong Markov property to concatenate multiple phases to obtain a result for $\ln K$ times.
\end{itemize}
Note that in the following we conduct the proof for a fixed resident trait $v\in V$.  Uniform values for $M$ and $T_\eps$ can be obtained by taking the maximum over all such traits since we work with a finite trait space. In steps 1 and 2, we prove that the desired bounds fail with a probability in $o(\lambda_K/\ln K)$. This allows us to concatenate $O(\ln K/\lambda_K)$ phases for an overall time horizon of order $\ln K$ in step 3.

\textbf{Step 1 (attaining the equilibrium):} We fix a phase $1\leq i\leq\ell$ and, without loss of generality, assume that an $i$ phase starts at time $t=0$ and lasts until $t=T_i\lambda_K$ (we will ``reset'' time with the help of the Markov property in step 3). We start by showing that there are constants $\underline{C}^i,\overline{C}^i<\infty$ such that, for any $\eps>0$ and interval $I=[a_1,a_2]\subset (0,\infty)$, there is a deterministic time $T^{I,i}_\eps<\infty$ such that
\begin{align}\label{eq:resident_reequil}
\P\Bigg(&\exists\ t\in[0,T^{I,i}_\eps\land S^{(K,\eps)}_v]: \frac{N^{K}_v(t)}{K}\notin[((\bar{n}^i_v-\eps \underline{C}^i)\land a_1)-\eps,((\bar{n}^i_v+\eps\overline{C}^i)\lor a_2)+\eps] \nonumber\\
&\text{or }T^{I,i}_\eps\leq S^{(K,\eps)}_v\ \&\ \frac{N^{K}_v(T^{I,i}_\eps)}{K}\notin[\bar{n}^i_v-\eps \underline{C}^i-2\eps,\bar{n}^i_v+\eps\overline{C}^i+2\eps]\bigg|\frac{N^{K}_v(0)}{K}\in[a_1,a_2]\Bigg)\nonumber\\
&=o\left(\frac{\lambda_K}{\ln K}\right)\text{ as }K\to\infty.
\end{align}

	While these bounds seem quite unintuitive, they come up naturally by first comparing the actual process $N^{K}_v$ to two branching processes and then comparing these to their deterministic equivalent. These approximations are discussed in detail below and are visualised in Figure \ref{fig:approxLV}. The first line of \ref{eq:resident_reequil} corresponds to a worst case bound of the population up to the time $T^{I,i}_\eps$ when the new equilibrium is (almost) obtained. The second line corresponds to (almost) reaching the new equilibrium at time $T^{I,i}_\eps$ itself.

\begin{figure}
	\includegraphics[width=\textwidth]{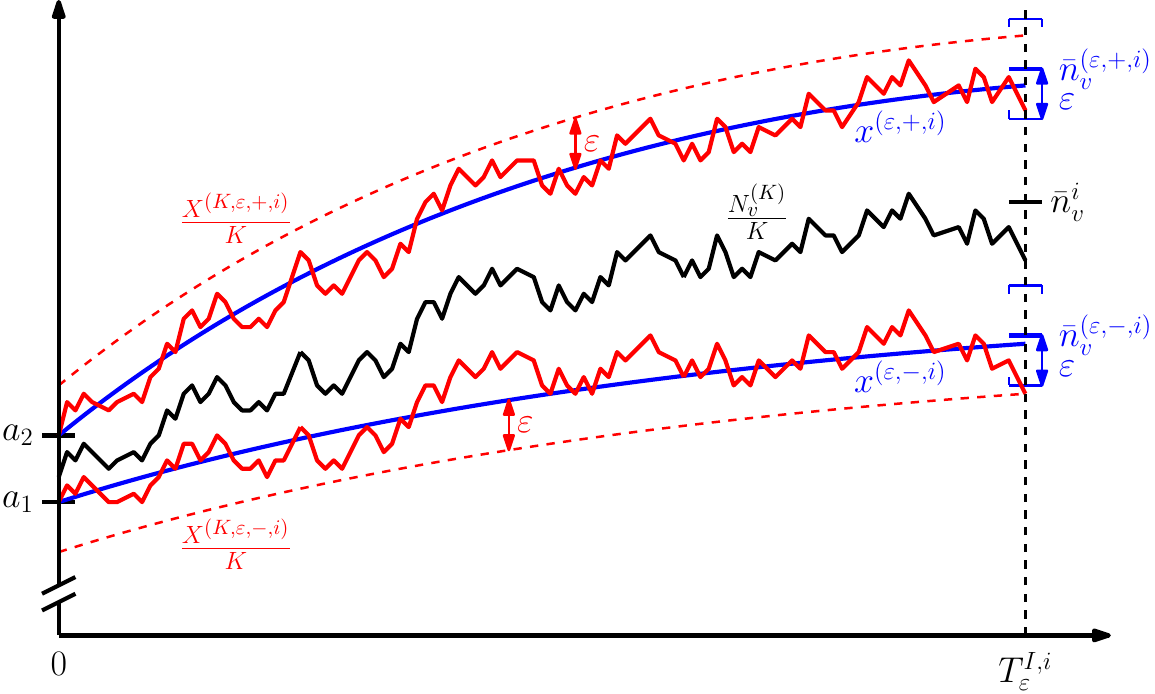}
	\caption{Two steps of approximation: Original process $N^{K}_v/K$ and new equilibrium $\bar{n}^i_v$ in black. Bounding birth death processes with self-competition $X^{(K,\eps,-,i)}/K$ and $X^{(K,\eps,+,i)}/K$ in red. Limiting deterministic solutions $x^{(\eps,-,i)}$ and $x^{(\eps,+,i)}$ with respective perturbed equilibrium sizes $\bar{n}^{(\eps,-,i)}_v=\bar{n}^i_v-\eps\underline{C}^i$ and $\bar{n}^{(\eps,+,i)}_v=\bar{n}^i_v+\eps\overline{C}^i$ in blue.}\label{fig:approxLV}
\end{figure}

We want to apply the results from Appendix \ref{app:A_attraction_equil}. To do so, we couple the process $N^{K}_v$ to two single-trait birth death processes with self-competition $X^{(K,\eps,-,i)}$ and $X^{(K,\eps,+,i)}$ such that
\begin{align}
X^{(K,\eps,-,i)}(t)\leq N^{K}_v(t)\leq X^{(K,\eps,+,i)}(t),\quad \forall\ t\in[0,S^{(K,\eps)}_v\land \tau^{(K,\eps)}_v \land T_i\lambda_K],
\end{align}
where
\begin{align}
\tau^{(K,\eps)}_v:=\inf\left\{t\geq0:N^{K}_v(t)< \eps K\right\}
\end{align}
ensures that the resident does not become too small and incoming mutants can hence be approximated by additional clonal births. We let $K$ be large enough such that $\mu_K<\eps$.

$X^{(K,\eps,-,i)}$ assumes the lowest starting value, maximal competition from other traits, maximal loss due to outgoing mutants and no incoming mutants. It hence has
\begin{itemize}
\item initial condition $X^{(K,\eps,-,i)}(0)=\lfloor a_1K\rfloor$,
\item birth rate $b^i_v(1-\eps)$, 
\item death rate $d^i_v+\eps\max_{w\neq v}c^i_{vw}$ and
\item self-competition rate $c^i_{vv}/K$.
\end{itemize}

$X^{(K,\eps,+,i)}$ assumes the highest starting value, no competition from other traits, no outgoing mutation and maximal incoming mutation. It hence has
\begin{itemize}
\item initial condition $X^{(K,\eps,+,i)}(0)=\lceil a_2K\rceil$,
\item birth rate $b^i_v+\eps\max_{w\neq v}b^i_w$,
\item death rate $d^i_v$ and
\item self-competition rate $c^i_{vv}/K$.
\end{itemize}

These couplings can be explicitly constructed via Poisson process representations, see e.g.\ \cite[Ch.\ 7.2]{BoCoSm19}.

By Theorem \ref{EthKu_variation}, on finite time intervals, the rescaled processes $X^{(K,\eps,-,i)}/K$ and $X^{(K,\eps,+,i)}/K$ converge uniformly on finite time intervals to the solutions $(x^{(\eps,-,i)}(t))_{t\geq 0}$ and  $(x^{(\eps,+,i)}(t))_{t\geq 0}$ of the ordinary differential equations
\begin{align}
	\begin{aligned}
		\dot{x}^{(\eps,-,i)}(t)&=x^{(\eps,-,i)}(t)\left(b^i_v(1-\eps)-(d^i_v+\eps\max_{w\neq v}c^i_{vw})-c^i_{vv}x^{(\eps,-,i)}(t)\right)&,\ x^{(\eps,-,i)}(0)=a_1,\\
		\dot{x}^{(\eps,+,i)}(t)&=x^{(\eps,+,i)}(t)\left(b^i_v+\eps\max_{w\neq v}b^i_w-d^i_v-c^i_{vv}x^{(\eps,+,i)}(t)\right)&,\ x^{(\eps,+,i)}(0)=a_2.
	\end{aligned}
\end{align}
These equations have unique attractive equilibrium points
\begin{align}
	\begin{aligned}
		\bar{n}^{(\eps,-,i)}_v&:=\frac{b^i_v(1-\eps)-(d^i_v+\eps\max_{w\neq v}c^i_{vw})}{c^i_{vv}}=\bar{n}^i_v-\eps\frac{b^i_v+\max_{w\neq v}c^i_{vw}}{c^i_{vv}}=:\bar{n}^i_v-\eps \underline{C}^i,\\
		\bar{n}^{(\eps,+,i)}_v&:=\frac{b^i_v+\eps\max_{w\neq v}b^i_w-d^i_v}{c^i_{vv}}=\bar{n}^i_v+\eps\frac{\max_{w\neq v}b^i_w}{c^i_{vv}}=:\bar{n}^i_v+\eps\overline{C}^i,
	\end{aligned}
\end{align}
which their solutions will attain up to an $\eps$ within a finite time $T^{I,i}_\eps<\infty$, i.e.
\begin{align}
\left|x^{(\eps,-,i)}(T^{I,i}_\eps)-\bar{n}^{(\eps,-,i)}\right|\leq\eps\text{ and } \left|x^{(\eps,+,i)}(T^{I,i}_\eps)-\bar{n}^{(\eps,+,i)}\right|\leq\eps.
\end{align}
Moreover, due to the monotonicity of the solutions, for all $t\in[0,T^{I,i}_\eps]$, 
\begin{align}
x^{(\eps,-,i)}(t)\geq a_1\land \bar{n}^{(\eps,-,i)}\text{ and } x^{(\eps,+,i)}(t)\leq a_2\lor \bar{n}^{(\eps,+,i)}.
\end{align}

We can hence bound
\begin{align}
&\ \P\Bigg(\exists\ t\in[0,T^{I,i}_\eps\land S^{(K,\eps)}_v\land \tau^{(K,\eps)}_v]:\frac{N^{K}_v(t)}{K}\notin[((\bar{n}^i_v-\eps \underline{C}^i)\land a_1)-\eps,((\bar{n}^i_v+\eps\overline{C}^i)\lor a_2)+\eps]\nonumber\\
&\qquad\ \text{or }T^{I,i}_\eps\leq S^{(K,\eps)}_v\land \tau^{(K,\eps)}_v\ \&\  \frac{N^{K}_v(T^{I,i}_\eps)}{K}\notin[\bar{n}^i_v-\eps \underline{C}^i-2\eps,\bar{n}^i_v+\eps\overline{C}^i+2\eps]\bigg|\frac{N^{K}_v(0)}{K}\in[a_1,a_2]\Bigg)\nonumber\\
\leq&\ \P\Bigg(\exists\ t\in[0,T^{I,i}_\eps]:\frac{X^{(K,\eps,-,i)}(t)}{K}<((\bar{n}^i_v-\eps \underline{C}^i)\land a_1)-\eps\text{ or }\frac{X^{(K,\eps,+,i)}(t)}{K}>((\bar{n}^i_v+\eps\overline{C}^i)\lor a_2)+\eps\nonumber\\
&\qquad\ \text{or }\frac{X^{(K,\eps,-,i)}(T^{I,i}_\eps)}{K}<\bar{n}^i_v-\eps \underline{C}^i-2\eps\text{ or }\frac{X^{(K,\eps,+,i)}(T^{I,i}_\eps)}{K}>\bar{n}^i_v+\eps\overline{C}^i+2\eps\Bigg)\nonumber\\
\leq&\ \P\Bigg(\exists\ t\in[0,T^{I,i}_\eps]:\left|\frac{X^{(K,\eps,-,i)}(t)}{K}-x^{(\eps,-,i)}(t)\right|>\eps\text{ or }\left|\frac{X^{(K,\eps,+,i)}(t)}{K}-x^{(\eps,+,i)}(t)\right|>\eps\Bigg)\nonumber\\
\leq&\ \P\left(\sup_{t\leq T^{I,i}_\eps}\left|\frac{X^{(K,\eps,-,i)}(t)}{K}-x^{(\eps,-,i)}(t)\right|>\eps\right)+ \P\left(\sup_{t\leq T^{I,i}_\eps}\left|\frac{X^{(K,\eps,+,i)}(t)}{K}-x^{(\eps,+,i)}(t)\right|>\eps\right)\nonumber\\
=&\ o\left(\frac{\lambda_K}{\ln K}\right),
\end{align}
where we apply Theorem \ref{EthKu_variation} in the last step.

Finally, we note that, if we choose $\eps$ small enough such that $\eps<((\bar{n}^i_v-\eps \underline{C}^i)\land a_1)-2\eps$, then $\tau^{(K,\eps)}_v<T^{I,i}_\eps\land S^{(K,\eps)}_v$ implies that $N^{K}_v(t)/K$ must have left the above intervals prior to this time and hence we can drop the stopping time $\tau^{(K,\eps)}_v$ from the probability on the left hand side.

\textbf{Step 2 (stability of the equilibrium):} We still study a specific $i$ phase from time $t=0$ up to $t=T_i\lambda_K$ and now consider the time span $[T^{I,i}_\eps,T_i\lambda_K]$. Since this is a time span of divergent length, we can no longer apply Theorem \ref{EthKu_variation} and the convergence to the deterministic system. Instead, we apply Theorem \ref{stability_lnK} on the stability of equilibrium points for $\ln K$ times to derive
\begin{align}
\lim_{K\to\infty}\frac{\ln K}{\lambda_K}\cdot\P\Bigg(\exists\ t\in[T^{I,i}_\eps,T_i\lambda_K\land S^{(K,\eps)}_v]:\frac{N^{K}_v(t)}{K}\notin[\bar{n}^i_v-\eps \underline{C}^i-16\eps,\bar{n}^i_v+\eps\overline{C}^i+16\eps]\bigg|\nonumber\\
\frac{N^{K}_v(T^{I,i}_\eps)}{K}\in[\bar{n}^i_v-\eps \underline{C}^i-2\eps,\bar{n}^i_v+\eps\overline{C}^i+2\eps]\Bigg)=0
\end{align}

We utilise the coupling processes $X^{(K,\eps,-,i)}$ and $X^{(K,\eps,+,i)}$ as in step 1, with the same birth, death and self-competition rates but this time with initial conditions
\begin{align}
X^{(K,\eps,-,i)}(T^{I,i}_\eps)=\lfloor(\bar{n}^i_v-\eps(\underline{C}^i+2)) K\rfloor\quad\text{and}\quad X^{(K,\eps,+,i)}(T^{I,i}_\eps)=\lceil(\bar{n}^i_v+\eps(\overline{C}^i+2)) K\rceil.
\end{align}
Then, if $N^{K}_v(T^{I,i}_\eps)/K\in[\bar{n}^i_v-\eps (\underline{C}^i+2),\bar{n}^i_v+\eps(\overline{C}^i+2)]$,
\begin{align}
X^{(K,\eps,-,i)}(t)\leq N^{K}_v(t)\leq X^{(K,\eps,+,i)}(t),\quad \forall\ t\in[T^{I,i}_\eps,S^{(K,\eps)}_v\land \tau^{(K,\eps)}_v \land T_i\lambda_K].
\end{align}

We apply Theorem \ref{stability_lnK} with $\eps'=16\eps$, and hence
\begin{align}
\left|X^{(K,\eps,\pm,i)}_v(T^{I,i}_\eps)-\bar{n}_v^{(\eps,\pm,i)}K\right|\leq2\eps K=\frac{\eps'K}{8}<\frac{1}{2}\left\lfloor\frac{\eps'K}{2}\right\rfloor,
\end{align}
to obtain
\begin{align}
&\ \P\Bigg(\exists\ t\in[T^{I,i}_\eps,T_i\lambda_K\land S^{(K,\eps)}_v\land\tau^{(K,\eps)}_v]:\frac{N^{K}_v(t)}{K}\notin[\bar{n}^i_v-\eps \underline{C}^i-16\eps,\bar{n}^i_v+\eps\overline{C}^i+16\eps]\bigg|\nonumber\\
&\qquad\ \frac{N^{K}_v(T^{I,i}_\eps)}{K}\in[\bar{n}^i_v-\eps \underline{C}^i-2\eps,\bar{n}^i_v+\eps\overline{C}^i+2\eps]\Bigg)\nonumber\\
\leq &\ \P\Bigg(\exists\ t\in[T^{I,i}_\eps,T_i\lambda_K]:\frac{X^{(K,\eps,-,i)}(t)}{K}<\bar{n}^i_v-\eps \underline{C}^i-16\eps\quad\text{or}\quad\frac{X^{(K,\eps,+,i)}(t)}{K}>\bar{n}^i_v+\eps\overline{C}^i+16\eps\Bigg)\nonumber\\
\leq&\ \P\Bigg(\exists\ t\in[T^{I,i}_\eps,T_i\lambda_K]:\left|\frac{X^{(K,\eps,-,i)}(t)}{K}-\bar{n}^{(\eps,-,i)}\right|>\eps'\quad\text{or}\quad\left|\frac{X^{(K,\eps,+,i)}(t)}{K}-\bar{n}^{(\eps,+,i)}_\eps\right|>\eps'\Bigg)\nonumber\\
\leq&\ \P\left(\sup_{t\in[T^{I,i}_\eps,T_i\lambda_K]}\left|X^{(K,\eps,-,i)}(t)-\bar{n}^{(\eps,-,i)}K\right|>\eps' K\right)+ \P\left(\sup_{t\in[T^{I,i}_\eps,T_i\lambda_K]}\left|X^{(K,\eps,+,i)}(t)-\bar{n}^{(\eps,+,i)}K\right|>\eps' K\right)\nonumber\\
=&\ o\left(\frac{\lambda_K}{\ln K}\right).
\end{align}

As in step 1, for sufficiently small $\eps$, we can again drop the stopping time $\tau^{(K,\eps)}_v$ from the probability on the left hand side.

\textbf{Step 3 (concatenating multiple phases):} We first piece together steps 1 and 2 to obtain a result for an entire $i$ phase and then concatenate multiple phases to prove the final result of the Theorem.

Applying the Markov property (at $T^{I,i}_\eps$) in the first step, we obtain
\begin{align}\label{iphase}
&\ \P\Bigg(\exists\ t\in[0,T^{I,i}_\eps\land S^{(K,\eps)}_v]:\frac{N^{K}_v(t)}{K}\notin[((\bar{n}^i_v-\eps \underline{C}^i)\land a_1)-\eps,((\bar{n}^i_v+\eps\overline{C}^i)\lor a_2)+\eps]\notag\\
&\qquad \text{ or }\exists\ t\in[T^{I,i}_\eps\land S^{(K,\eps)}_v,T_i\lambda_K\land S^{(K,\eps)}_v]:\frac{N^{K}_v(t)}{K}\notin[\bar{n}^i_v-\eps \underline{C}^i-16\eps,\bar{n}^i_v+\eps\overline{C}^i+16\eps]\bigg|\notag\\
&\qquad \frac{N^{K}_v(0)}{K}\in[a_1,a_2]\Bigg)\notag\\
\leq&\ \P\Bigg(\exists\ t\in[0,T^{I,i}_\eps\land S^{(K,\eps)}_v]:\frac{N^{K}_v(t)}{K}\notin[((\bar{n}^i_v-\eps \underline{C}^i)\land a_1)-\eps,((\bar{n}^i_v+\eps\overline{C}^i)\lor a_2)+\eps]\notag\\
&\qquad \text{or }T^{I,i}_\eps\leq S^{(K,\eps)}_v\ \&\ \frac{N^{K}_v(T^{I,i}_\eps)}{K}\notin[\bar{n}^i_v-\eps \underline{C}^i-2\eps,\bar{n}^i_v+\eps\overline{C}^i+2\eps]\bigg|\frac{N^{K}_v(0)}{K}\in[a_1,a_2]\Bigg)\notag\\
&\ + \P\Bigg(\exists\ t\in[T^{I,i}_\eps,T_i\lambda_K\land S^{(K,\eps)}_v]:\frac{N^{K}_v(t)}{K}\notin[\bar{n}^i_v-\eps \underline{C}^i-16\eps,\bar{n}^i_v+\eps\overline{C}^i+16\eps]\bigg|\notag\\
&\qquad \frac{N^{K}_v(T^{I,i}_\eps)}{K}\in[\bar{n}^i_v-\eps \underline{C}^i-2\eps,\bar{n}^i_v+\eps\overline{C}^i+2\eps]\Bigg)\notag\\
=&\ o\left(\frac{\lambda_K}{\ln K}\right).
\end{align}

Here we impose stronger bounds in the first time period up to $T^{(I,i)}_\eps$ to ensure good initial conditions for the remaining diverging time.

Note that these probabilities are in $o(\lambda_K/\ln K)$ uniformly in $1\leq i\leq\ell$.

Now we can finally link together multiple phases. For ease of notation, we index the phases by $i\in\N$ instead of $1\leq i\leq\ell$, where every $(k\ell+i)^\textbf{th}$ phase, $k\in\N$, is of type $i$ and length $T_i\lambda_K$. Similarly, we extend the definitions of $T^\Sigma_i$, $\bar{n}^i_v$ and $T^{I,i}_\eps$.

Choosing $T_\eps=\max_{1\leq i\leq\ell}T^{I,i}_\eps$ and $M=\max_{1\leq i\leq\ell}(\underline{C}^i\lor\overline{C}^i)+17$ in the definition of $\phi_v^{(K,\eps,\pm)}$, and the intervals $I=[a^1_1,a^1_2]=[\bar{n}^{\ell}_v-\eps(M-1),\bar{n}^{\ell}_v+\eps(M-1)]$ as well as $I=[a^i_1,a^i_2]=[\bar{n}^{i-1}_v-\eps (\underline{C}^{i-1}+16),\bar{n}^{i-1}_v+\eps(\overline{C}^{i-1}+16)]$, $i\geq2$, in \eqref{iphase}, we deduce the convergence for any $T<\infty$. See Figure \ref{fig:concatenate} for a visualisation of the concatenation of two phases.

\begin{align}
&\ \P\Bigg(\exists\ t\in[0,T\ln K\land S^{(K,\eps)}_v]:\frac{N_v^K(t)}{K}\notin[\phi_v^{(K,\eps,-)}(t),\phi_v^{(K,\eps,+)}(t)]\bigg|\nonumber\\
&\ \frac{N_v^{K}(0)}{K}\in[\phi_v^{(K,\eps,-)}(0)+\eps,\phi_v^{(K,\eps,+)}(0)-\eps]\Bigg)\nonumber\\
= &\ \P\Bigg(\exists\ i\in\N:T^\Sigma_{i-1}\lambda_K\leq T\ln K\land S^{(K,\eps)}_v \text{ and}\nonumber\\
&\ \exists\ t\in(T^\Sigma_{i-1}\lambda_K,(T^\Sigma_{i-1}\lambda_K+T^{I,i}_\eps)\land S^{(K,\eps)}_v): \frac{N^{K}_v(t)}{K}\notin[(\bar{n}^{i-1}_v\land\bar{n}^i_v)-M\eps,(\bar{n}^{i-1}_v\lor\bar{n}^i_v)+M\eps]\nonumber\\
&\ \text{ or }\exists\ t\in[(T^\Sigma_{i-1}\lambda_K+T^{I,i}_\eps)\land S^{(K,\eps)}_v,T^\Sigma_i\lambda_K\land S^{(K,\eps)}_v]:\frac{N^{K}_v(t)}{K}\notin[\bar{n}^i_v-M\eps,\bar{n}^i_v+M\eps]\bigg|\nonumber\\
&\ \frac{N_v^{K}(0)}{K}\in[\phi^{(K,\eps,-)}(0)+\eps,\phi^{(K,\eps,+)}(0)-\eps]\Bigg)\nonumber \\
\leq &\ \P\Bigg(\exists\ i\in\N:T^\Sigma_{i-1}\lambda_K\leq T\ln K\land S^{(K,\eps)}_v \text{ and}\nonumber\\
&\ \exists\, t \! \in[T^\Sigma_{i-1}\lambda_K,(T^\Sigma_{i-1}\lambda_K+T^{I,i}_\eps)\land S^{(K,\eps)}_v]:\! \frac{N^{K}_v(t)}{K}\! \notin\! [((\bar{n}^i_v-\eps \underline{C}^i)\land a^i_1)-\eps,((\bar{n}^i_v+\eps\overline{C}^i)\lor a^i_2)+\eps]\nonumber\\
&\ \text{ or }\exists\, t\! \in\! [(T^\Sigma_{i-1}\lambda_K+T^{I,i}_\eps)\land S^{(K,\eps)}_v,T^\Sigma_i\lambda_K\land S^{(K,\eps)}_v]:\! \frac{N^{K}_v(t)}{K}\!\notin[\bar{n}^i_v-\eps \underline{C}^i-16\eps,\bar{n}^i_v+\eps\overline{C}^i+16\eps]\bigg|\nonumber\\
&\ \frac{N_v^{K}(0)}{K}\in[\phi^{(K,\eps,-)}(0)+\eps,\phi^{(K,\eps,+)}(0)-\eps]\Bigg)\nonumber\\
\leq &\ \sum_{\substack{i\in\N:\nonumber\\T^\Sigma_{i-1}\lambda_K<T\ln K}}
\P\Bigg(\exists\ t\in[T^\Sigma_{i-1}\lambda_K,(T^\Sigma_{i-1}\lambda_K+T^{I,i}_\eps)\land S^{(K,\eps)}_v]:\nonumber\\
&\ \frac{N^{K}_v(t)}{K}\notin[((\bar{n}^i_v-\eps \underline{C}^i)\land a^i_1)-\eps,((\bar{n}^i_v+\eps\overline{C}^i)\lor a^i_2)+\eps]\nonumber\\
&\ \text{ or }\exists\ t\in[(T^\Sigma_{i-1}\lambda_K+T^{I,i}_\eps)\land S^{(K,\eps)}_v,T^\Sigma_i\lambda_K\land S^{(K,\eps)}_v]:\frac{N^{K}_v(t)}{K}\notin[\bar{n}^i_v-\eps \underline{C}^i-16\eps,\bar{n}^i_v+\eps\overline{C}^i+16\eps]\bigg|\nonumber\\
&\ \frac{N_v^{K}(T^\Sigma_{i-1}\lK)}{K}\in[a^i_1,a^i_2]\Bigg)\nonumber\\
=&\ o(1),
\end{align}
where we utilise that we have $O(\ln K/\lambda_K)$ summands that are (uniformly) of order $o(\lambda_K/\ln K)$ to conclude.

Note that, in contrast to the second to last expression, in \eqref{iphase} the initial time of the phase is set to 0. This however does not change the probability due to the Markov property and the periodic time-homogeneity of the Markov process.
Letting $K$ tend to infinity, this yields the proof of Theorem \ref{Thm_phi}.
\end{proof}

\begin{figure}[h]
	\includegraphics[width=.9\textwidth]{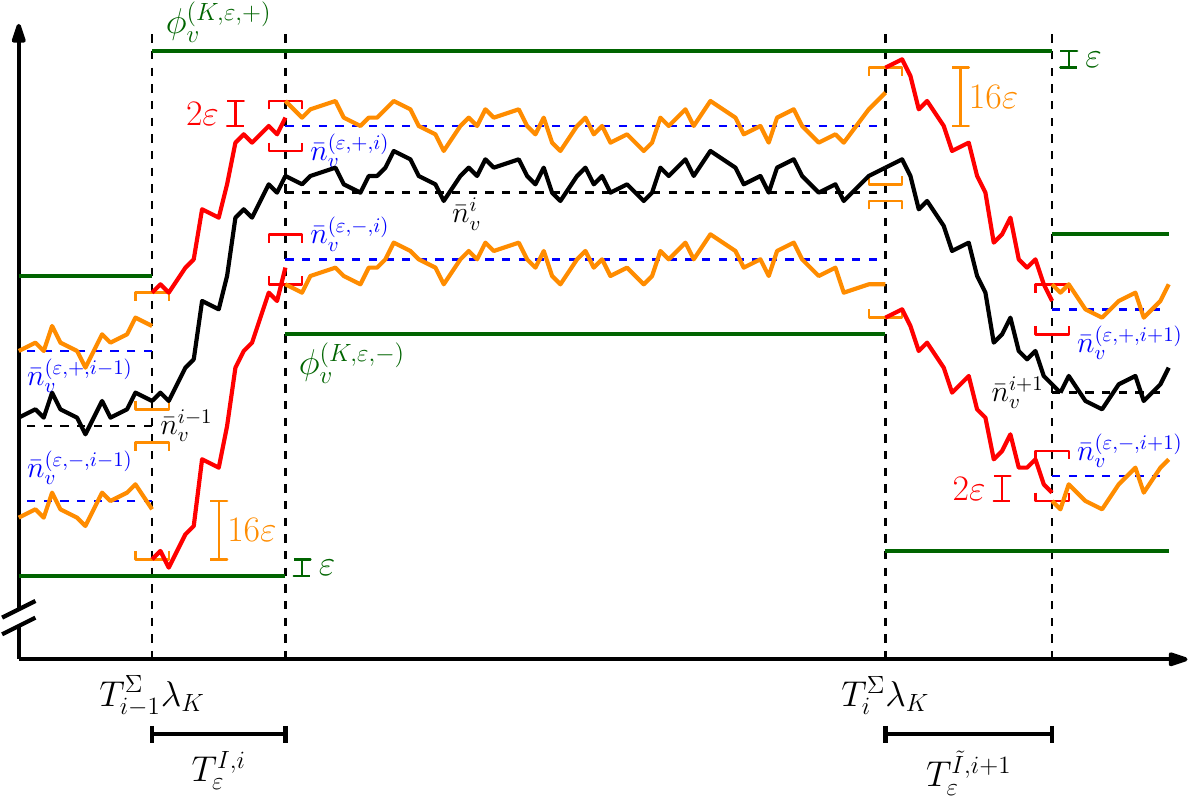}
	\caption{Concatenation of phases $i-1$, $i$ and $i+1$. Original process $N^{K}_v/K$ and corresponding equilibria in black. Bounding birth death processes with self-competition $X^{(K,\eps,-,i)}/K$ and $X^{(K,\eps,+,i)}/K$ in red (fast re-equilibration from step 1) and orange (long stability from step 2). Equilibrium sizes $\bar{n}^{(\eps,-,i)}_v=\bar{n}^i_v-\eps\underline{C}^i$ and $\bar{n}^{(\eps,+,i)}_v=\bar{n}^i_v+\eps\overline{C}^i$ of the corresponding (perturbed) deterministic system in blue. Bounding functions $\phi^{(K,\eps,-)}_v$ and $\phi^{(K,\eps,+)}_v$ in green.}\label{fig:concatenate}
\end{figure}

\subsection{Convergence of the orders of population sizes}
\label{sec:4.2_convergence_beta}
	The proof of Theorem \ref{thm:conv_beta} is based on an induction principle and similar to the proof of the main theorem of \cite{CoKrSm21}. We therefore do not repeat every single detail but point out how to deal with the important difficulties arising from our extended model with time-dependent growth parameters. This is done in five steps:
	\begin{enumerate}[1)]
		\item Define the main stopping times and set up the induction.
		\item Use the convergence of $\b_w^K(0)$ for the base case of the induction.
		\item Couple the process with non-interacting birth death processes to control the growth of the mutant populations.
		\item Ensure that mutants become macroscopic only in a fit phase $i$.
		\item Finish the induction step by comparison to the deterministic Lotka-Volterra system.
	\end{enumerate}
	
	\textbf{Step 1 (preparation):}
	The induction is set up in such a way that each step corresponds to the invasion of a new mutant. We divide these steps into two alternating substeps. During the first one, the resident population is stable in a certain sense and we approximate the growth of the mutant populations on the $\ln K$-time-scale. The second one is started when one of the mutant populations becomes macroscopic and we therefore observe a Lotka-Volterra interaction between the mutant and the former resident population.

	In order to make this distinction into substeps rigorous, we introduce, for $k\in\N_0$, the pair of stopping times (visualised in Figure \ref{fig:StoppingTimes})
	\begin{align}
	\label{eq:pf_stoppingtimes}
		\s^{K}_{k}&:=\inf\dset{t\geq\theta^{K}_{k}: \frac{N^K_{v_{k}}(t)}{K}\in [\phi^{(K,\eps_k,-)}_{v_{k}}(t)+\eps_{k},\phi^{(K,\eps_k,+)}_{v_{k}}(t)-\eps_k] \text{ and } \sum_{w\neq v_{k}}N^K_w(t)< \ve^2_k K},\nonumber\\
		\theta^{K}_{k+1}&:=\inf \dset{t\geq\s^K_{k}: \frac{N^K_{v_{k}}(t)}{K}\notin \left[\phi^{(K,\eps_k,-)}_{v_{k}}(t), \phi^{(K,\eps_k,+)}_{v_{k}}(t)\right] \text{ or } \sum_{w\neq v_{k}}N^K_w(t)\geq \ve_k K},
	\end{align}
	where the $\ve_k>0$ are chosen at the very end, in reverse order. More precisely, to ensure good estimates until the end of our time horizon $[0,T\ln K]$, one has to keep the accumulating error low from the very beginning and choose each $\eps_k$ small enough to provide good initial bounds for the next invasion step.
	
	This means that at time $\s^K_k$ the process has reached the monomorphic Lotka-Volterra-equilibrium of trait $v_k\in V$, and $v_k$ remains the only macroscopic trait until time $\theta^K_{k+1}$. Moreover, its population size lies inside the $\ve$-tunnel described by $\phi^{(K,\eps_k,-)}$ and $\phi^{(K,\eps_k,+)}$ during $[\s^K_k,\theta^K_{k+1}]$.
	
	In step 3, we introduce the stopping times $s^K_{k+1}$, when the first non-resident population becomes \textit{almost macroscopic}, i.e.\ attains a population size of order $K^{1-\eps_{k+1}}$ for some small $\eps_{k+1}>0$ (see Figure \ref{fig:StoppingTimes}), as well as the appearance times $t^K_{w,k+1}$ of new mutants. The first one is necessary for technical reasons and gives good bounds for $\theta^K_{k+1}$. The second one is needed to keep track of these new populations. As in \cite{CoKrSm21}, let $(\t^K_h)_{h\geq 0}$ be the collection of both $(s^K_k)_{k\geq 0}$ and  $(t^K_{w,k})_{k\geq 0}$. The main part of the proof then consists of approximating the growth dynamics in the intervals $[\t^K_h\land\theta^K_{k+1}\land T,\t^K_{h+1}\land\theta^K_{k+1}\land T]$. Subsequently, we estimate the time between $\theta^K_{k+1}$ and $\s^K_{k+1}$, which completes the induction step.
	
	\begin{figure}[h]
		\centering
		\includegraphics[scale=0.6]{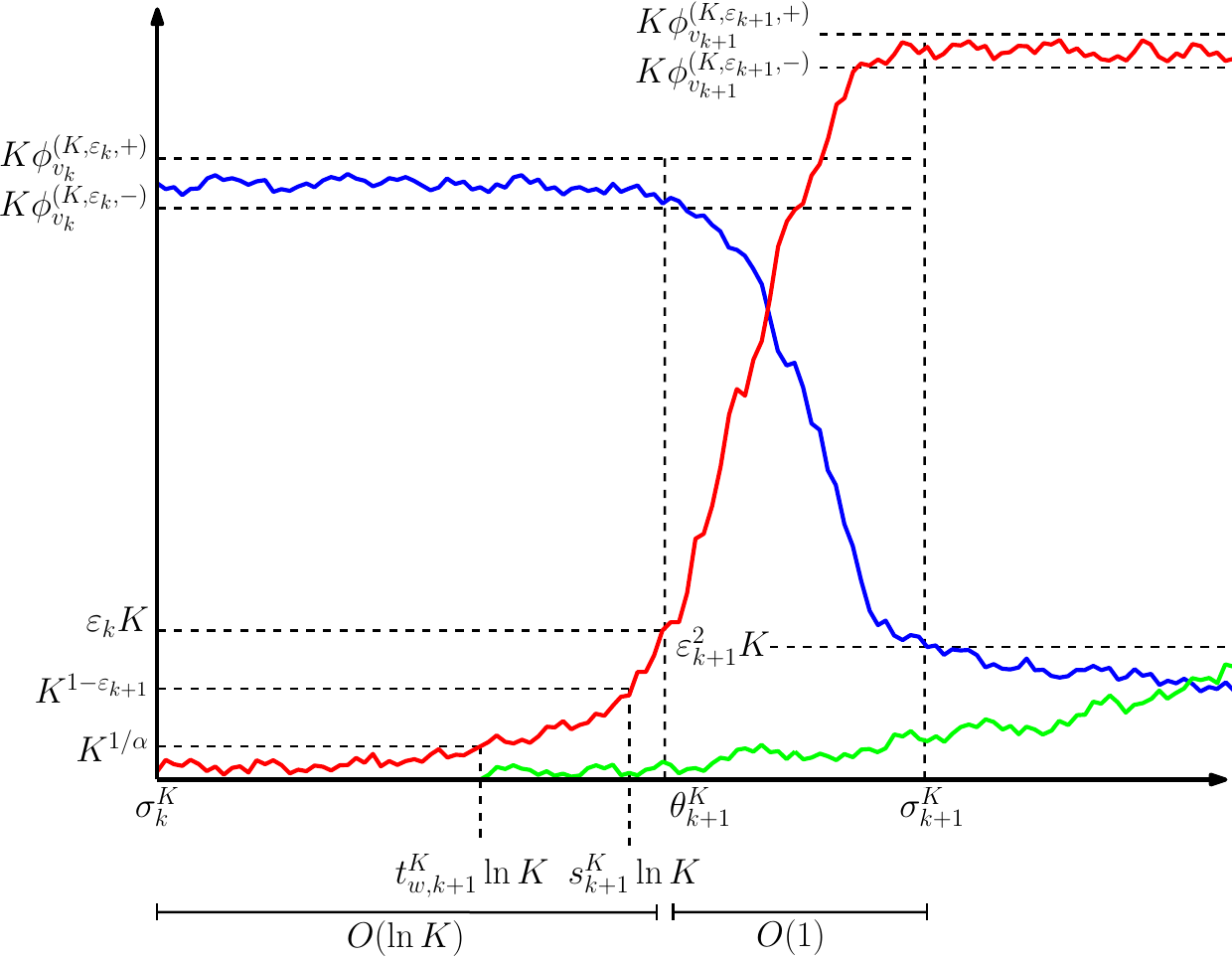}
		\caption{Substeps of the $(k+1)^\text{st}$ invasion: Resident population $N_{v_k}^K(t)$ (blue), invading mutant $N_{v_{k+1}}^K(t)$ (red) and new emerging subpopulation $N_w^K(t)$ (green), together with the triggered stopping times and the corresponding bounds and thresholds.}
		\label{fig:StoppingTimes}
	\end{figure}	
	
	\textbf{Step 2 (base case):}
	We set $\theta^K_0$=0. Then the base case is reminiscent of \cite{CoKrSm21} since, within a finite time horizon $[0,T'_\ve]$, the parameter functions $b^K,d^K,c^K$ stay constant, for $K$ large enough. Hence, we can apply Lemma A.6(ii) of \cite{CoKrSm21} to get, for every $\ve_0>0$, a $T'_{\ve_0}<\infty$ such that
	\begin{align}
		\lim_{K\to\infty}\Prob{\norm{\frac{N^K_{v_0}(T'_{\ve_0})}{K}-\bar{n}^1_{v_0}}_\infty<\ve_0}=1.
	\end{align}
	
	Here, our $\eps_0$ corresponds to $\eps_2$ in \cite{CoKrSm21} and a lower bound of order $K$ for the initial population size of trait $v_0$ gives their $\eps_1 K$.
	We use that our assumption on the initial condition \eqref{eq:result_initialcondition_mono} guarantees that $\lim_{K\to\infty}\beta_w(0)<1$, for all $w\neq v_0$. Hence $\lim_{K\to\infty}\beta_w(t)<1$ for all $w\neq v_0$ and $t\in[0,T'_{\ve_0}]$, and in particular $\sum_{w\neq v_0}N^K_w(t)<\eps_0^2K$ for such $t$. Therefore, Assumption A.5 in \cite{CoKrSm21} is satisfied for standard couplings to birth death processes with immigration, where the birth, death and (self)competition rates do not deviate from $b^1_{v_0}$, $d^1_{v_0}$ and $c^1_{v_0,v_0}$ by more than $\eps_0^2\hat{C}$, and immigration is bounded by $\eps_0^2K\mu_K\hat{C}$, for some $\hat{C}<\infty$.
	Overall, we obtain that $\s^K_0\leq T'_{\ve_0}<\infty$, for $K$ large enough.

	\textbf{Step 3 (growth of mutants):} 
	To show the induction step, let us assume that at time $\s^K_k$ the process has reached the monomorphic Lotka-Volterra-equilibrium of trait $v_k\in V$. Our first goal is to estimate the competitive interaction between the subpopulations in the interval $[\s^K_k,\theta^{K}_{k+1}]$. Recalling the rates of the different events for the population of trait $w\in V$, which has size $N^K_w(t)$ at time $t$, we have

	\begin{itemize}
		\item Reproduction without mutation:
		\begin{align}
			\mathfrak{b}^K_w(t)=b^K_w(t)(1-\mu_K)N^K_w(t),
		\end{align}
		\item Death (natural and by competition):
		\begin{align}
			\mathfrak{d}^K_w(t)=\left(d^K_w(t)+\sum_{u\in V}\frac{c^K_{w,u}(t)}{K}N^K_u(t)\right)N^K_w(t),
		\end{align}
		\item Reproduction from mutation:
		\begin{align}
			\mathfrak{bm}^K_w(t)=\mu_K\sum_{u\in V}b^K_u(t)m_{u,w}N^K_u(t).
		\end{align}
	\end{itemize}
	
	Using the shorthand notation $v:=v_{k}$ and $\hat{c}_w:=\max_{u\neq v, 1\leq i\leq\ell}c^i_{w,u}$, we can introduce the approximating parameter functions, for $w\in V$,
	\begin{align}
		\begin{aligned}
			b^{(K,\ve,+)}_w(t)&:=b^K_w(t),\\
			b^{(K,\ve,-)}_w(t)&:=(1-\ve)b^K_w(t),\\
			d^{(K,\ve,+)}_w(t)&:=d^K_w(t)+c^K_{w,v}(t)\phi^{(K,\ve,-)}_v(t),\\
			d^{(K,\ve,-)}_w(t)&:=d^K_w(t)+c^K_{w,v}(t)\phi^{(K,\ve,+)}_v(t)+\hat{c}_w\ve.
		\end{aligned}
	\end{align}
	If $K$ is taken large enough such that $\mu_K<\ve$, we therefore have, for $t\in[\s^K_{k},\theta^{K}_{k+1}]$,
	\begin{align}
		\begin{aligned}
			b^{(K,\ve,-)}_w(t)N^K_w(t)&\leq\mathfrak{b}^K_w(t)\leq b^{(K,\ve,+)}_w(t)N^K_w(t),\\
			d^{(K,\ve,+)}_w(t)N^K_w(t)&\leq\mathfrak{d}^K_w(t)\leq d^{(K,\ve,-)}_w(t)N^K_w(t).
		\end{aligned}
	\end{align}
	Moreover, defining $f^{(K,\ve,\pm)}_{w,v}(t):=b^{(K,\ve,\pm)}_w(t)-d^{(K,\ve,\pm)}_w(t)$,
	we have	
	\begin{align}
		\begin{aligned}
			f^{(K,\ve,+)}_{w,v}(t)&=\left\{\begin{array}{l}
			f^i_{w,v}+c^i_{w,v}M\ve+c^i_{w,v}\left(\bar{n}^{i}_v-\bar{n}^{i-1}_v\right)_+\\
			f^i_{w,v}+c^i_{w,v}M\ve
			\end{array}\right.
			&\hspace{-1.5em}
			\begin{array}{l}
				:t\in[T^\S_{i-1}\lK,T^\S_{i-1}\lK+T_\ve),\\
				:t\in[T^\S_{i-1}\lK+T_\ve,T^\S_{i}\lK),
			\end{array}\\
			f^{(K,\ve,-)}_{w,v}(t)&=\left\{\begin{array}{l}
				f^i_{w,v}-\left(c^i_{w,v}M+b^i_w+\hat{c}_w\right)\ve-c^i_{w,v}\left(\bar{n}^{i-1}_v-\bar{n}^{i}_v\right)_+\\
				f^i_{w,v}-\left(c^i_{w,v}M+b^i_w+\hat{c}_w\right)\ve
			\end{array}\right.
			&\hspace{-1.5em}
			\begin{array}{l}
				:t\in[T^\S_{i-1}\lK,T^\S_{i-1}\lK+T_\ve),\\
				:t\in[T^\S_{i-1}\lK+T_\ve,T^\S_{i}\lK).
			\end{array}	
		\end{aligned}
	\end{align}
	We can use these new parameter and fitness functions to define suitable couplings to simpler bounding branching processes to approximate the original processes $N^K_w$.
	
	In contrast to the estimates of \cite{CoKrSm21}, we have to work with periodic functions instead of constants. Another 
	peculiarity is that we only have good estimates in the intervals\linebreak \mbox{$[T^\S_{i-1}\lK+T_\ve,T^\S_{i}\lK)$}. On the intervals $[T^\S_{i-1}\lK,T^\S_{i-1}\lK+T_\ve)$, we have to deal with deviations staying macroscopic (i.e.\ not scaling with $\eps$). Fortunately, these bad estimates are only given for a finite time $T_\ve$ (not increasing with $\lK$) and for the remaining time,  which scales with $\lK$, we have the estimates that are arbitrarily accurate, proportional to $\ve>0$. In Appendix \ref{app:B_branching_processes}, we work out how to capture both of these characteristics.

	In order to make these results applicable, let us first define the stopping time when the first non-resident population becomes almost macroscopic
	\begin{align}
		s^K_{k+1}:=\inf\dset{t\geq \s^K_{k}/\ln K: \exists\ w\neq v_{k}: \b^K_w(t)>1-\ve_{k+1}},
	\end{align}	
	as well as the time of appearance of a mutant
	\begin{align}
		t^K_{w,k+1}:=\left\{\begin{array}{ll}
			\inf\dset{t\geq \s^K_{k}/\ln K:\exists\ u\in V: d(u,w)=1, \b^K_u(t)=\frac{1}{\a}}
			&\text{if\ }\b^K_w(\s^K_{k}/\ln K)=0,\\
			\s^K_{k}/\ln K
			&\text{else}.
		\end{array}\right.
	\end{align}
	Building on this, we can define the sequence of important events $(\t_h)_{h\in\N_0}$ via $\t^K_0=\s^K_0/\ln K$ and, for $\s^K_{k}/\ln K\leq\t^K_{h-1}<s^K_{k+1}$,
	\begin{align}
		\t^K_h:=s^K_{k+1}\wedge\min\dset{t_{w,k+1}:w\in V,t_{w,k+1}>\t_{h-1}}.
	\end{align}
	Moreover, we can then define the sequence of sets of living traits $(M^K_h)_{h\in\N_0}$ via
	\begin{align}
		M^K_h:=& \dset{w\in V: \b^K_w(\t^K_h)>0 \text{ or } \t^K_h=t^K_{w,k+1}}\nonumber\\
		=&\left(M^K_{h-1}\setminus\dset{w\in V: \b^K_w(\t^K_h)=0}\right)
		\cup\dset{w\in V: \t^K_h=t^K_{w,k+1}}.
	\end{align}
	
	After establishing these stopping times and estimates on the rate functions, we are in a similar framework as in \cite{CoKrSm21} but now adapted to the time-dependence of the driving parameters.  We can couple the mutant populations $N^K_w$ to time-dependent birth death processes (with immigration) with parameters $b^{(K,\ve_k,\pm)}_w(t)$ and $d^{(K,\ve_k,\pm)}_w(t)$. Together with the results presented in Appendix \ref{app:B_branching_processes}, one can now follow the arguments of Section 4.2 in \cite{CoKrSm21}. To be precise, we just have to replace their Lemma A.1 by our Theorem \ref{thm:bd_main} and their Corollary A.4 by our Theorem \ref{thm:bdi_main} and use similar inductive arguments to show that, for $w\in M^K_{h-1}$ and $t\in[\t^K_{h-1}\ln K\wedge\theta^{K}_{k+1}\wedge T\ln K,\t^K_{h}\ln K\wedge\theta^{K}_{k+1}\wedge T\ln K]$, we obtain the bounds
	\begin{multline}
	\label{eq:pf_beta_bounds}
		\max_{u\in M^K_{h-1}} \left[\b^K_u(\t^K_{h-1})-\frac{d(u,w)}{\a}+(t-\t^K_{h-1})(f^\av_{u,v_k}-C_h\ve_k)\right]_+\\
		\leq\b^K_w(t)\leq\\
		\max_{u\in M^K_{h-1}} \left[\b^K_u(\t^K_{h-1})-\frac{d(u,w)}{\a}+C_h\ve_k+(t-\t^K_{h-1})(f^\av_{u,v_k}+C_h\ve_k)\right]_+.
	\end{multline}
	A heuristic for these bounds has already been given in Remark \ref{rem:betabar}. As a brief reminder, without mutation, every living trait $u\in M^K_{h-1}$ would grow/shrink at the rate of its own fitness $f^\av_{u,v_k}$ on the $\ln K$-time-scale, yielding $\beta^K_u(t)\approx \b^K_u(\t^K_{h-1})+(t-\t^K_{h-1})f^\av_{u,v_k}$. Through mutation however, a trait $w\in V$ receives a $\mu_K^{d(u,w)}=K^{-d(u,w)/\alpha}$ portion of incoming mutants from all living traits $u\in M^K_{h-1}$, and its actual population size corresponds the leading order term, i.e.\ the maximum of all these exponents $\b^K_u(\t^K_{h-1})+(t-\t^K_{h-1})f^\av_{u,v_k}-d(u,w)/\alpha$.
	
	Note that these estimates on $\b^K_w$ are also where the errors accumulate. Namely, knowing the initial value of $\b^K_u$ at time $\t_{h-1}\ln K$ allows for approximations until $\t_h\ln K$ but at a cost of an additional error term of order $\ve_{k}$. To eventually ensure the convergence of Theorem \ref{thm:conv_beta}, which means having good estimates until time $T\ln K$, one has to choose the $\ve_{k}$ in reverse order, such that every approximation step gives good enough bounds to the initial values of the next one.
	
	By analogous arguments to Section 4.3 and 4.4 of \cite{CoKrSm21}, we can deduce the formulas for and the convergence of $(\t^K_h-\t^K_{h-1})$ and $M^K_h$, as well as for $s^K_{k+1}$ and $\theta^{K}_{k+1}$. 
	Note that we need to introduce the stopping times $s^K_{k+1}$ for the following technical reason: At time $\theta^K_{k+1}$ we only know that the overall mutant population (summing over all non-resident traits) has reached the threshold $\eps_k K$. The above bounds on $\beta^K_w$ only allow us to estimate a single mutant's population size up to a multiplicative factor of $K^{\pm C \eps_k}$, which is not sufficient to imply that $v_{k+1}$ has a non-vanishing population size (when rescaled by $K$) at this time. This is however necessary to have a good initial condition for the deterministic Lotka-Volterra approximation. Hence, $s^K_{k+1}$ is chosen to guarantee the existence of one large mutant population, where we choose a threshold slightly smaller than $\beta^K_w=1$ to ensure that the Lotka-Volterra dynamics are not triggered before this time either.
	
	 We can show that, in the limit of $K\nearrow\infty$ and for small $\ve_{k+1}$, the times $\theta^K_{k+1}/\ln K$ and $s^K_{k+1}$ are arbitrarily close.
	Namely, following again the argument in \cite{CoKrSm21}, we can show by contradiction that
	\begin{align}\label{eq:compareST}
		s^K_{k+1}\ln K<\theta^{K}_{k+1}<(s^K_{k+1}+\ve_{k+1}C)\ln K.
	\end{align}
	To be precise, let $w^K_{k+1}$ be the mutant trait that triggers $s^K_{k+1}$ and take
	\begin{align}
		\eta^K_{k+1}:=2\ve_{k+1}/(f^\av_{w_{k+1}^K,\v_k}-C_h\ve_{k}).
	\end{align}
	If one assumes that $(s^K_{k+1}+\eta^K_{k+1})\ln K\leq\theta^{K}_{k+1}$, then \eqref{eq:pf_beta_bounds} would still be applicable and directly lead to
	\begin{align}
		\b^K_{w^K_{k+1}}(s^K_{k+1}+\eta^K_{k+1})\geq 1+\ve_{k+1}.
	\end{align}
	This however is a contradiction since $\lim_{K\to\infty}\b^K_w(s)\leq 1$ for all $w\in V$ and $s\geq0$, and hence the upper bound in \ref{eq:compareST} is satisfied {for $C=2/(f^\av_{w_{k+1}^K,\v_k}-C_h\ve_{k})$. The lower bound is satisfied by definition of the stopping times.

	\textbf{Step 4 (time of invasion):}	
	Now the last difference to \cite{CoKrSm21} that we have to address is that the trait reaching a macroscopic size at time $\theta^{K}_{k+1}$, which is with high probability $v_{k+1}$, might be unfit at that time. In the following, we show that this only happens with vanishing probability. In order to track the sizes of the different subpopulations more carefully, let us introduce the additional stopping times
	\begin{align}
		R^K_{k+1}&:=\inf\dset{t\geq\s^K_k:N^K_{v_{k+1}}(t)\geq\ve^2_{k} K \text{ and } f^{(K,\eps,-)}_{v_{k+1},v_k}(t)>0}, \\
		\tilde{R}^K_{k+1}&:=\inf\dset{t\geq\s^K_k:N^K_{v_{k+1}}(t)\geq\ve^2_{k} K}, \\
		\check{R}^K_{k+1}&:=\inf\dset{t\geq\tilde{R}^K_{k+1}:\int_{\tilde{R}^K_{k+1}}^{t} f^{(K,\eps_k,-)}_{v_{k+1},v_k}(s)\dd s>0}.
	\end{align}

	The first time $R^K_{k+1}$ is the time we are ultimately looking for, namely the starting point for the deterministic Lotka-Volterra system involving the resident trait $v_k$ and the (at that time fit) mutant $v_{k+1}$ (see Step 5 below). The second time $\tilde{R}^K_{k+1}$ gives us a first warning before reaching $\theta^K_{k+1}$, with $v_{k+1}$ possibly being unfit. The last time $\check{R}^K_{k+1}$ helps to estimate the first one and can be computed deterministically in relation to the second one.
	
	Our goal in this step is to prove that $R^K_{k+1}<\theta^K_{k+1}$, such that all branching process approximations apply up to this point. While in the next step we deduce from the Lotka-Volterra system $\sigma^K_{k+1}<R^K_{k+1}+O(1)$.
	
	We know that, for $\ve_k>0$ small and $K$ large enough,
		\begin{align}
			\int_{\tilde{R}^K_{k+1}}^{\tilde{R}^K_{k+1}+\lK T^\S} f^{(K,\ve_{k},-)}_{v_{k+1},v_k}(s)\dd s
			&\geq \lK T^\S_\ell \left(f^{\text{av}}_{v_{k+1},v_k}-\tilde{M}\ve_{k}\right) -\ell T_{\ve_k} C_{v_{k+1},v_k}>0,
	\end{align}
	since $f^{\text{av}}_{v_{k+1},v_k}>0$. Which implies directly
	\begin{align}
		\check{R}^K_{k+1}\leq \tilde{R}^K_{k+1}+\lK T^\S.
	\end{align}
	Moreover, since $f^{(K,\eps,-)}_{v_{k+1},v_k}(s)$ is piecewise constant and the defining inequality of $\check{R}^K_{k+1}$ is strict, there is a small $\d>0$ (not scaling with $K$) such that, with probability 1,
	\begin{align}
	\label{eq:pf_f>0}
		f^{(K,\eps,-)}_{v_{k+1},v_k}(\check{R}^K_{k+1}+t)>0,\qquad \forall t\in (0,\d).
	\end{align}
	As argued above, the interval $[\tilde{R}^K_{k+1},\check{R}^K_{k+1}+\d]$ is of length $O(\lK)$. Hence from Corollary \ref{cor:growth_in_lambda_time} and an application of the Markov property at time $\tilde{R}^K_{k+1}$, we can deduce that, for $\d$ small enough,
	\begin{align}
	\label{eq:pf_still_bdd}
		\Prob{\sup_{t\in [\tilde{R}^K_{k+1},\check{R}^K_{k+1}+\d]} \sum_{w\neq v_k}N^K_w(t)<\ve_k K}\overset{K\to\infty}{\longrightarrow}1,\\
	\label{eq:pf_threshold_crossed}
		\Prob{N^K_{v_{k+1}}(\check{R}^K_{k+1}+\d)>\ve_k^2 K}\overset{K\to\infty}{\longrightarrow}1.
	\end{align}
	
	The statement of \eqref{eq:pf_still_bdd} tells us that the mutant population is still bounded from above and thus the assumptions for Theorem \ref{Thm_phi} are still satisfied up to time $\check{R}^K_{k+1}+\d$. Hence we know that the resident population only fluctuates inside the $\phi$-tube. This implies
	\begin{align}
		\check{R}^K_{k+1}+\d\leq\theta^K_{k+1}.
	\end{align}
	Finally \eqref{eq:pf_threshold_crossed} together with \eqref{eq:pf_f>0} leads to $R^K_{k+1}\leq\check{R}^K_{k+1}+\d$.
	This eventuelly gives 
	\begin{align}
		R^K_{k+1}\leq\theta^K_{k+1},
	\end{align}
	i.e.\ we are still allowed to use the couplings with birth death processes to approximate the mutant population up to $R^K_{k+1}$.

	\textbf{Step 5 (Lotka-Volterra):}
	At time $R^K_{k+1}$, we are in position to use the convergence  result for the fast Lotka-Volterra phase. By definition of this stopping time, we know that the invading trait $v_{k+1}$ is fit with respect to the resident $v_{k}$ and of a size that does not vanish as $K\to\infty$ when rescaled by $K$. Moreover, termination criterion (d) of the algorithm in Theorem \ref{thm:conv_beta} ensures that the resident trait is unfit with respect to the invading mutant. By standard arguments, the corresponding deterministic system gets close to its equilibrium in finite time and we have convergence of the stochastic process towards the deterministic system on finite time intervals \cite{EtKu86}. This implies the existence of a finite and deterministic time $T'_{\ve_{k+1}}<\infty$ such that
	\begin{align}
		\s^K_{k+1}\leq R^K_{k+1}+T'_{\ve_{k+1}}.
	\end{align}
	
	Moreover, the condition $f^i_{v_k,v_{k+1}}<0$, for all $i=1\ldots\ell$, guarantees that (with probability converging to 1 as $K\to\infty$) the former resident population cannot reach the threshold $\ve_{k+1}K$ any more after time $\s^K_{k+1}$.

	Overall, we have proved that, with probability converging to 1 as $K\to\infty$,
	\begin{align}
		R^K_{k+1}
		\leq \theta^K_{k+1}
		\leq \s^K_{k+1}
		\leq R^K_{k+1}+T'_{\ve_{k+1}},
	\end{align}
	which means that on the logarithmic time-scale there is no difference between $R^K_{k,v_{k+1}},\theta^K_{k+1}$ and $\s^K_{k+1}$ and dividing by $\ln K$ they all converge to $s_{k+1}$ as claimed.
	
This finishes the proof of Theorem \ref{thm:conv_beta}.

\subsection{Sequence of resident traits}
\label{sec:4.3_seq_residents}
	We now turn to the proof of Corollary \ref{cor:seq_residents}. To prove the convergence with respect to $\mathcal{M}(V)$, equipped with the weak topology, we have to study the integrals $\dangle{\nu,h}=\int h\dd \nu$ of all bounded and continuous functions $h:V\mapsto\R$ with respect to the measures $\nu^K_\eps(s\ln K)$. Since $V$ is discrete and finite, all finite functions satisfy these conditions. For later purpose we denote $\bar{h}:=\max_{v\in V}\abs{h(v)}$. Under use of \eqref{eq:pf_stoppingtimes}, we have
	\begin{align}
		\sum_{k\in\N_0}\left(-\ifct{\theta^K_k\leq s<\s^K_{k}}2\bar{h}+\ifct{\s^K_k\leq s<\theta^K_{k+1}}h(v_k) \right)
		&\leq\dangle{\nu^K_\ve(s\ln K),h}
		\leq\sum_{k\in\N_0}\left(\ifct{\theta^K_k\leq s<\s^K_{k}}2\bar{h}+\ifct{\s^K_k\leq s<\theta^K_{k+1}}h(v_k) \right),\\
		&\dangle{\nu(s),h}
		=\sum_{k\in\N_0}\ifct{s_k\leq s<s_{k+1}}h(v_k).
	\end{align}
	Since we want to show convergence in $L^p([0,T],\mathcal{M}(V))$, for $p\in [1,\infty)$, we have to compute the distance between the two integrals in the $\norm{\cdot}_{L^p([0,T])}$-norm, which can be estimated as follows
	\begin{align}
		&\norm{\dangle{\nu^K_\ve(\cdot\ln K),h}-\dangle{\nu(\cdot),h}}^p_{L^p([0,T])} \nonumber\\
		\leq& \sum_{k\in\N_0:s_k<T} \left(
		(3\bar{h})^p\abs{\frac{\theta^K_k}{\ln K}-\frac{\s^K_k}{\ln K}}
		+(2\bar{h})^p\abs{\frac{\s^K_k}{\ln K}-s_k}
		+(2\bar{h})^p\abs{s_{k+1}-\frac{\theta^K_{k+1}}{\ln K}}\right) \nonumber\\
		\leq& (5\bar{h})^p \sum_{k\in\N_0:s_k<T}
		\left(\abs{\frac{\s^K_k}{\ln K}-s_k}
		+\abs{\frac{\theta^K_{k+1}}{\ln K}-s_{k+1}} \right).
	\end{align}
	Here the last step consists of an application of triangle inequality at $s_k$ to estimate the first term, followed by a reordering of the sum.
	Since for fixed $T>0$ the sum in fact only consists of finitely many summands and moreover $\s^K_k/\ln K\to s_k$ and $\theta^K_{k+1}/\ln K\to s_{k+1}$ in probability, for $K\to\infty$, we deduce, for all $\d>0$,
	\begin{align}
			\Prob{\norm{\dangle{\nu^K_\ve(\cdot\ln K),h}-\dangle{\nu(\cdot),h}}_{L^p([0,T])}>\d}\overset{K\to\infty}{\longrightarrow}0,
	\end{align}
	which is the claimed convergence.


\appendix
\section{Birth death processes with self-competition}
\label{app:A_attraction_equil}

In this chapter, we prove some general results on birth death processes with self-competition that are used to obtain bounds on the resident's population size in Section \ref{sec:4.1_proof_stability}. In the first section, we quantify the asymptotic probability of such processes to stay close to their equilibrium state for a long time as $K$ tends to infinity. In the second section,  we derive asymptotics for the probability of these processes to stay close to the corresponding deterministic system for a finite time.

Both results apply to processes with constant parameters. More precisely, we study stochastic processes $(X^{K}_t)_{t\geq0}$ with birth rate $b$, natural death rate $d$ and self-competition rate $c/K$, i.e.\ with infinitesimal generators
\begin{align}
\left(\mathcal{L}^Kf\right)(n)=nb(f(n+1)-f(n))+n\left(d+\frac{c}{K}n\right)(f(n-1)-f(n)),
\end{align}
for bounded functions $f:\N_0\to\R$.

\subsection{Attraction to the equilibrium}

We study the probability of birth death processes with self-competition to stay close to their equilibrium $(b-d)K/c$ for a long time. In order to be able to concatenate this result for infinitely many phases in Section \ref{sec:4.1_proof_stability}, we need to bound the probability of diverging from the equilibrium by a sequence that tends to zero fast enough as $K$ tends to infinity. We start by proving a general result for time horizons $\theta_K$. The proof uses a potential theoretic approach, similar to the proof of \cite[Lem.\ 6.3]{BaBoCh17}.

\begin{theorem}\label{stability_general}
Let $X^{K}$ be a birth death process with self-competition and parameters\linebreak $0<d<b$ and $0<c/K$. Define $\bar{n}:=(b-d)/c$. Then there are constants $0<C_1,C_2,C_3<\infty$ such that, for any $\eps$ small and any $K$ large enough, any initial value $0\leq |x-\lceil\bar{n}K\rceil|\leq\frac{1}{2}\left\lfloor\frac{\eps K}{2}\right\rfloor$, any $m\geq0$, and any non-negative sequence $(\theta_K)_{K\in\N}$,
\begin{align}
\P_x(\exists\ t\in[0,\theta_K]:|X^{K}(t)-\lceil \bar{n}K\rceil|>\eps K)\leq mC_1e^{-C_2\eps^2K}+\sum_{l=m}^\infty\left(4\left(1-e^{-C_3K\theta_K/l}\right)^{1/2}\right)^l.
\end{align}
\end{theorem}

\begin{proof}
We start by defining a couple of new processes based on $X^{K}$. Let
\begin{align}
	V^{K}(t):=|X^{K}(t)-\lceil \bar{n}K\rceil|
\end{align}
be the distance of $X^{K}$ from its equilibrium state $\bar{n}K$ at time $t$. Note that this is no longer a Markov process. For $V^{K}$, let $(Y^{K}_n)_{n\in\N_0}$ be its discrete jump chain (taking values in $\N$, not Markovian) and $(S^K_n)_{n\in\N}$ its jump times.

The proof is divided into multiple steps:
\begin{itemize}
\item[1)] Bound the transition probabilities of $Y^{K}$.
\item[2)] Define a discrete time Markov chain $(Z^{K}_n)_{n\in\N_0}$ such that $Z^{K}_n\geq Y^{K}_n$, for all $n\in\N$.
\item[3)] Derive an upper bound for the probability of $Z^{K}$ hitting $\lfloor\eps K\rfloor$ before 0.
\item[4)] Derive an upper bound for the probability of $Z^{K}$ returning to 0 at most $m$ times before hitting $\lfloor\eps K\rfloor$.
\item[5)] Consider a continuous time version $\tilde{Z}^{K}$ of $Z^{K}$ to deduce the final result.
\end{itemize}

\textbf{Step 1:} The discrete-time process $Y^{K}$ changes its state due to either a birth or a death event in the original process $X^{K}$ and hence moves by increments of $\pm1$ in each step. It is therefore a random walk on $\N_0$ that is reflected in 0. For the boundary case, we obtain
\begin{align}
	\P(Y^{K}_{n+1}=1|Y^{K}_n=0)=1.
\end{align}
For any other $1\leq i\leq\eps K$, using that $c\lceil\bar{n}K\rceil/K\in[b-d,b-d+c/K]$, we can bound
\begin{align}
\P&(Y^{K}_{n+1}=i+1|Y^{K}_n=i)\nonumber\\
&=\P(\text{birth event if }X^{K}=\lceil\bar{n}K\rceil+i \text{ or death event if }X^{K}=\lceil\bar{n}K\rceil-i)\nonumber\\
&\leq \frac{b}{b+d+\frac{c}{K}(\lceil\bar{n}K\rceil+i)}\lor\frac{d+\frac{c}{K}(\lceil\bar{n}K\rceil-i)}{b+d+\frac{c}{K}(\lceil\bar{n}K\rceil-i)}\nonumber\\
&\leq \frac{b}{2b+\frac{c}{K}i}\lor\frac{b-\frac{c}{K}(i-1)}{2b-\frac{c}{K}i}=\left(\frac{1}{2}-\frac{\frac{c}{2K}i}{2b+\frac{c}{K}i}\right)\lor\left(\frac{1}{2}-\frac{\frac{c}{2K}(i-2)}{2b-\frac{c}{K}i}\right)\nonumber\\
&\leq\frac{1}{2}-C\frac{i}{K}=:p^{K}_+(i),
\end{align}
for some constant $C>0$, as long as $\eps\leq \bar{n}$ and hence $ci/K\leq c\eps\leq b-d$.

\textbf{Step 2:} Define a discrete-time process $(Z^{K}_n)_{n\in\N_0}$ that is coupled to $(Y^{K}_n)_{n\in\N_0}$ by
\begin{itemize}
\item $Z^{K}_0=Y^{K}_0$
\item Whenever $Z^{K}_n=Y^{K}_n=i$ and $Y^{K}_{n+1}=i+1$, we set $Z^{K}_{n+1}=i+1$.
\item Whenever $Z^{K}_n=Y^{K}_n=i$ and $Y^{K}_{n+1}=i-1$,  we set $Z^{K}_{n+1}=i+1$ with probability $(p^{K}_+(i)-\P(Y^{K}_{n+1}=i+1|Y^{K}_n=i))/\P(Y^{K}_{n+1}=i-1|Y^{K}_n=i)$ and $Z^{K}_{n+1}=i-1$ else.
\item Whenever $Z^{K}_n=i>Y^{K}_n$, we set $Z^{K}_{n+1}=i+1$ with probability $p^{K}_+(i)$ and $Z^{K}_{n+1}=i-1$ else.
\end{itemize}
Then $Z^{K}$ is a discrete-time Markov chain such that $Z^{K}_n\geq Y^{K}_n$, for all $n\in\N_0$, and
\begin{align}
p^{K}(i,j):=\P(Z^{K}_{n+1}=j|Z^{K}_n=i)=\begin{cases}1&i=0,\ j=1,\\
p^{K}_+(i)&i\geq1,\ j=i+1,\\
1-p^{K}_+(i)&i\geq1,\ j=i-1,\\
0&\text{else.}\end{cases}
\end{align}

\textbf{Step 3:} For the Markov chain $Z^{K}$, we define the stopping times
\begin{align}
	\tau^{(Z,K)}_j:=\inf\{n\in\N_0:Z^{K}_n=j\}.
\end{align}
By standard potential theoretic arguments (see \cite[Ch.\ 7.1.4]{BovHol15}), we obtain, for initial values $0\leq z\leq\lfloor\eps K\rfloor$,
\begin{align}
\P_z\left(\tau^{(Z,K)}_{\lfloor\eps K\rfloor}<\tau^{(Z,K)}_0\right)
=\frac{\sum_{i=1}^z\prod_{j=1}^{i-1}\frac{p(j,j-1)}{p(j,j+1)}}{\sum_{i=1}^{\lfloor\eps K\rfloor}\prod_{j=1}^{i-1}\frac{p(j,j-1)}{p(j,j+1)}}
=\frac{\sum_{i=1}^z\exp\left(\sum_{j=1}^{i-1}\ln\left(\frac{1+2C\frac{j}{K}}{1-2C\frac{j}{K}}\right)\right)}{\sum_{i=1}^{\lfloor\eps K\rfloor}\exp\left(\sum_{j=1}^{i-1}\ln\left(\frac{1+2C\frac{j}{K}}{1-2C\frac{j}{K}}\right)\right)}.
\end{align}

Using that $\ln(1+\xi)=\xi+O(\xi^2)$, as $\xi\to0$, and $j\leq\lfloor\eps K\rfloor$, we can approximate, as $\eps\to0$,
\begin{align}
\ln\left(\frac{1+2C\frac{j}{K}}{1-2C\frac{j}{K}}\right)&=\ln\left(1+\frac{4C\frac{j}{K}}{1-2C\frac{j}{K}}\right)=\frac{4C\frac{j}{K}}{1-2C\frac{j}{K}}+O\left(\left(\frac{4C\frac{j}{K}}{1-2C\frac{j}{K}}\right)^2\right)\nonumber\\
&=4C\frac{j}{K}\left(1+\frac{2C\frac{j}{K}}{1-2C\frac{j}{K}}\right)+O\left(\left(\frac{j}{K}\right)^2\right)\nonumber\\
&=4C\frac{j}{K}+O\left(\left(\frac{j}{K}\right)^2\right)=4C\frac{j}{K}+O(\eps^2).
\end{align}

Plugging these asymptotics back into the above expression yields
\begin{align}
\P_z\left(\tau^{(Z,K)}_{\lfloor\eps K\rfloor}<\tau^{(Z,K)}_0\right)
&=\frac{\sum_{i=1}^z\exp\left(\sum_{j=1}^{i-1}4C\frac{j}{K}+O(\eps^2)\right)}{\sum_{i=1}^{\lfloor\eps K\rfloor}\exp\left(\sum_{j=1}^{i-1}4C\frac{j}{K}+O(\eps^2)\right)}\nonumber\\
&\leq\frac{\sum_{i=1}^z\exp\left(4C\frac{i(i-1)}{2K}+O((i-1)\eps^2)\right)}{\sum_{i=\left\lfloor\frac{\eps K}{2}\right\rfloor}^{\lfloor\eps K\rfloor}\exp\left(4C\frac{i(i-1)}{2K}+O((i-1)\eps^2)\right)}\nonumber\\
&\leq\frac{z\exp\left(2C\frac{z^2}{K}+O(z\eps^2)\right)}{\left\lfloor\frac{\eps K}{2}\right\rfloor\exp\left(\frac{2C}{K}\left(\left\lfloor\frac{\eps K}{2}\right\rfloor^2-\left\lfloor\frac{\eps K}{2}\right\rfloor\right)+O(K\eps^3)\right)}\nonumber\\
&=\frac{z}{\left\lfloor\frac{\eps K}{2}\right\rfloor}\exp\left(\frac{2C}{K}\left(z^2-\left\lfloor\frac{\eps K}{2}\right\rfloor^2+\left\lfloor\frac{\eps K}{2}\right\rfloor\right)+O(z\eps^2)+O(K\eps^3)\right)\nonumber\\
&\leq\frac{z}{\left\lfloor\frac{\eps K}{2}\right\rfloor}\exp\left(\frac{2C}{K}\left(z^2-\left\lfloor\frac{\eps K}{2}\right\rfloor^2+\left\lfloor\frac{\eps K}{2}\right\rfloor+O(K^2\eps^3)\right)\right)\nonumber\\
&\leq\frac{1}{2}\exp\left(-\frac{2C}{K}\frac{1}{4}\left\lfloor\frac{\eps K}{2}\right\rfloor^2\right)\nonumber\\
&\leq C_1 e^{-C_2\eps^2K},
\end{align}
for some uniform constants $C_1,C_2>0$, as long as $0\leq z\leq\frac{1}{2}\left\lfloor\frac{\eps K}{2}\right\rfloor$ and $\eps$ small enough such that $\left\lfloor\frac{\eps K}{2}\right\rfloor+O(K^2\eps^3)\leq \frac{1}{2}\left\lfloor\frac{\eps K}{2}\right\rfloor^2$ for large $K$.

\textbf{Step 4:} Let $B^{K}$ be the random variable that describes the number of visits to 0 of $Z^{K}$ before first hitting $\lfloor\eps K\rfloor$ (not counting the first visit/start in case $Z^{K}_0=0$).  First consider $1\leq z\leq\frac{1}{2}\left\lfloor\frac{\eps K}{2}\right\rfloor$. Then, for $\eps$ small and $K$ large enough,
\begin{align}
\P_z\left(B^{K}=0\right)=\P_z\left(\tau^{(Z,K)}_{\lfloor\eps K\rfloor}<\tau^{(Z,K)}_0\right)\leq C_1e^{-C_2\eps^2K}
\end{align}
and, for all $l\geq1$, due to the strong Markov property,
\begin{align}
\P_z\left(B^{K}=l\right)&=\P_z\left(\tau^{(Z,K)}_0<\tau^{(Z,K)}_{\lfloor\eps K\rfloor}\right)\left(1-\P_1\left(\tau^{(Z,K)}_{\lfloor\eps K\rfloor}<\tau^{(Z,K)}_0\right)\right)^{l-1}\P_1\left(\tau^{(Z,K)}_{\lfloor\eps K\rfloor}<\tau^{(Z,K)}_0\right)\nonumber\\
&\leq C_1e^{-C_2\eps^2K}.
\end{align}
For $z=0$ and $l\geq0$,
\begin{align}
	\P_0\left(B^{K}=l\right)=\P_1\left(B^{K}=l\right)\leq C_1e^{-C_2\eps^2K}.
\end{align}
Hence, for $0\leq z\leq\frac{1}{2}\left\lfloor\frac{\eps K}{2}\right\rfloor$, $\eps$ small and $K$ large enough, and any $m\geq0$,
\begin{align}
	\P_z\left(B^{K}\leq m\right)\leq(m+1)C_1e^{-C_2\eps^2K}.
\end{align}

\textbf{Step 5:} Finally, let $\tilde{Z}^{K}$ be the continuous time process that has $(Z^{K}_n)_{n\in\N_0}$ as a jump chain and the same jump times $(S^K_n)_{n\in\N}$ as the original process $X^{K}$. By the above construction, for all $t\geq0$, we have
\begin{align}
\tilde{Z}^{K}(t)\geq V^{K}(t)=|X^{K}(t)-\lceil\bar{n}K\rceil|.
\end{align}
We can therefore deduce that, for $\eps$ small and $K$ large enough, initial value $x=\lceil\bar{n}K\rceil\pm z$ such that $0\leq z\leq\frac{1}{2}\left\lfloor\frac{\eps K}{2}\right\rfloor$, and any $m\geq0$,
\begin{align}
\P_{x}&(\exists\ t\in[0,\theta_K]:|X^{K}(t)-\lceil \bar{n}K\rceil|>\eps K)\nonumber\\
&\leq\P_z(\exists\ t\in[0,\theta_K]:\tilde{Z}^{K}(t)>\eps K)\nonumber\\
&=\sum_{l=0}^\infty\P_z(\exists\ t\in[0,\theta_K]:\tilde{Z}^{K}(t)>\eps K,B^{K}=l)\nonumber\\
&\leq \P_z(B^{K}\leq m-1)+\sum_{l=m}^\infty\P_z(\exists\ t\in[0,\theta_K]:\tilde{Z}^{K}(t)>\eps K,B^{K}=l)
\end{align}

Now let $(I^{K}_j)_{j\geq1}$ be the times in between visits to 0 of $\tilde{Z}^{K}$, i.e., for $j\geq1$,
\begin{align}
	I^{K}_j:=\inf\left\{t>0: \exists\ 0<s<t: \tilde{Z}^{K}_{s+\sum_{i=1}^{j-1}I^{K}_i}\neq0, \tilde{Z}^{K}_{t+\sum_{i=1}^{j-1}I^{K}_i}=0\right\}.
\end{align}
Then, since each return takes at least the time of a single jump in the original Markov chain $X^{K}$, as long as $\tilde{Z}^{K}$ does not surpass $2\bar{n}K>\lceil\bar{n}K\rceil+\eps K$, there are independent identically distributed exponential random variables $(E^{K}_j)_{j\geq1}$ with parameter $4\bar{n}K(b+d+c\bar{n})=:\bar{C}K$ such that, for each $a\in\R$,
\begin{align}
	\P(I^{K}_j<a)\leq\P(E^{K}_j<a)=(1-e^{-a\bar{C}K})_+.
\end{align}

To bound the probabilities $\P_z(\exists\ t\in[0,\theta_K]:\tilde{Z}^{K}(t)>\eps K,B^{K}=l)$, we argue as follows: If there were at least $l/2$ occurrences of $I^{K}_j\geq 2\theta_K/l$ (out of the $l$ times between visits to 0), the time until $\tilde{Z}^{K}$ first surpasses $\eps K$ could be bounded from below by
\begin{align}
	\sum_{j=1}^lI^{K}_j\geq\frac{l}{2}\frac{2\theta_K}{l}=\theta_K.
\end{align}
Hence, by contradiction we can bound
\begin{align}
\P_z&(\exists\ t\in[0,\theta_K]:\tilde{Z}^{K}(t)>\eps K,B^{K}=l)
\leq \P_z\left(\sum_{j=1}^l\ifct{I^{K}_j<2\theta_K/l}>\frac{l}{2},B^{K}=l\right)\nonumber\\
&\leq \P_z\Bigg(\underbrace{\sum_{j=1}^l\ifct{E^{K}_j<2\theta_K/l}}_{\sim\text{Bin}\left(l,1-e^{-\bar{C}K2\theta_K/l}\right)}>\frac{l}{2}\Bigg)
= \sum_{j=\lceil l/2\rceil}^l {l\choose j}\left(1-e^{-\bar{C}K2\theta_K/l}\right)^j\left(e^{-\bar{C}K2\theta_K/l}\right)^{l-j}\nonumber\\
&\leq \frac{l}{2}2^l\left(1-e^{-\bar{C}K2\theta_K/l}\right)^{l/2}
\leq \left(4\left(1-e^{-\bar{C}K2\theta_K/l}\right)^{1/2}\right)^l,
\end{align}
where we used that ${l\choose j}\leq2^l$ and $l/2\leq 2^l$. 

Combining this with step 4 yields
\begin{align}
\P_{x}&(\exists\ t\in[0,\theta_K]:|X^{K}(t)-\lceil \bar{n}K\rceil|>\eps K)\nonumber\\
&\leq \P_z(B^{K}\leq m-1)+\sum_{l=m}^\infty\P_z(\exists\ t\in[0,\theta_K]:\tilde{Z}^{K}(t)>\eps K,B^{K}=l)\nonumber\\
&\leq mC_1e^{-C_2\eps^2K}+\sum_{l=m}^\infty\left(4\left(1-e^{-\bar{C}K2\theta_K/l}\right)^{1/2}\right)^l,
\end{align}
This concludes the proof with $C_3=2\bar{C}$.
\end{proof}

From this general theorem, we can now derive the result necessary for the proof in Section \ref{sec:4.1_proof_stability}, considering time horizons of size $\theta_K=\ln K$ as an upper bound for phases with length of order $\lambda_K$ and bounding the probability of diverging from the equilibrium in $o(\lambda_K/\ln K)$ to concatenate $O(\ln K/\lambda_K)$ many phases.

\begin{theorem}\label{stability_lnK}
	Let $X^{K}$ be a birth death process with self-competition and parameters\linebreak$0<d<b$ and $0<c/K$. Define $\bar{n}:=(b-d)/c$ and let $1\ll\lambda_K\ll\ln K$ as $K\to\infty$. Then, for $\eps$ small enough and any sequence of initial values $0\leq |x^{K}-\lceil\bar{n}K\rceil|\leq\frac{1}{2}\left\lfloor\frac{\eps K}{2}\right\rfloor$,
	\begin{align}
		\lim_{K\to\infty}\frac{\ln K}{\lambda_K}\cdot\P_{x^{K}}(\exists\ t\in[0,\ln K]:|X^{K}(t)-\lceil \bar{n}K\rceil|>\eps K)=0.
	\end{align}
\end{theorem}

\begin{proof}
We apply Theorem \ref{stability_general} with $\theta_K=\ln K$ and $m^{K}=K^2$ to obtain that, for $\eps$ small, $K$ large enough and $0\leq |x^{K}-\lceil\bar{n}K\rceil|\leq\frac{1}{2}\left\lfloor\frac{\eps K}{2}\right\rfloor$, 
\begin{align}
\P_{x^{K}}&(\exists\ t\in[0,\ln K]:|X^{K}(t)-\lceil \bar{n}K\rceil|>\eps K)\nonumber\\
&\leq K^2C_1e^{-C_2\eps^2K}+\sum_{l=K^2}^\infty\left(4\left(1-e^{-C_3K\ln K/l}\right)^{1/2}\right)^l\nonumber\\
&\leq K^2C_1e^{-C_2\eps^2K}+\sum_{l=K^2}^\infty\bigg(\underbrace{4\left(1-e^{-C_3K\ln K/K^2}\right)^{1/2}}_{<1\text{ for $K$ large enough}}\bigg)^l\nonumber\\
&\leq K^2C_1e^{-C_2\eps^2K}+\left(4^2\left(1-e^{-C_3K\ln K/K^2}\right)\right)^{K^2/2}\underbrace{\frac{1}{1-4\left(1-e^{-C_3K\ln K/K^2}\right)^{1/2}}}_{\leq C_4<\infty\text{ for $K$ large enough}}\nonumber\\
&\leq K^2C_1e^{-C_2\eps^2K}+C_4\bigg(\underbrace{16C_3\frac{\ln K}{K}}_{<1\text{ for $K$ large enough}}\bigg)^{K^2/2}\nonumber\\
&\leq K^2C_1e^{-C_2\eps^2K}+C_416C_3\frac{\ln K}{K}.
\end{align}

Now, for fixed $\eps^2$,
\begin{align}
	e^{-C_2\eps^2K}\ll\frac{\lambda_K}{K^2\ln K}\ \Leftrightarrow\ K^2e^{-C_2\eps^2K}\ll\frac{\lambda_K}{\ln K}
\end{align}
and $1\ll\lambda_K$ implies
\begin{align}
	\frac{(\ln K)^2}{K}\ll\lambda_K\ \Leftrightarrow\ \frac{\ln K}{K}\ll\frac{\lambda_K}{\ln K}.
\end{align}
Hence we obtain that, for fixed $\eps>0$,
\begin{align}
\P_{x^{K}}&(\exists\ t\in[0,\ln K]:|X^{K}(t)-\lceil \bar{n}K\rceil|>\eps K)=O(K^2e^{-C_2\eps^2K})+O(\ln K/K)=o(\lambda_K/\ln K),
\end{align}
which yields the desired result.
\end{proof}

\subsection{Convergence to the deterministic system}

We now provide a result on the convergence of stochastic birth death processes with competition to the corresponding deterministic system, for finite time horizons. The proof is similar to the one of \cite[Ch.\ 11, Thm.\ 2.1]{EtKu86}. Instead of almost sure convergence, we derive convergence in probability, but are able to quantify the convergence speed in return, which again allows us to concatenate the result for infinitely many phases in Section \ref{sec:4.1_proof_stability}.

\begin{theorem}\label{EthKu_variation}
Let $X^{K}$ be a birth death process with self-competition and parameters\linebreak $0<d<b$ and $0<c/K$. Assume that $X^{K}(0)/K\to x_0$ and $1\ll\lambda_K\ll\ln K$, as $K\to\infty$, and let $(x(t))_{t\geq0}$ be the solution of the ordinary differential equation
\begin{align}\label{ODEresident}
\dot{x}(t)=x(t)\left[b-d-c\cdot x(t)\right]
\end{align}
with initial value $x(0)=x_0$. Then, for every $0\leq T<\infty$ and $\eps>0$,
\begin{align}
	\lim_{K\to\infty}\frac{\ln K}{\lambda_K}\cdot\P\left(\sup_{t\leq T}\left|\frac{X^{K}(t)}{K}-x(t)\right|>\eps\right)=0.
\end{align}
\end{theorem}

\begin{proof}
We start be writing $X^{K}$ in terms of independent standard Poisson processes $Y_b$ and $Y_d$,
\begin{align}
X^{K}(t)\overset{(d)}{=}X^{K}(0)+Y_b\left(K\int_0^t b\frac{X^{K}(s)}{K}ds\right)-Y_d\left(K\int_0^t d\frac{X^{K}(s)}{K}+c\left(\frac{X^{K}(s)}{K}\right)^2ds\right)
\end{align}
Note that we only have equality in distribution here, since we choose $Y_b$ and $Y_d$ uniformly across different values of $K$ and the $X^{K}$ stand in no specific relation to each other. We will omit this from the notation for the remainder of the proof, as we are only proving convergence in probability and equality in distribution is therefore sufficient.

Setting $\tilde{Y}_b(u):=Y_b(u)-u$ and $\tilde{Y}_d:=Y_d(u)-u$ (i.e.\ centering the Poisson processes at their expectations), this representation yields
\begin{align}
\frac{X^{K}(t)}{K}=&\ \frac{X^{K}(0)}{K}+\int_0^t\left(b-d-c\frac{X^{K}(s)}{K}\right)\frac{X^{K}(s)}{K}ds\nonumber\\
&\ +\frac{1}{K}\tilde{Y}_b\left(K\int_0^t b\frac{X^{K}(s)}{K}ds\right)-\frac{1}{K}\tilde{Y}_d\left(K\int_0^t d\frac{X^{K}(s)}{K}+c\left(\frac{X^{K}(s)}{K}\right)^2ds\right).
\end{align}

Now we introduce the stopping time
\begin{align}
	\tau^K:=\inf\left\{t\geq0:\frac{X^{K}(t)}{K}\neq[0,M]\right\},
\end{align}
for some large $M$ (e.g. $M\geq 2(\bar{n}\lor x_0)$, where $\bar{n}=(b-d)/c$). Up to time $\tau^K$, the population size of our process is bounded by $KM$ and, using the integral form of \eqref{ODEresident}, we deduce
\begin{align}
\left|\frac{X^{K}(t\land\tau^K)}{K}-x(t\land\tau^K)\right|\leq\ &\left|\frac{X^{K}(0)}{K}-x(0)\right|\nonumber\\
&\ +\int_0^{t\land\tau^K}(b+d)\left|\frac{X^{K}(s)}{K}-x(s)\right|+c\underbrace{\left|\left(\frac{X^{K}(s)}{K}\right)^2-x(s)^2\right|}_{\leq 2M\left|\frac{X^{K}(s)}{K}-x(s)\right|}ds\nonumber\\
&\ +\frac{1}{K}\left|\tilde{Y}_b\left(K\int_0^{t\land\tau^K} b\frac{X^{K}(s)}{K}ds\right)\right|\nonumber\\
&\ +\frac{1}{K}\left|\tilde{Y}_d\left(K\int_0^{t\land\tau^K} d\frac{X^{K}(s)}{K}+c\left(\frac{X^{K}(s)}{K}\right)^2ds\right)\right|\nonumber\\
\leq&\ \left|\frac{X^{K}(0)}{K}-x(0)\right|+\int_0^{t}C\left|\frac{X^{K}(s\land\tau^K)}{K}-x(s\land\tau^K)\right|ds\nonumber\\
&\ +\frac{1}{K}\sup_{u\in[0,t\land\tau^K]}\left|\tilde{Y}_b\left(K bMu\right)\right|+\frac{1}{K}\sup_{u\in[0,t\land\tau^K]}\left|\tilde{Y}_d\left(K(dMu+cM^2u)\right)\right|.
\end{align}
Here, $C:=b+d+2Mc$ and in the last line we used that, even though the centred Poisson processes can take positive and negative values, we can bound their absolute value by considering the supremum over all possible rates.

Next, Gronwall's inequality implies that
\begin{align}
&\left|\frac{X^{K}(t\land\tau^K)}{K}-x(t\land\tau^K)\right|\nonumber\\
&\leq \left[ \left|\frac{X^{K}(0)}{K}-x(0)\right|+\frac{1}{K}\left(\sup_{u\in[0,t]}\left|\tilde{Y}_b\left(K bMu\right)\right|+\sup_{u\in[0,t]}\left|\tilde{Y}_d\left(K(dMu+cM^2u)\right)\right|\right)\right]e^{Ct}.
\end{align}

Now fix a $T\geq 0$. With probability 1, for some large enough $K_0$ and all $K\geq K_0$,\linebreak$|X^{K}(0)-x(0)|\leq e^{-CT}\eps/2$. Hence, for $\eps<M/2$ and $K\geq K_0$,
\begin{align}
&\P\left(\sup_{t\leq T}\left|\frac{X^{K}(t)}{K}-x(t)\right|\leq\eps\right)\nonumber\\
&=\P\left(\sup_{t\leq T}\left|\frac{X^{K}(t)}{K}-x(t)\right|\leq\eps,\tau^K>T\right)\nonumber\\
&= \P\left(\sup_{t\leq T}\left|\frac{X^{K}(t\land\tau^K)}{K}-x(t\land\tau^K)\right|\leq\eps\right)\nonumber\\
&\geq \P\left( \left[ \left|\frac{X^{K}(0)}{K}-x(0)\right|+\frac{1}{K}\left(\sup_{u\in[0,T]}\left|\tilde{Y}_b\left(K bMu\right)\right|+\sup_{u\in[0,T]}\left|\tilde{Y}_d\left(K(dMu+cM^2u)\right)\right|\right)\right]e^{CT}\leq\eps\right)\nonumber\\
&\geq \P\left(\frac{1}{K}\left(\sup_{u\in[0,T]}\left|\tilde{Y}_b\left(K bMu\right)\right|+\sup_{u\in[0,T]}\left|\tilde{Y}_d\left(K(dMu+cM^2u)\right)\right|\right)e^{CT}\leq\frac{\eps}{2}\right)\nonumber\\
&\geq 1-\P\left(\sup_{u\in[0,T]}\left|\tilde{Y}_b\left(K bMu\right)\right|+\sup_{u\in[0,T]}\left|\tilde{Y}_d\left(K(dMu+cM^2u)\right)\right|>\frac{Ke^{-CT}\eps}{2}\right)\nonumber\\
&\geq 1-\P\left(\sup_{u\in[0,T]}\left|\tilde{Y}_b\left(K bMu\right)\right|>\frac{Ke^{-CT}\eps}{4}\right)-\P\left(\sup_{u\in[0,T]}\left|\tilde{Y}_d\left(K(dMu+cM^2u)\right)\right|>\frac{Ke^{-CT}\eps}{4}\right).
\end{align}
To finish the proof, it is left to show that the two probabilities in the last line are of order $o(\lambda_K/\ln K)$. We run through the calculation for $\tilde{Y}_b$, the other summand works equivalently.

Since $\tilde{Y}_b$ is a martingale, $|\tilde{Y}_b|$ is a submartingale. We set $\tilde{T}=bMT$ and $\tilde{\eps}=e^{-CT}\eps/4$. Then, using Doob's maximum inequality \cite[Thm.\ 3.87]{BovHol15}, we obtain
\begin{align}
&\P\left(\sup_{u\in[0,T]}\left|\tilde{Y}_b\left(K bMu\right)\right|>\frac{Ke^{-CT}\eps}{4}\right)
=\P\left(\sup_{u\in[0,\tilde{T}]}\left|\tilde{Y}_b(Ku)\right|^2>(K\tilde{\eps})^2\right)\nonumber\\
&\leq\frac{1}{(K\tilde{\eps})^2}\E\left(\left|\tilde{Y}_b(K\tilde{T})\right|^2\right)
=\frac{1}{(K\tilde{\eps})^2}\E\left(\langle\tilde{Y}_b(K\tilde{T})\rangle\right)\nonumber\\
&=\frac{K\tilde{T}}{(K\tilde{\eps})^2}=\frac{\tilde{T}}{K\tilde{\eps}^2}=o\left(\frac{\lambda_K}{\ln K}\right).
\end{align}
This concludes the proof.
\end{proof}


\section{Branching processes at varying rates}
\label{app:B_branching_processes}

	In this chapter, we collect some technical results for birth death processes with time-dependent rates. These are used to approximate the micro- and mesoscopic populations in the proof of the main result of this paper in Section \ref{sec:4.2_convergence_beta}. In Section \ref{app:B.1_pureBDproc} we focus on pure birth death processes and then add on the effects of immigration in Section \ref{app:B.2_BDIproc}. We build on the results of \cite{ChMeTr19} and work out the averaging effects of growth in a periodically changing environment.
	
	The particular form of time-dependent parameters in this chapter depicts two different effects. Firstly, the system parameters jump periodically between finitely many values on the divergent time-scale $\lK$. Secondly, after each parameter change, the macroscopic subpopulation restabilizes at the corresponding equilibrium, which takes a finite time (independent of $K$). During this short re-equilibration time, we only have weaker estimates on the effective parameters for the growth of the micro- and mesoscopic subpopulations. However, part of the following results is that the general behaviour of the branching processes is not effected significantly by these short phases.

\subsection{Pure birth death processes}
\label{app:B.1_pureBDproc}

	Let us first consider pure birth death processes with time-dependent rates.
		As before, take $1\ll\lK\ll\ln K$ and $\ell\in\N$. Let $b_i,d_i,b_{i,*},d_{i,*},>0$, $T_i>0$ and $T_{*,i}\geq 0$, for $i\in\dset{1,\ldots,\ell}$. Writing $T^\Sigma_i:=\sum_{j=1}^iT_j$, we define the rate functions for birth and death to be the periodic extensions of
	\begin{align}
		\begin{aligned}
			b^K(t)&=\left\{\begin{array}{ll}
				b_{*,i}&:t\in[T^\S_{i-1}\lK,T^\S_{i-1}\lK+T_{*,i})\\
				b_i&:t\in[T^\S_{i-1}\lK+T_{*,i},T^\S_{i}\lK),\\
			\end{array}\right.\\
			d^K(t)&=\left\{\begin{array}{ll}
				d_{*,i}&:t\in[T^\S_{i-1}\lK,T^\S_{i-1}\lK+T_{*,i})\\
				d_i&:t\in[T^\S_{i-1}\lK+T_{*,i},T^\S_{i}\lK),\\
			\end{array}\right.
		\end{aligned}
	\end{align}

	Moreover, for $i\in\dset{1,\ldots,\ell}$, we write $r_i:=b_i-d_i$ and $r^K(t):=b^K(t)-d^K(t)$ to refer to the net growth rate. Finally, we define the average growth rate by $r_\av:=(\sum_{i=1}^{\ell}r_iT_i)/T^\S_\ell$.

	We analyse the processes $\left(Z^K_t\right)_{t\geq 0}$, which are Markov processes with $Z^K_0=\gauss{K^\b-1}$, for some $\b\geq 0$, and with generators
	\begin{align}
		\left(\mathcal{L}^K_tf\right)(n)=b^K(t)n\left(f(n+1)-f(n)\right)+d^K(t)n\left(f(n-1)-f(n)\right),
	\end{align}
	acting on all bounded functions $f:\N_0\to\R$. We refer to the law of these processes by $Z^K\sim\mathrm{BD}\left(b^K,d^K,\beta\right)$.
	
	Our aim is to show that, under logarithmic rescaling of time and size, such population processes grow (or shrink) according to their average net growth rate. Note that the process becomes trivial if $\b=0$. We therefore exclude this case in the entire section without further announcement.
	
	\begin{theorem}
		\label{thm:bd_main}
		Let $Z^K$ follow the law of $\mathrm{BD}\left(b^K,d^K,\beta\right)$, where $\b>0$. Then, for all fixed $T\in(0,\infty)$, the following convergence holds in probability, with respect to the $L^\infty([0,T])$ norm,
		\begin{align}
			\left(\frac{\ln\left(Z^K_{s\ln K}+1\right)}{\ln K}\right)_{s\in[0,T]}
			\overset{K\to\infty}{\longrightarrow}
			\left((\b+r_\av s)\vee 0\right)_{s\in[0,T]}.
		\end{align}
	\end{theorem}
	
	The rest of this section is dedicated to the proof of this theorem and we split up the claim into several lemmas.
	
	\begin{remark}
		To avoid complicated notation, we only conduct the proofs for the case of $\ell=2$ and $T_{*,1}=0$. The general case is proven analogously and there is no deeper insight or additional difficulty to it. This choice allows us to use the shorthand notation $b_*:=b_{*,2}$, $ d_*:=d_{*,2}$, and $T_*:=T_{*,2}$, which leads to the rate functions taking the form
		\begin{align}
			b^K(t)=\left\{\begin{array}{ll}
			b_1&t\in[0,\lK T_1)\\
			b_*&t\in[\lK T_1,\lK T_1+T_*)\\
			b_2&t\in[\lK T_1+T_*,\lK (T_1+T_2)),\\
			\end{array}\right.
			&d^K(t)=\left\{\begin{array}{ll}
			d_1&t\in[0,\lK T_1)\\
			d_*&t\in[\lK T_1,\lK T_1+T_*)\\
			d_2&t\in[\lK T_1+T_*,\lK (T_1+T_2)),\\
			\end{array}\right.
		\end{align}
		with periodic extension.
	\end{remark}
	
	We start by stating an explicit representation of the processes in terms of Poisson measures and derive the corresponding Doob's martingale decomposition.
	
	\textbf{Poisson representation:} Let $Q^{(b)}(\dd s,\dd\theta)$ and $Q^{(d)}(\dd s,\dd\theta)$ be independent homogenous Poisson random measures on $(\R^2_{\geq 0},\dd s ,\dd\theta)$ and denote by $\tilde{Q}^{(*)}=Q^{(*)}-\dd s\dd\theta$, for $*\in\{b,d\}$, their normalized versions. Then we can represent $Z^K$ as
	\begin{align}
		Z^K_t=Z^K_0
			+\int_0^t\int_{\R_{\geq 0}}\ifct{\theta\leq b^K(s^-)Z^K_{s^-}}Q^{(b)}(\dd s,\dd\theta)
			-\int_0^t\int_{\R_{\geq 0}}\ifct{\theta\leq d^K(s^-)Z^K_{s^-}}Q^{(d)}(\dd s,\dd\theta).
	\end{align}

	In particular we have the martingale decomposition $Z^K_t=Z^K_0+M^K_t+A^K_t$, where
	\begin{align}
		M^K_t=
			\int_0^t\int_{\R_{\geq 0}}\ifct{\theta\leq b^K(s^-)Z^K_{s^-}}\tilde{Q}^{(b)}(\dd s,\dd\theta)
			-\int_0^t\int_{\R_{\geq 0}}\ifct{\theta\leq d^K(s^-)Z^K_{s^-}}\tilde{Q}^{(d)}(\dd s,\dd\theta).
	\end{align}
	and
	\begin{align}
		\label{eq:bd_mrtgl_corr}
		A^K_t=\int_0^t(b^K(s)-d^K(s))Z^K_s\dd s.
	\end{align}
	In terms of Itô's calculus this impliess $\dd Z^K_t=\dd M^K_t+r^K(t)Z^K_t\dd t$. Therefore, we directly obtain the bracket process
	\begin{align}\label{eq:bd_bracket}
		\langle M^K\rangle_t=\int_0^t (b^K(s)+d^K(s))Z^K_s\dd s.
	\end{align}
	
	Towards proving Theorem \ref{thm:bd_main}, we first determine the expected value of the process and check that it satisfies the desired convergence. 
	
	\begin{lemma}
		\label{lem:bd_expectation_formula}
		Let $Z^K$ follow the law of $\mathrm{BD}\left(b^K,d^K,\b\right)$. Then
		\begin{align}
			\Exd{Z^K_t}=\gauss{K^\b-1}\ee^{g^K(t)},\qquad\text{where}\ g^K(t)=\int_0^t r^K(s)\dd s.
		\end{align}
	\end{lemma}
	\begin{proof}
		Using the martingale decomposition and \eqref{eq:bd_mrtgl_corr} we obtain the integral equation
		\begin{align}
			\Exd{Z^K_t}
			&=\Exd{Z^K_0}+\Exd{A^K_t}=\gauss{K^\b-1}+\int_0^t r^K(s)\Exd{Z^K_s}\dd s.
		\end{align}
		Due to existence and uniqueness of the solution to this integral equation, this directly gives the claim.
	\end{proof}
	
	\begin{lemma}
		\label{lem:bd_conv_expectation}
		For all fixed $T<\infty$ we have the uniform convergence
		\begin{align}
			\sup_{s\leq T}\abs{\frac{\ln\left(K^\b\ee^{g^K(s\ln K)}\right)}{\ln K} -\left(\b+r_\av s\right)}
			\overset{K\to\infty}{\longrightarrow}0.
		\end{align}
		
	\end{lemma}
	
	\begin{proof}
		We slice the whole time span $[0,T\ln K]$ into equal pieces of order $\lK$. Since $r^K$ has a period of length $\lK(T_1+T_2)$ and is piecewise constant, we obtain, for all $n\in\N$,
		\begin{align}
			g^K(n\lK(T_1+T_2))
			&=n[\lK T_1r_1 +T_*r_* +(\lK T_2-T_*)r_2]\nonumber\\
			&=n[\lK (T_1r_1 +T_2r_2) +T_*(r_*-r_2)].
		\end{align}
		Moreover, within such a period the growth of $g^K$ is linearly bounded, i.e.\ for all $0\leq s\leq t$ such that $t-s\leq\lK(T_1+T_2)$, one has
		\begin{align}
			\abs{g^K(t)-g^K(s)}\leq(t-s)\max\dset{r_1,r_*,r_2}\leq C\lK,
		\end{align}
		for some uniform finite constant $C$.
		Hence, we can estimate (in the case of $r_\av>0$),
		\begin{align}
			g^K(s\ln K)
			&\geq g^K\left(\gauss{\frac{s\ln K}{\lK(T_1+T_2)}}\lK(T_1+T_2)\right)-C\lK\nonumber\\
			&= \gauss{\frac{s\ln K}{\lK(T_1+T_2)}}
				[\lK (T_1r_1+T_2r_2)+T_*(r_*-r_2)]
				-C\lK\nonumber\\
			&\geq r_\av s\ln K 
				+s\frac{\ln K}{\lK}\frac{T_*(r_*-r_2)}{T_1+T_2}
				-2C\lK,
		\end{align}
		where we use in the last inequality that the term in the brackets is positive for $K$ large enough. Similarly, we obtain
		\begin{align}
			g^K(s\ln K)\leq r_\av s\ln K 
			+s\frac{\ln K}{\lK}\frac{T_*(r_*-r_2)}{T_1+T_2}
			+2C\lK.
		\end{align}
		Both estimates can be achieved in the same way for $r_\av<0$. With this at hand, the claim can be shown directly.
	\end{proof}
	
	Our next aim is to study the deviation of the original process $Z^K$ from its expected value.
	
	\begin{lemma}
		\label{lem:bd_Doob}
		For all $T<\infty$ and all $\d>0$, there exists a $K_0\in\N$ such that, for all $K\geq K_0$ and all $\h\in(0,\b)$, it holds that
		\begin{align}
			\Prob{\sup_{t\leq T\ln K} \abs{\ee^{-g^K(t)}Z^K_t-K^\b}\geq K^\h} \leq (\bar{b}+\bar{d})K^{\b+\d-2\h}\frac{1}{r_\av}\left(1-K^{-r_\av T}\right),
		\end{align}
		where $\bar{b}:=\max_{t\in[0,T^\S_\ell\lambda_K]}b^K(t)$ and $\bar{d}:=\max_{t\in[0,T^\S_\ell\lambda_K]}d^K(t)$.
	\end{lemma}
	\begin{proof}
		The main idea is to use Doob's maximum inequality for a rescaled martingale. We introduce the process
		\begin{align}
			\hat{M}^K_t:=\int_0^t\ee^{-g^K(s)}\dd M^K_s,
		\end{align}
		which is a martingale since $M^K$ is a martingale.
		Following the techniques of proof step 1 of \cite[Lem.\ A.1]{ChMeTr19} using It\^o's isometry, It\^o's formula and Doob's maximum inequality, we get
		\begin{align}
			\Prob{\sup_{t\leq T\ln K} \abs{\ee^{-g^K(t)}Z^K_t-K^\b}\geq K^\h}
			&\leq \left(\bar{b}+\bar{d}\right) K^{\b-2\h} \int_0^{T\ln K}\ee^{-g^K(t)}\dd t\nonumber\\
			&\leq \left(\bar{b}+\bar{d}\right)K^{\b-2\h} \int_0^{T\ln K}\ee^{-r_\av t} K^\d\dd t\nonumber\\
			&=\left(\bar{b}+\bar{d}\right)K^{\b+\d-2\h}\frac{1}{r_\av}\left(1-K^{-r_\av T}\right).
		\end{align}
		Here, for the last inequality, $K$ has to be choosen large enough, such that
		\begin{align}
			\frac{g^K(t)}{\ln K}\geq r_\av\frac{t}{\ln K}-\d,\qquad\forall t\in[0,T\ln K],
		\end{align}
		which is possible by Lemma \ref{lem:bd_conv_expectation}.
	\end{proof}

	With the above results, we are now able to prove the desired convergence for populations that tend to grow (i.e.\ for $r_\av>0$).

	\begin{lemma}
		\label{lem:bd_conv_beta_r_pos}
		Let $Z^K$ follow the law of $\mathrm{BD}\left(b^K,d^K,\beta\right)$ and assume that $r_\av>0$, then the convergence of Theorem \ref{thm:bd_main} holds true.
	\end{lemma}
	\begin{proof}
		Fix $T<\infty$ and choose $\d=\b/4$ and $\h=3\b/4$. Moreover define the set
		\begin{align}
			\O^K_1:=\dset{\sup_{t\leq T\ln K} \abs{\ee^{-g^K(t)}Z^K_t-K^\b}\leq K^\h}.
		\end{align}
		Then $\lim_{K\to\infty}\Prob{\O^K_1}=1$ by Lemma \ref{lem:bd_Doob} since $\b+\d-2\h<0$ and an analogous computation to proof step 2 of \cite[Lem.\ A.1]{ChMeTr19} together with Lemma \ref{lem:bd_conv_expectation} and the triangle inequality yield the claim.
	\end{proof}
	
	In order to study birth death processes with tendency to shrink (i.e.\ with $r_\av<0$), we have to take care of the extinction event. Let us point out that in our situation of the changing environment, this is a little more involved and we cannot use the results for time-homogenoeus branching processes. To this end, we first determine the probability generating function of general birth death processes with piecewise constant rates, which we then use to establish bounds on the distribution function of the extinction time. 

	\begin{lemma}[generating function]
		\label{lem:bd_generating_function}
		For $\ell\in\N$, let $b_i,d_i,T_i>0$, $1\leq i\leq\ell$, and write \linebreak$T:=\sum_{i=1}^{\ell}T_i$ and $r_i:=b_i-d_i$. We consider the birth death process $\left(X_t\right)_{t\in[0,T]}$ with initial value $X_0=1$ that is driven by the birth and death rates $b_i,d_i$ on $\left[\sum_{j=1}^{i-1}T_j,\sum_{j=1}^iT_j\right)$. Then the (probability) generating function $g$ of $X_T$ is given by
		\begin{align}
			g(s):&=\Exd{s^{X_T}}\\
			&=1-\frac{\ee^{r_1T_1+\cdots+r_\ell T_\ell}}{\frac{b_1}{r_1}\left(\ee^{r_1T_1}-1\right)\ee^{r_2T_2+\cdots+r_\ell T_\ell}+\frac{b_2}{r_2}\left(\ee^{r_2T_2}-1\right)\ee^{r_3T_3+\cdots+r_\ell T_\ell}+\cdots+\frac{b_\ell}{r_\ell}\left(\ee^{r_\ell T_\ell}-1\right)-\frac{1}{s-1}}.\nonumber
		\end{align}
	\end{lemma}
	
	\begin{proof}
		For homogeneous birth death processes $\left(Y_t\right)_{t\geq 0}$ with constant rates $b$ and $d$ and initial value $Y_0=1$, the probability generating function at time $t>0$ is given by
		\begin{align}
			F(s,t):=\Exd{s^{Y_t}\big|Y_0=1}
			&=\frac{d(s-1)-\ee^{-rt}(bs-d)}{b(s-1)-\ee^{-rt}(bs-d)}\nonumber\\
			&=1-\frac{\ee^{rt}}{\frac{b}{r}\left(\ee^{rt}-1\right)-\frac{1}{s-1}}.
		\end{align}
		Due to independence, for initial values $Y_0=k\in\N$ we obtain $\Exd{s^{Y_t}\big|Y_0=k}=F(s,t)^k$.
		
		Let us first consider the case $\ell=2$. Studying $X$ at time $T_1+T_2$, we can interpret it as a birth death process on $[T_1,T_1+T_2]$ with rates $b_2$ and $d_2$, initialized with $X_{T_1}$ individuals. Letting $F_1$ and $F_2$ be the generating functions corresponding to the parameters of the two phases (again assuming initial values of 1), this leads to
		\begin{align}
			\Exd{s^{X_{T_1+T_2}}}
			&=\Exd{\Exd{s^{X_{T_1+T_2}}\big|X_{T_1}}}
			=\Exd{F_2(s,T_2)^{X_{T_1}}}
			=F_1(F_2(s,T_2),T_1)\nonumber\\
			&=1-\frac{\ee^{r_1T_1}}{\frac{b_1}{r_1}\left(\ee^{r_1T_1}-1\right)+\frac{(b_2/r_2)(\ee^{r_2T_2}-1)-1/(s-1)}{\ee^{r_2T_2}}}\nonumber\\
			&=1-\frac{\ee^{r_1T_1+r_2T_2}}{\frac{b_1}{r_1}\left(\ee^{r_1T_1}-1\right)\ee^{r_2T_2}+\frac{b_2}{r_2}\left(\ee^{r_2T_2}-1\right)-\frac{1}{s-1}}
		\end{align}
		The claim for larger $\ell\in\N$ follows by induction.
	\end{proof}
	
	With this preparation, we can now prove the following helpful bounds for the extinction time.
	
	\begin{lemma}[extinction time]
		\label{lem:bd_ext_time}
		Let $Z^K$ be the birth death process with varying rates defined above. Denote by $T^{Z^K}_\ext:=\inf\dset{t\geq 0:Z^K_t=0}$ the extinction time of $Z^K$. Then, for all\linebreak $\d,\d_1,\d_2>0$ and all $M>(T_1\abs{r_1}+T_2\abs{r_2})/(T_1+T_2)$ there exists a $K_0<\infty$ such that, for all $K\geq K_0$ and all $t\geq 0$,
		\begin{align}
			\label{eq:bd_ext_upper}
			&\Prob{T^{Z^K}_\ext>t\ln K\big|Z^K_0=1} \leq\exp\left[\left(\left(r_\av+\d_1\right)t+\d_2\right)\ln K\right],\\
			\label{eq:bd_ext_lower}
			&\Prob{T^{Z^K}_\ext>t\ln K\big|Z^K_0=1} \geq\exp\left[-\left(Mt+\d\right)\ln K\right].
		\end{align}
	\end{lemma}
	
	\begin{proof}
		Let us first consider a time discretisation $\left(Y^K_n\right)_{n\in\N_0}$ of $Z^K$, namely $Y^K_n:=Z^K_{n\lK(T_1+T_2)}$. Then, by periodicity of the rate functions, $Y^K$ is a Galton-Watson process. An application of Lemma \ref{lem:bd_generating_function} (with $\ell=3$) yields that the one-step offspring distribution is determined by the generating function
		\begin{align}
			g(s)&=\Exd{s^{Z^K_{\lK(T_1+T_2)}}|Z^K_0=1}\\
			&=1-\frac{\ee^{\lK T_1r_1+T_*r_*+(\lK T_2-T_*)r_2}}{\frac{b_1}{r_1}\left(\ee^{\lK r_1T_1}-1\right)\ee^{r_*T_*+(\lK T_2-T_*)r_2}+\frac{b_*}{r_*}\left(\ee^{r_*T_*}-1\right)\ee^{(\lK T_2-T_*)r_2}+\frac{b_2}{r_2}\left(\ee^{(\lK T_2-T_*)r_2}-1\right)-\frac{1}{s-1}}.\nonumber
		\end{align}
		Therefore, the probability for $Z^K$ to go extinct in the first step is
		\begin{align}
			p^K_0&=g(0)\\
			&=1-\frac{\ee^{\lK T_1r_1+T_*r_*+(\lK T_2-T_*)r_2}}{\frac{b_1}{r_1}\left(\ee^{\lK r_1T_1}-1\right)\ee^{r_*T_*+(\lK T_2-T_*)r_2}+\frac{b_*}{r_*}\left(\ee^{r_*T_*}-1\right)\ee^{(\lK T_2-T_*)r_2}+\frac{b_2}{r_2}\left(\ee^{(\lK T_2-T_*)r_2}-1\right)+1}.\nonumber
		\end{align}
		Moreover, for the mean offspring we obtain
		\begin{align}
			m^K=g'(1)
			=\ee^{\lK T_1r_1+T_*r_*+(\lK T_2-T_*)r_2}
			=\ee^{\lK(r_1T_1+r_2T_2)+T_*(r_*-r_2)}.
		\end{align}
		Denote now by $g_n(s)=\Exd{s^{Y^K_n}}$ the generating function of the $n$-th generation of $Y^K$. Then it is well known \cite{AthNey72} that $g_n=g_{n-1}\circ g=g\circ g_{n-1}$. Since $g:[0,1]\to[0,1]$ is convex and strictly increasing, we can deduce that (see Figure \ref{fig:lemB7})
		\begin{align}
		\label{eq:bd_GenFct}
			(1-p^K_0)(1-s)\leq 1-g(s)\leq m^K (1-s)&&\forall s\in[0,1].
		\end{align}
		This can be iterated (cf.\ \cite[eq.\ I.11.7]{AthNey72}) to obtain
		\begin{align}
			(1-p^K_0)^n(1-s)\leq 1-g_n(s)\leq (m^K)^n (1-s)&&\forall s\in[0,1],\forall n\in\N
		\end{align}
		Together with $g_n(0)=\Prob{Y^K_n=0}=\Prob{T^{Y^K}_\ext\leq n}$, where $T^{Y^K}_\ext:=\inf\{n\geq0:Y^K_n=0\}$ is the extinction time of $Y^K$, this leads to
		\begin{align}
			(1-p^K_0)^n\leq \Prob{T^{Y^K}_\ext> n}\leq (m^K)^n.
		\end{align}
		\begin{figure}[h]
			\centering
			\includegraphics[scale=0.7,page=1]{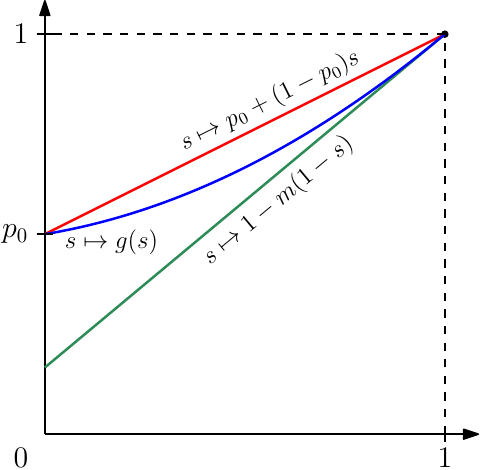}
			\caption{Generating function $g(s)$ and the corresponding affine upper and lower bounds from \eqref{eq:bd_GenFct}.}
			\label{fig:lemB7}
		\end{figure}
		
		Let us now check the upper bound. For $r_\av\geq0$, the claim is trivially satisfied since the right hand side is larger than 1.  In the case of $r_\av<0$, which, for $K$ large enough,  implies $m^K<1$, we can estimate
		\begin{align}
		\label{eq:bd_ext_pf_upper}
			\Prob{T^{Z^K}_\ext>t\ln K\big|Z^K_0=1}
			&\leq \Prob{T^{Z^K}_\ext>\gauss{\frac{t\ln K}{\lK (T_1+T_2)}}\lK (T_1+T_2)|Z^K_0=1}\nonumber\\
			&=\Prob{T^{Y^K}_\ext>\gauss{\frac{t\ln K}{\lK (T_1+T_2)}}|Y^K_0=1}\nonumber\\
			&\leq \left(\ee^{\lK(r_1T_1+r_2T_2)+T_*(r_*-r_2)}\right)^\gauss{\frac{t\ln K}{\lK (T_1+T_2)}}\nonumber\\
			&\leq \left(\ee^{\lK(r_1T_1+r_2T_2)+T_*(r_*-r_2)}\right)^{\left(\frac{t\ln K}{\lK (T_1+T_2)}-1\right)}\nonumber\\
			&\leq \exp\left[r_\av t\ln K + \frac{T_*(r_*-r_2)}{\lK(T_1+T_2)}t\ln K +C\lK \right]\nonumber\\
			&\leq \exp\left[r_\av t\ln K + \d_1 t\ln K +\d_2 \ln K \right]\nonumber\\
			&\leq \exp\left[\left(\left(r_\av+\d_1\right)t+\d_2\right)\ln K\right],
		\end{align}
		again for $K$ large enough.
		
		To obtain the lower bound, for $r_\av\in\R$, we we first calculate that
		\begin{align}
		\left(1-p^K_0\right)^{-1}
		&=\frac{b_1}{r_1}(1-\ee^{-\lK T_1r_1})
			+\frac{b_*}{r_*}(1-\ee^{- T_*r_*})\ee^{-\lK T_1r_1} \nonumber\\
			&\quad +\frac{b_2}{r_2}(1-\ee^{-(\lK T_2-T_*)r_2})\ee^{-\lK T_1r_1-T_*r_*}
			+\ee^{-\lK T_1r_1-T_*r_*-(\lK T_2-T_*)r_2}\nonumber\\
		&\leq\left(\frac{b_1}{\abs{r_1}}+\frac{b_*}{\abs{r_*}}+\frac{b_2}{\abs{r_2}}+1\right) \exp\left[\lK T_1\abs{r_1}+T_*\abs{r_*}+(\lK T_2-T*)\abs{r_2}\right]\nonumber\\
		&=C\exp\left[\lK (T_1\abs{r_1}+T_2\abs{r_2})+T_*(\abs{r_*}-\abs{r_2})\right]\nonumber\\
		&\leq\exp[\lK\tilde{M}],
		\end{align}
		for all $\tilde{M}>T_1\abs{r_1}+T_2\abs{r_2}$ and $K$ large enough.
		This gives $1-p^K_0\geq\exp[-\lK\tilde{M}]$. Using again the connection between $Z^K$ and $Y^K$, we can finally estimate similarly to \eqref{eq:bd_ext_pf_upper}
		\begin{align}
			\Prob{T^{Z^K}_\ext>t\ln K\big|Z^K_0=1}
			&\geq\exp\left[-t\frac{\tilde{M}}{T_1+T_2}\ln K -\lK\tilde{M}\right]\nonumber\\
			&\geq\exp\left[-(tM+\d)\ln K\right],
		\end{align}
		for all $M>(T_1\abs{r_1}+T_2\abs{r_2})/(T_1+T_2)$ and $K$ large enough.
	\end{proof}
	
	Now we have collected all the tools to derive the convergence for shrinking populations and thus conclude the proof of Theorem \ref{thm:bd_main}.
	
	\begin{lemma}
		\label{lem:bd_conv_beta_r_neg}
		Let $Z^K$ follow the law of $\mathrm{BD}\left(b^K,d^K,\beta\right)$ and assume that $r_\av<0$, then the convergence of Theorem \ref{thm:bd_main} holds true.
	\end{lemma}
	\begin{figure}[h]
		\centering
		\includegraphics[scale=0.7,page=1]{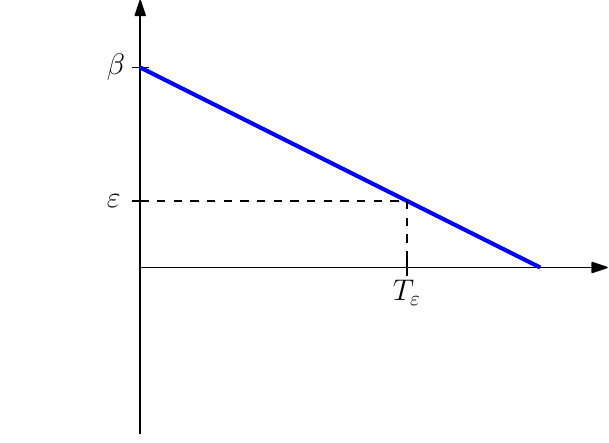}
		\caption{Graph of the limiting exponent $\b+r_\av s$ and the time $T_\ve$ that separates steps 1 and 2 of the proof of Lemma \ref{lem:bd_conv_beta_r_neg}.}
		\label{fig:lemB8}
	\end{figure}
	\begin{proof}
		Fix $\ve\in(0,\b)$ and set $\d=\ve/4,\h=\b-\ve/3$ and $T_\ve=(\b-\ve)/\abs{r_\av}$ (see Figure \ref{fig:lemB8}).
		
		\textbf{Step 1}
		We first show convergence on $[0,T_\ve]$. The computations are very similar to the proof of Lemma \ref{lem:bd_conv_beta_r_pos} and proof step 3(i) of \cite[Lem.\ A.1]{ChMeTr19}, respectively. Here one uses the event
		\begin{align}
			\O^K_2:=\dset{\sup_{t\leq T_\ve\ln K} \abs{\ee^{-g^K(t)}Z^K_t-K^\b}\leq K^\h},
		\end{align}
		which by Lemma \ref{lem:bd_Doob} has probability converging to 1.
		
		\textbf{Step 2}
		Next up, we show that extinction occurs before $T_\ve+3\ve/\abs{r_\av}$. From the previous step, we know that, with high probability, the population at time $t=T_\ve\ln K$ can be bounded from above by $2K^\ve$, for large $K$. Moreover, the first part of Lemma \ref{lem:bd_ext_time} provides an estimate on the extinction time of a subcritical birth death process initialized with one individual. Thus, we obtain the upper bound
		\begin{align}
			\Prob{T^{Z^K}_\ext>s\ln K\big|Z^K_0=2K^\ve}
			&=1-\left(1-\Prob{T^{Z^K}_\ext>s\ln K\big|Z^K_0=1}\right)^{2K^\ve}\nonumber\\
			&\leq1-\left(1-\exp\left[\left(\left(r_\av+\d_1\right)s+\d_2\right)\ln K\right]\right)^{2K^\ve}\nonumber\\
			&=1-\left(1-K^{(r_\av+\d_1)s+\d_2}\right)^{2K^\ve}\nonumber\\
			&=1-\left(1-K^{-2\ve}\right)^{2K^\ve}\nonumber\\
			&\sim 2K^{-\ve}
			\overset{K\to\infty}{\longrightarrow}0.
		\end{align}
		Here, the inequality is a consequence of \eqref{eq:bd_ext_upper}, while for the last equality we set
		\begin{align}
			s=\frac{2\ve+\d_2}{\abs{r_\av+\d_1}}<\frac{3\ve}{\abs{r_\av}},
		\end{align}
		for proper choice of $\d_1,\d_2>0$. The final asymptotic behaviour can be deduced using the limit representation of the exponential function and a first-order approximation.\\
		This now allows to deduce
		\begin{align}
			\Prob{T^{Z^K}_\ext>\frac{3\ve}{\abs{r_\av}}\ln K\big|Z^K_0=2K^\ve}
			\overset{K\to\infty}{\longrightarrow}0.
		\end{align}
		
		\textbf{Step 3} To finally conclude the convergence on the whole interval $[0,T]$, we have to ensure that the process stays bounded in the time interval $[T_\ve,T_\ve+3\ve/\abs{r_\av}]$. To this end, it is sufficient to consider the very rough upper bound obtained from the coupling with a Yule process (a pure birth process) with constant birth rate $\bar{b}=\max\dset{b_1,b_*,b_2}$. Since those processes are non-decreasing, we just need to control the endpoint. The size of a family at time\linebreak$(T_\ve+3\ve/\abs{r_\av})\ln K$, stemming from a single ($i^\text{th}$) individual at time $T_\ve\ln K$, is given by a geometric random variable $G^K_i$ with expectation $K^{3\bar{b}\ve/\abs{r_\av}}$. Since families evolve independently of each other, we obtain by Chebyshev's inequality, for $d:=\bar{b}+\abs{r_\av}$,	
		\begin{align}
			&\Prob{\sup_{s\in[T_\ve,T_\ve+3\ve/\abs{r_\av}]}Z^K_{s\ln K}\geq K^{3d\ve/\abs{r_\av}}}
			\leq\Prob{\sum_{i=1}^{2K^\ve}G^K_i\geq K^{3d\ve/\abs{r_\av}}}
			\overset{K\to\infty}{\longrightarrow}0.
		\end{align}
		
		Overall, this means that, with probability converging to $1$, we have
		\begin{align}
			\sup_{s\in[T_\ve,T_\ve+3\ve/\abs{r_\av}]}\frac{\ln(1+Z^K_{s\ln K})}{\ln K} \leq \frac{3d}{\abs{r_\av}}\ve.
		\end{align}
	\end{proof}

\subsection{Branching processes with immigration}
	\label{app:B.2_BDIproc}
	
	We now turn to the study of birth death processes with immigration. In addition to the birth and death rates defined in the beginning of Section \ref{app:B.1_pureBDproc}, we introduce the following parameters connected to the effects of immigration. Let $c\in\R$ describe the initial order of incoming migration, i.e.\ initially immigrants arrive at overall rate $K^c$. When applied in Section \ref{sec:4.2_convergence_beta},  this is representing the initial size of the neighbouring population multiplied by the mutation rate. Let $a_i,a_{*,i}\in\R$, for $i\in\dset{1,\ldots,\ell}$, be the immigrants net growth rates in the respective (sub-)phases. Thus, we define the time-dependent net growth rate of the immigrants as the periodic extension of
	\begin{align}
		\tilde{a}^K(s):=\left\{\begin{array}{ll}
		a_{*,i}&:s\in[T^\S_{i-1}\lK,T^\S_{i-1}\lK+T_{*,i}),\\
		a_i&:s\in[T^\S_{i-1}\lK+T_{*,i},T^\S_{i}\lK).\\
		\end{array}\right. 
	\end{align}
	Moreover, (corresponding to $g_K(t)$ and $r_\av$) we define the time integral
	\begin{align}
		a^K(t):=\int_0^t \tilde{a}^K(s)\dd s.
	\end{align}
	and the average growth rate of the immigrants $a_\av:=\left(\sum_{i=1}^\ell a_iT_i\right)/T^\S_\ell$.
	Hence, the overall rate of immigration at time $t$ is given by $K^c\ee^{a^K(t)}$.  We can define the Markov processes $(Z^K_t)_{t\geq0}$ generated by
	\begin{align}
		\left(\mathcal{L}^K_tf\right)(n)=\left(b^K(t)n+K^c \ee^{a^K(t)}\right)\left(f(n+1)-f(n)\right)+d^K(t)n\left(f(n-1)-f(n)\right)
	\end{align}
	and with $Z^K_0=\gauss{K^\b-1}$. We refer to the law of such processes by $Z^K\sim\mathrm{BDI}\left(b^K,d^K,\b,a^K,c\right)$.
	
	As in the previous section, we derive a convergence result for the logarithmically rescaled process.

	\begin{theorem}
	\label{thm:bdi_main}
		Let $Z^K$ follow the law of $\mathrm{BDI}\left(b^K,d^K,\b,a^K,c\right)$. Then, for all fixed $T\in(0,\infty)$, the following convergence holds in probability, with respect to the $L^\infty([0,T])$ norm,
		\begin{align}
			\left(\frac{\ln\left(1+Z^K_{s\ln K}\right)}{\ln K}\right)_{s\in[0,T]}
			\overset{K\to\infty}{\longrightarrow}
			\left(\bar{\b}_s\right)_{s\in[0,T]},
		\end{align}
		where
		\begin{align}
			\bar{\b}_s=\left\{\begin{array}{ll}
				(c+(r_\av\vee a_\av)s)\vee 0&:c>\b.\\
				(\b+r_\av s)\vee(c+a_\av s)\vee(c+r_\av s)\vee 0&:\b>0,c\leq \b,\\
				(r_\av\vee a_\av)(s-\abs{c}/a_\av)\vee0&:\b=0,c\leq 0,a_\av>0,\\
				0&:\b=0,c< 0,a_\av\leq 0,\\
				0&:\b=c=0,a_\av\leq 0,r_\av\leq0.
				\end{array}\right.
		\end{align}
	\end{theorem}
Note that the one case that is not cover by this result is that of $\b=c=0$, $a_\av\leq0$, $r_\av>0$.	
	
	The remainder of this section is dedicated to the proof of this theorem. We first study a number of specific cases in a series of lemmas and then outline how these can be combined to prove the final general result.
	
	Since $a^K$ is of the same form as $g^K$, we directly obtain.
	
	\begin{lemma}
		\label{lem:bdi_conv_imm_rate}
		For all fixed $T<\infty$ we have the uniform convergence
		\begin{align}
		\sup_{s\leq T}\abs{\frac{\ln\left(K^c\ee^{a^K(s\ln K)}\right)}{\ln K} -\left(c+a_\av s\right)}
		\overset{K\to\infty}{\longrightarrow}0.
		\end{align}
	\end{lemma}
	\begin{proof}
		See the proof of Lemma \ref{lem:bd_conv_expectation}.
	\end{proof}

	As before, we can construct the processes $Z^K$ in terms of the Poisson random measures $Q^{(b)}(\dd s,\dd\theta)$ and $Q^{(d)}(\dd s,\dd\theta)$ and derive the martingale decomposition $Z^K_t=Z^K_0+M^K_t+A^K_t$, where
	
	\begin{align}
		Z^K_t&=Z^K_0
		+\int_0^t\int_{\R_{\geq 0}}\ifct{\theta\leq b^K(s^-)Z^K_{s^-}+K^c\ee^{a^K(s)}}Q^{(b)}(\dd s,\dd\theta)
		-\int_0^t\int_{\R_{\geq 0}}\ifct{\theta\leq d^K(s^-)Z^K_{s^-}}Q^{(d)}(\dd s,\dd\theta),\\
		M^K_t&=\int_0^t\int_{\R_{\geq 0}}\ifct{\theta\leq b^K(s^-)Z^K_{s^-}+K^c\ee^{a^K(s)}}\tilde{Q}^{(b)}(\dd s,\dd\theta)
		-\int_0^t\int_{\R_{\geq 0}}\ifct{\theta\leq d^K(s^-)Z^K_{s^-}}\tilde{Q}^{(d)}(\dd s,\dd\theta),\\
		A^K_t&=\int_0^t r^K(s)Z^K_s+K^c\ee^{a^K(s)}\dd s,\\
		\langle M^K\rangle_t &= \int_0^t \left(b^K(s)+d^K(s)\right)Z^K_s+K^c\ee^{a^K(s)}\dd s,\\
		\label{eq:bdi_mrtgl_decomp}
		\dd Z^K_t&=\dd M^K_t+\left(r^K(t)Z^K_t+K^c\ee^{a^K(t)}\right)\dd t.
	\end{align}
	Note that, as in Section \ref{app:B.1_pureBDproc}, $\tilde{Q}^{(*)}:=Q^{(*)}-\dd s\dd \theta$ and equalities only hold in distribution.
	
	As in the case without immigration, we first take a look at the expected value. Moreover, we derive a bound on the variance of the process.
	
\begin{lemma}
	\label{lem:bdi_expectation_formula}
	Let $Z^K$ follow the law of $\mathrm{BDI}\left(b^K,d^K,\b,a^K,c\right)$ and assume that $r_\av\neq a_\av$. Then, for fixed $T<\infty$ and all $t\in[0,T\ln K]$,
	\begin{align}
	\label{eq:bdi_expectation formula}
	z^K_t:=\Exd{Z^K_t}\approx\left(\ee^{r_\av t}\left(K^\b-1\right)+\frac{K^c}{a_\av-r_\av}\left[\ee^{a_\av t}-\ee^{r_\av t}\right]\right)K^{\pm 3\d}.
	\end{align}
	
	Moreover, under the additional assumption of $r_\av\neq 0$ and $2r_\av\neq a_\av$, one obtains
	\begin{align}
	\mathrm{Var}\left(Z^K_t\right)\leq
	\left[(\bar{b}+\bar{d})\left(K^\b-1+\frac{K^c}{r_\av-a_\av}\right)\frac{\ee^{2r_\av t}-\ee^{r_\av t}}{r_\av} + K^c\left(1-\frac{\bar{b}+\bar{d}}{r_\av-a_\av}\right)\frac{\ee^{a_\av t}-\ee^{2r_\av t}}{a_\av-2 r_\av}\right]K^{5\d}
	\end{align}
\end{lemma}

Here and in the remainder of the appendix, these kind of approximation results are meant in the following way: For any $\d>0$, there exists $K_0\in\N$ such that, for all $K\geq K_0$, plugging in $K^{+3\d}$ yields an upper bound and $K^{-3\d}$ a lower bound for the left hand side.

\begin{proof}
	As in Lemma \ref{lem:bd_expectation_formula}, we use the martingale decomposition to derive the integral equation
	\begin{align}
	z^K_t=\gauss{K^\b-1}+\int_0^t r^K(s)z^K_s+K^c\ee^{a^K(s)}\dd s.
	\end{align}
	By variation of constants, this leads to
	\begin{align}
	\label{eq:bdi_zKt}
	z^K_t=\ee^{g^K(t)}\left(\gauss{K^\b-1}+K^c\int_0^t\ee^{-g^K(s)}\ee^{a^K(s)} \dd s\right).
	\end{align}
	
	Using the convergence of $g^K$ and $a^K$ (Lemmas \ref{lem:bd_conv_expectation} and \ref{lem:bdi_conv_imm_rate}), we can estimate, for any $t\in[0,T\ln K]$, fixed $\d>0$, and $K$ large enough,
	\begin{align}
	z^K_t\approx\ee^{r_\av t}\left(K^\b-1+K^c\int_0^t\ee^{(a_\av-r_\av)s}\dd s\right)K^{\pm 3\d}.
	\end{align}
	This directly gives the first claim.
	
	For the estimate on the variance, we see that
	\begin{align}
	\dd\left(Z^K_t\right)^2
	=2Z^K_t\dd Z^K_t+\dd\langle Z^K\rangle_t
	=2Z^K_t\dd M^K_t +2Z^K_t\left(r^K(t)Z^K_t+K^c\ee^{a^K(t)}\right)\dd t +\dd\langle Z^K\rangle_t.
	\end{align}
	Define $u^K_t:=\Exd{\left(Z^K_t\right)^2}$, then $u^K_0=\gauss{K^\b-1}^2$ and
	\begin{align}
	\dot{u}^K_t&=0+2 r^K(t)u^K_t+2 K^c\ee^{a^K(t)}z^K_t+\left(\left(b^K(t)+d^K(t)\right)z^K_t+K^c\ee^{a^K(t)}\right)\nonumber\\
	&=2 r^K(t)u^K_t
	+\left(2 K^c\ee^{a^K(t)}+b^K(t)+d^K(t)\right)z^K_t+K^c\ee^{a^K(t)}.
	\end{align}
	Using variation of constants, we deduce
	\begin{align}\label{eq:bdi_uKt}
	u^K_t=&\ee^{2g^K(t)}\left(\gauss{K^\b-1}^2 +\int_0^t\ee^{-2g^K(s)}\left[\left(2K^c\ee^{a^K(s)}+b^K(s)+d^K(s)\right)z^K_s+K^c\ee^{a^K(s)}\right]\dd s\right).
	\end{align}
	
	We now focus on the integral term, where we plug in \ref{eq:bdi_zKt} and treat each summand separately. For the first summand, involving $2K^c\ee^{a^K(s)}$, we obtain
	\begin{align}
	&2K^c\int_0^t\ee^{-g^K(s)}\ee^{a^K(s)}\left(\gauss{K^\b-1}+K^c\int_0^s\ee^{-g^K(w)}\ee^{a^K(w)} \dd w\right)\dd s\nonumber\\
	=&2\gauss{K^\b-1}K^c\int_0^t\ee^{-g^K(s)}\ee^{a^K(s)}\dd s
	+2K^{2c}\int_0^t\ee^{-g^K(s)}\ee^{a^K(s)}\int_0^s\ee^{-g^K(w)}\ee^{a^K(w)} \dd w\dd s\nonumber\\
	=&2\gauss{K^\b-1}K^c\int_0^t\ee^{-g^K(s)}\ee^{a^K(s)}\dd s +\left(K^{c}\int_0^t\ee^{-g^K(s)}\ee^{a^K(s)}\dd s\right)^2.
	\end{align}
	Together with the term $\gauss{K^\b-1}^2$ and the prefactor $\ee^{2g^K(t)}$ from \ref{eq:bdi_uKt}, this is equals to the square of $z^K_t$. Thus, the other two summands of \ref{eq:bdi_uKt} give us the desired variance. Overall, using the convergence of $g^K$ and $a^K$ as above and plugging in \ref{eq:bdi_zKt}, this yields 
	\begin{align}
	&\mathrm{Var}\left(Z^K_t\right)=u^K_t-(z^K_t)^2\nonumber\\
	&=\ee^{2g^K(t)}\left(\int_0^t\ee^{-2g^K(s)}\left(b^K(s)+d^K(s)\right)z^K_s \dd s
	+K^c \int_0^t\ee^{-2g^K(s)}\ee^{a^K(s)} \dd s \right)\nonumber\\
	&\leq\ee^{2g^K(t)}\left(\left(\bar{b}+\bar{d}\right)\int_0^t\ee^{-g^K(s)}\left(K^\b-1+K^c\int_0^s\ee^{-g^K(w)}\ee^{a^K(w)} \dd w\right) \dd s
	+K^c \int_0^t\ee^{-2g^K(s)}\ee^{a^K(s)} \dd s \right)\nonumber\\
	&\leq K^{5\d}\ee^{2r_\av t}\left(\left(\bar{b}+\bar{d}\right)\int_0^t\ee^{-r_\av s}\left(K^\b-1+K^c\int_0^s\ee^{-r_\av w}\ee^{a_\av w} \dd w\right) \dd s
	+K^c \int_0^t\ee^{-2r_\av s}\ee^{a_\av s} \dd s \right)\nonumber\\
	&=K^{5\d}\left(\left(\bar{b}+\bar{d}\right)\left(K^\b-1+\frac{K^c}{r_\av-a_\av}\right) \frac{\ee^{2r_\av t}-\ee^{r_\av t}}{r_\av}
	+K^c \left(1-\frac{\bar{b}+\bar{d}}{r_\av-a_\av}\right) \frac{\ee^{a_\av t}-\ee^{2r_\av t}}{a_\av-2r_\av} \right).
	\end{align}
\end{proof}
	
	Similar to the model without immigration (c.f.\ Lemma \ref{lem:bd_Doob}), the starting point for proving the different parts of Theorem \ref{thm:bdi_main} is an estimate on the bracket of the rescaled martingale, which will be used in combination with Doob's maximum inequality.
	
	\begin{lemma}
		\label{lem:bdi_Doob}
		Let $Z^K\sim\mathrm{BDI}\left(b^K,d^K,\b,a^K,c\right)$. Then the process
		\begin{align}
			\tilde{M}^K_t:=\ee^{-g^K(t)}\left(Z^K_t-z^K_t\right),
		\end{align}
		with $z^K_t$ being defined in \eqref{eq:bdi_expectation formula}, is a martingale.  Under the assumption that $r_\av\notin \{a_\av,2a_\av\}$, for all $T<\infty$ and all $\d>0$, there is a $K_0\in\N$ such that, for all $K\geq K_0$, it holds that
		\begin{align}
			&\Exd{\langle\tilde{M}^K\rangle_{T\ln K}}K^{-5\d}
			\leq K^c\frac{K^{(a_\av-2r_\av)T}-1}{a_\av-2r_\av}\nonumber\\
			&+\left(\bar{b}+\bar{d}\right)\left(\left(K^\b-1\right)\frac{1-K^{-r_\av T}}{r_\av} +\frac{K^c}{a_\av-r_\av}\left[\frac{K^{(a_\av-2r_\av)T}-1}{a_\av-2r_\av}-\frac{1-K^{-r_\av T}}{r_\av}\right]\right).
		\end{align}
	\end{lemma}
	\begin{proof}
		Using Itô's formula and the martingale decomposition \eqref{eq:bdi_mrtgl_decomp}, we get
		\begin{align}
		\dd \tilde{M}^K_t=\ee^{-g^K(t)}\dd M^K_t,
		\end{align}
		which yields that $\tilde{M}^K$ is a martingale.
		Moreover, an application of Itô's isometry gives
		\begin{align}
			\dd\langle\tilde{M}^K\rangle_t
			=\ee^{-2g^K(t)}\dd\langle M^K\rangle_t
			=\ee^{-2g^K(t)}\left(\left(b^K(t)+d^K(t)\right)Z^K_t+K^c\ee^{a^K(t)}\right)\dd t.
		\end{align}
		Thus,
		\begin{align}
			\Exd{\langle\tilde{M}^K\rangle_{T\ln K}}
			=\int_0^{T\ln K}\ee^{-2g^K(t)}\left(\left(b^K(t)+d^K(t)\right)z^K_t+K^c\ee^{a^K(t)}\right)\dd t,
		\end{align}
		and finally, using Lemma \ref{lem:bdi_expectation_formula} and the asymptotics of $g^K$, we obtain, for $K$ large enough,
		\begin{align}
			&\int_0^{T\ln K}\ee^{-2g^K(t)}z^K_t\dd t
			\leq K^{5\d}\int_0^{T\ln K}\ee^{-r_\av t}\left(K^c\frac{\ee^{(a_\av-r_\av)t}-1}{a_\av-r_\av}+(K^\b-1)\right)\dd t,
		\end{align}
	which yields the claimed estimate.
	\end{proof}
	
	Now we can start checking the convergence of Theorem \ref{thm:bdi_main} in the cases without extinction or newly emerging populations.
	
	\begin{lemma}
		\label{lem:bdi_positive_beta}
		The claim of Theorem \ref{thm:bdi_main} holds true for $c\leq\b,\b>0$ and all $T<\infty$ such that
		\begin{align}
			\inf_{t\in[0,T]}(\b+r_\av t)\vee(c+a_\av t)>0.
		\end{align}
	\end{lemma}
	\begin{proof}
		We split the proof into several steps. First, in step 1, we apply Doob's maximum inequality for the rescaled martingales $\tilde{M}^K$ to prove the convergence in most of the possible cases. Next, in step 2, we make use of another maximum inequality to check the convergence for some other cases. Finally, in step 3, we go through a case distinction of all the possible scenarios of the lemma and explain the strategy of glueing together parts 1 and 2 with the help of the Markov property to cover the remaining cases.
		
		\textbf{Step 1:}
			Let us consider the case where we can find an $\h$ such that
			\begin{align}
			\label{eq:bdi_first_constraint}
				\frac{1}{2}\left[\b\vee (\b-r_\av T)\vee (c+(a_\av-2r_\av)T)\right]
					<\h<\b.
			\end{align}
			Then, applying Doob's maximum inequality to $\tilde{M}^K$ and adapting the computations of case 1(b) of the proof of \cite[Lem.\ B.1]{ChMeTr19} one obtains the desired convergence.
			
		\textbf{Step 2:}
			Let us now consider the specific case of $\b=c$ and $a_\av>r_\av$. Then we can apply the maximum inequality of \cite[Ch. ~VI.1.2. ~p. ~66]{DelMey82} to the supermartingales $\left(\ee^{-\left(a_\av t-r_\av t\right)}\tilde{M}^K_t\right)_{t\geq 0}$ to obtain, for $K$ large enough,
			\begin{align}
			\label{eq:bdi_max_ineq}
				\Prob{\sup_{t\leq T\ln K}\ee^{-a_\av t}\abs{Z^K_t-z^K_{t}}>K^\h}
				&=\Prob{\forall t\leq T\ln K: \ee^{-\left(a_\av t-r_\av t\right)}\abs{\tilde{M}^K_t}>\ee^{r_\av t-g^K(t)}K^\h} \nonumber\\
				&\leq\Prob{\sup_{t\leq T\ln K} \ee^{-\left(a_\av t-r_\av t\right)}\abs{\tilde{M}^K_t} >K^{\h-\d}} \nonumber\\
				&\leq 3 K^{-\h+\d} \sup_{t\leq T\ln K} \ee^{-\left(a_\av t-r_\av t\right)} \Exd{\langle\tilde{M}^K\rangle_t}^{\frac{1}{2}} \nonumber\\
				&\leq C K^{-\h+\d} \sup_{s\leq T} K^{-\left(a_\av -r_\av \right)s+\frac{1}{2}(\b\vee(\b-r_\av s)\vee(c+(a_\av-2r_\av)s))}.
			\end{align}
			To conclude, we again proceed as in proof step 2 of \cite[Lem.\ B.1]{ChMeTr19}.
		
		\textbf{Step 3:}
			Now we can check whether all cases are covered. Under the assumptions of the lemma, i.e.\ $\b>0,c\leq\b,T\in(0,\infty)$ and $\inf_{s\in[0,T]}(\b+r_\av s)\vee(c+a_\av s)>0$, we see that the constraint \eqref{eq:bdi_first_constraint} of step 1, which is equivalent to
			\begin{align}
				r_\av >-\frac{\b}{T}\quad\text{and}\quad c+(a_\av-2r_\av)T<2\b,
			\end{align}
			holds true, if
			\begin{enumerate}[(i)]
				\item $r_\av\geq 0, a_\av\leq 2r_\av$,
				\item $r_\av\geq 0, a_\av>2r_\av$ and $T<T^*:=(2\b-c)/(a_\av-2r_\av)$,
				\item $r_\av<0, a_\av\leq r_\av$,
				\item $r_\av<0, a_\av>r_\av, c+a_\av\b/\abs{r_\av}\leq 0$ (in this case $T< \b/|r_\av|$),
				\item $r_\av<0, a_\av>r_\av, c+a_\av\b/\abs{r_\av}>0$ and $T<T^*$ (in this case $T^*\leq \b/|r_\av|$).
			\end{enumerate}
			Hence the only remaining cases are (see Figure \ref{fig:lemB13})
			\begin{itemize}
				\item $r_\av\geq 0, a_\av>2r_\av$ and $T>T^*:=(2\b-c)/(a_\av-2r_\av)$,
				\item $r_\av<0, a_\av>r_\av, c+a_\av\b/\abs{r_\av}>0$ and $T>T^*$.
			\end{itemize}
			The strategy for these is the following: Take $t^*<T_1<T^*$,  where $t^*:=(\b-c)/(a_\av-r_\av)\geq0$ is the first time when $(s\mapsto\b+r_\av s)$ crosses $(s\mapsto c+a_\av s)$. Now apply step 1 on the interval $[0,T_1]$ to get the desired convergence up to $T_1$. In particular, $1+Z^K_{T_1\ln K}\approx K^{c+a_\av T_1\pm\eps}$ for $K$ large enough. On $[T_1,T]$ we apply step 2 to the approximating processes $Z^{(K,+)},Z^{(K,-)}$, with parameters $b^\pm_i:=b_i,d^\pm_i:=d_i,a^\pm_i:=a_i$, for $i\in\dset{*,1,2}$, and $\b^\pm=c^\pm:=c+a_\av T_1 \pm\ve$. Together with the Markov property these approximations give the claim.
			\begin{figure}[h]
				\centering
				\includegraphics[scale=0.7,page=2]{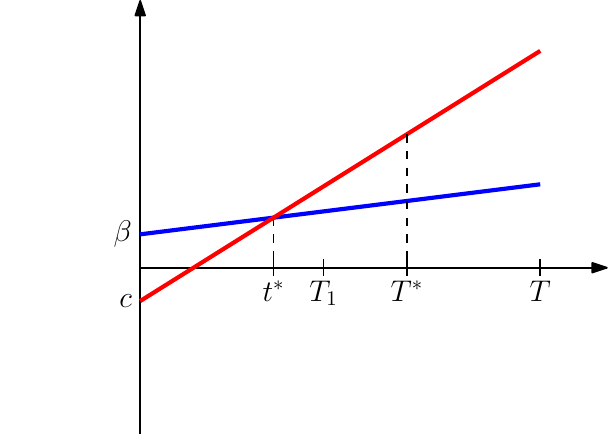}
				\includegraphics[scale=0.7,page=3]{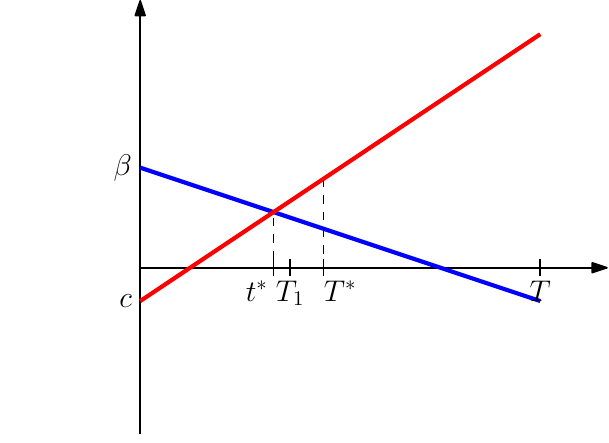}
				\caption{Illustration of the two remaining cases of step 3 of the proof of Lemma \ref{lem:bdi_positive_beta}.}
				\label{fig:lemB13}
			\end{figure}		
	\end{proof}

	Let us now collect some technical results on the (non-)emergence, extinction and instantaneous immigration of populations, as well as the continuity of the exponent, to complete the proof of Theorem \ref{thm:bdi_main}.

	\begin{lemma}[Non-emergence of any new population]
		\label{lem:bdi_non_emerge}
		Let $Z^K\sim\mathrm{BDI}\left(b^K,d^K,\b,a^K,c\right)$ such that $\b=0$ and $c<0$. Then, for all $T>0$ such that $c+a_\av T<0$ it holds
		\begin{align}
			\lim_{K\to\infty}\Prob{\forall s\in[0,T]: Z^K_{s\ln K}=0}=1.
		\end{align}
	\end{lemma}
	\begin{proof}
		This can be shown by following the proof of \cite[Lem.\ B.7]{ChMeTr19}.
	\end{proof}
	\begin{figure}[h]
		\centering
		\includegraphics[scale=0.7,page=4]{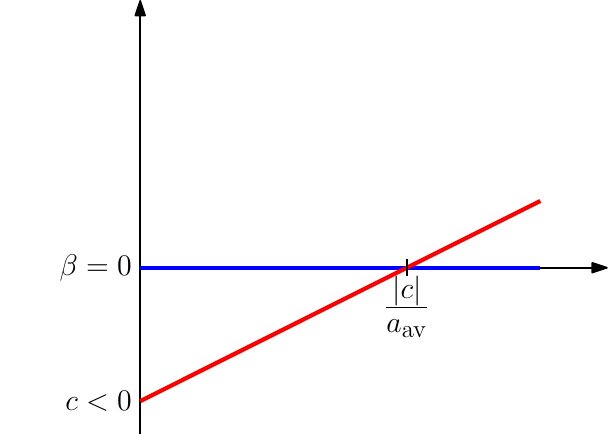}
		\includegraphics[scale=0.7,page=5]{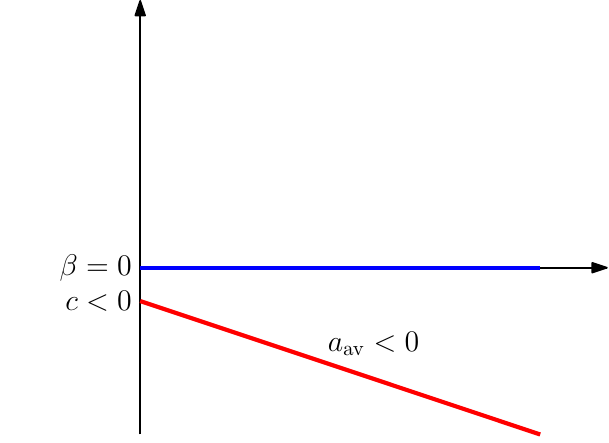}
		\caption{Illustrations of the two possible situations for non-emergence of any new population, as in Lemma \ref{lem:bdi_non_emerge}.}
		\label{fig:lemB14}
	\end{figure}

	\begin{lemma}[Emergence of a new population]
		\label{lem:bdi_emerge}
		Let $Z^K\sim\mathrm{BDI}\left(b^K,d^K,\b,a^K,c\right)$ such that $\b=0,c=-\ve<0$ and $a_\av>0$. Then, for all $\h>(1\vee2r_\av/a_\av)\ve$, it holds that
		\begin{align}
			\lim_{K\to\infty}\Prob{K^{\ve/2}-1\leq Z^K_{\frac{2\ve}{a_\av}\ln K}\leq K^\h-1}=1.
		\end{align}
	\end{lemma}
	\begin{proof}
		Let us first consider the lower bound. The number of immigrant families that arrive within the time interval $[0,(2\ve/a_\av)\ln K]$ and that survived up to time $(2\ve/a_\av)\ln K$, is, by thinning, a Poisson random variable with parameter
		\begin{align}
			\vartheta=\int_0^{\frac{2\ve}{a_\av}\ln K} K^{-\ve}\ee^{a^K(t)} \Prob{T^{\tilde{Z}^K}_\ext>\frac{2\ve}{a_\av}\ln K-t\big\vert \tilde{Z}^K_0=1}\dd t
		\end{align}
		where $\tilde{Z}^K$ is the corresponding birth death process without immigration.
		Hence we can apply the second part of Lemma \ref{lem:bd_ext_time} and bound this from below, for all $\d>0$ and $K$ large enough, 
		\begin{align}
			\vartheta
			&=\int_0^{\frac{2\ve}{a_\av}} K^{-\ve}\ee^{a^K(s\ln K)} \Prob{T^{\tilde{Z}^K}_\ext>\left(\frac{2\ve}{a_\av}-s\right)\ln K\big\vert \tilde{Z}^K_0=1} \ln K\dd s \nonumber\\
			&\geq \int_0^{\frac{2\ve}{a_\av}} K^{-\ve}K^{-\d}\ee^{a_\av s\ln K} \exp\left[-\left(M\left(\frac{2\ve}{a_\av}-s\right)+\d\right)\ln K\right] \ln K\dd s \nonumber\\
			&=K^{-\ve-2\d-M\frac{2\ve}{a_\av}} \int_0^{\frac{2\ve}{a_\av}}\ee^{\left(a_\av+M\right)s\ln K} \ln K \dd s \nonumber\\
			&=K^{-\ve-2\d-M\frac{2\ve}{a_\av}} \frac{K^{\frac{2\ve}{a_\av}\left(a_\av+M\right)}-1}{a_\av+M} \nonumber\\
			&\geq \frac{1}{2}\frac{1}{a_\av+M} K^{-\ve-2\d-M\frac{2\ve}{a_\av}+\frac{2\ve}{a_\av}\left(a_\av+M\right)} \nonumber\\
			&=C K^{\ve-2\d}.
		\end{align}
		Therefore, taking $\d>0$ small enough, we have
		\begin{align}
			\lim_{K\to\infty}\Prob{Z^K_{\frac{2\ve}{a_\av}\ln K} \geq K^{\frac{\ve}{2}}}=1
		\end{align}
		and can continue as in the proof of \cite[Lem.\ B.8]{ChMeTr19}.
	\end{proof}
	\begin{figure}[h]
		\centering
		\includegraphics[scale=0.7,page=6]{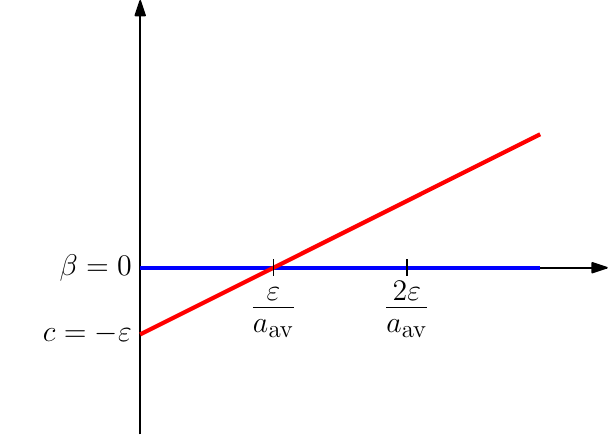}
		\caption{Illustration of the situation for emergence of a new population, as in Lemma \ref{lem:bdi_emerge}.}
		\label{fig:lemB15}
	\end{figure}	

	\begin{lemma}[Continuity of the exponent]
		\label{lem:bdi_continuity}
		Let $Z^K\sim\mathrm{BDI}\left(b^K,d^K,\b,a^K,c\right)$ such that $c\leq\b$. Then there exists a constant $\bar{c}=\bar{c}(\bar{b},\bar{d},a_\av)$ such that, for all $\ve>0$, it holds that
		\begin{align}
			\lim_{K\to\infty}\Prob{\forall s\in[0,\ve]: K^{\b-\bar{c}\ve}-1\leq Z^K_{s\ln K}\leq K^{\b+\bar{c}\ve}-1}=1.
		\end{align}
	\end{lemma}
	\begin{proof}
		Using the estimate of Lemma \ref{lem:bdi_conv_imm_rate} one can adapt the proof of \cite[Lem.\ B.9]{ChMeTr19}.
	\end{proof}

	\begin{lemma}[Extinction]
		\label{lem:bdi_extinction}
		Let $Z^K\sim\mathrm{BDI}\left(b^K,d^K,\b,a^K,c\right)$ such that $r_\av<0$.
		\begin{enumerate}[(a)]
			\item If in addition $c<0$ and $c+a_\av\b/\abs{r_\av}<0$, then, for all $0<\h<T$ with the property $c+a_\av(\b/\abs{r_\av}+T)<0$, it holds that
			\begin{align}
				\lim_{K\to\infty}\Prob{\forall s\in \left[\frac{\b}{\abs{r_\av}}+\h,\frac{\b}{\abs{r_\av}}+T\right]:Z^K_{s\ln K}=0}=1.
			\end{align}
			\item If in addition $a_\av<0$ and $c+a_\av\b/\abs{r_\av}>0$, then, for all $0<\h<T$, we have
			\begin{align}
				\lim_{K\to\infty}\Prob{\forall s\in \left[\frac{c}{\abs{a_\av}}+\h,\frac{c}{\abs{a_\av}}+T\right]:Z^K_{s\ln K}=0}=1.
			\end{align}
			
		\end{enumerate}
	\end{lemma}	
	\begin{proof}
		The proof follows the lines of the one for \cite[Lem.\ B.10]{ChMeTr19}. We emphazise, that the results of Lemma \ref{lem:bd_ext_time} on the extinction time are crucial for the processes studied in the present paper.
	\end{proof}
	\begin{figure}[h]
		\centering
		\includegraphics[scale=0.7,page=7]{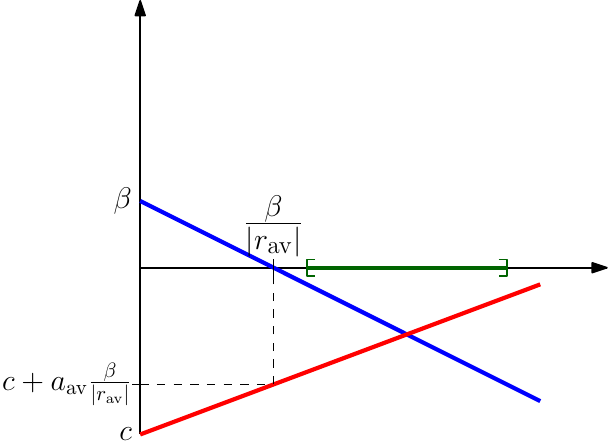}
		\includegraphics[scale=0.7,page=8]{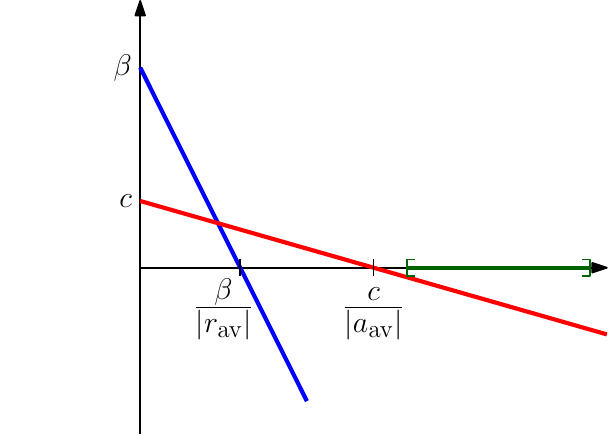}
		\caption{Illustration of cases (a) and (b) for extinction in Lemma \ref{lem:bdi_extinction}.}
		\label{fig:lemB17}
	\end{figure}

	\begin{lemma}[Instantaneous immigration]
		\label{lem:bdi_instantanous_immigration}
		Let $Z^K\sim\mathrm{BDI}\left(b^K,d^K,\b,a^K,c\right)$ such that \linebreak$0\leq\b<c$. Then, for all $\ve>0$ and all $\bar{a}>\abs{r_\av}\vee\abs{a_\av}$ it holds that
		\begin{align}
			\lim_{K\to\infty}\Prob{Z^K_{\ve\ln K}\in [K^{c-\bar{a}\ve},K^{c+\bar{a}\ve}]}=1.
		\end{align}
	\end{lemma}
	\begin{proof}
		With the help of Lemma \ref{lem:bd_expectation_formula}, this is verified as for \cite[Lem.\ B.4]{ChMeTr19}.
	\end{proof}
	
	Having all tools at hand, we can now prove the main convergence result step by step. 	
	
	\begin{proof}[Proof of Theorem \ref{thm:bdi_main}]
	The previous lemmas cover all the different scenarios of Theorem \ref{thm:bdi_main} either directly or in combination, under application of the Markov property and the continuity result in Lemma \ref{lem:bdi_continuity}. For convenience, we provide a summary of how to treat each case:\\
	First, we check all the different cases with $c\leq\b$.
	\begin{itemize}
		\item If $\b=0,c<0$ and $a_\av\leq0$, we apply Lemma \ref{lem:bdi_non_emerge}.
		\item If $\b=0,c<0$ and $a_\av>0$, we apply Lemma \ref{lem:bdi_non_emerge} up to time $(\abs{c}-\ve)/a_\av$. On $[(\abs{c}-\ve)/a_\av,(\abs{c}+\ve)/a_\av]$, we use Lemma \ref{lem:bdi_emerge}, and finally, Lemma \ref{lem:bdi_positive_beta} on \linebreak$[(\abs{c}+\ve)/a_\av,T]$. 
		\item If $\b>0$, the case $r_\av\geq 0$ is already checked in Lemma \ref{lem:bdi_positive_beta}.
		\item If $\b>0,r_\av<0$ and $a_\av\geq 0$ such that $c+a_\av\b/\abs{r_\av}>0$, we can also apply Lemma \ref{lem:bdi_positive_beta}.
		\item If $\b>0,r_\av<0$ and $a_\av> 0$ such that $c+a_\av\b/\abs{r_\av}<0$, we use Lemma \ref{lem:bdi_positive_beta} on $[0,\b/\abs{r_\av}-\ve]$. After that time, an application of Lemma \ref{lem:bdi_extinction} leads to extinction until $(\abs{c}-\ve)/a_\av$. On $[(\abs{c}-\ve)/a_\av,(\abs{c}+\ve)/a_\av]$, we use Lemma \ref{lem:bdi_emerge}, and finally, Lemma \ref{lem:bdi_positive_beta} on $[(\abs{c}+\ve)/a_\av,T]$. 
		\item If $\b>0$ and either $(r_\av<0,a_\av<0)$ or $(r_\av<0,a_\av=0,c<0)$, we use Lemma \ref{lem:bdi_positive_beta} on $[0,\b/\abs{r_\av}\lor c/|a_\av|-\ve]$ or $[0,\b/\abs{r_\av}-\ve]$, respectively. After that time, an application of Lemma \ref{lem:bdi_extinction} leads to extinction.
		\item Cases $(\b=c=0, a_\av>0)$, $(\b=c=0, a_\av\leq0, r_\av\leq0)$, $(\b>0,r_\av<0,c=a_\av=0)$, as well as $(\b>0,r_\av<0,a_\av>0,c<0,\b/\abs{r_\av}=\abs{c}/a_\av)$ can be treated using comparisons.
	\end{itemize}
	Finally, for $c>\b$, we apply Lemma \ref{lem:bdi_instantanous_immigration} on $[0,\ve]$ and then continue with the previous results.
	\end{proof}

\section{Phase of invasion}
\label{app:C}
	In contrast to the results of Appendix \ref{app:B_branching_processes}, in this chapter we focus on birth and death processes with many individuals on shorter but still divergent time horizons. These restrictions allow for relatively strong  bounds on the growth of these processes.
	The results are used in the proof of the main theorem of this paper to control the mutant population when it gets almost macroscopic, to ensure that this happens in the right (fit) phase.

	\begin{lemma}
		\label{lem:growth_in_lambda_time}
		Let $Z^K$ be birth death processes with constant parameters $b,d\geq 0$ and $r=b-d$. Moreover, for $q\in(0,\infty)$ and $\a\in\R$, assume initial values $Z^K_0=q\ee^{\a\lK}K$. Then, for all $T\geq 0$ and all $\g\in(0,1)$,
		\begin{enumerate}[(a)]
			\item $\Prob{Z^K_{T\lK}<p\ee^{rT\lK}Z^K_0}
			=o(K^{-\g})\overset{K\to\infty}{\longrightarrow}0$,\qquad for $p\in(0,1)$,
			\item $\Prob{Z^K_{T\lK}>p\ee^{rT\lK}Z^K_0}
			=o(K^{-\g})\overset{K\to\infty}{\longrightarrow}0$,\qquad for $p\in(1,\infty)$.
		\end{enumerate}
	\end{lemma}
	
	\begin{proof}
		For this proof we make use of the rescaled martingale $\hat{M}^K_t=\ee^{-rt}Z^K_t-Z^K_0$. Then $\dd \hat{M}^K_t=\ee^{-rt}\dd M^K_t$ and thus $\dd\langle\hat{M}^K\rangle_t=\ee^{-2rt}\dd\langle M^K\rangle_t$ (cf.\ \eqref{eq:bd_bracket} and corresponding discussion). Therefore, in the case of $r\neq 0$, we can compute
		\begin{align}
			\Exd{\langle\hat{M}^K\rangle_t}
			&=\int_0^t\ee^{-2rs}\left(b+d\right)\Exd{Z^K_s}\dd s \nonumber\\
			&=\left(b+d\right)Z^K_0\int_0^t\ee^{-rs}\dd s \nonumber\\
			&=\left(b+d\right)Z^K_0\frac{1-\ee^{-rt}}{r}.
		\end{align}
		An application of Doob's maximum inequality \cite[Thm. ~3.87]{BovHol15} yields that, for $0<p<1$,
		\begin{align}
			\Prob{Z^K_{T\lK}<p\ee^{rT\lK}Z^K_0}
			&=\Prob{\ee^{-rT\lK}Z^K_{T\lK}-Z^K_0<-(1-p)Z^K_0} \nonumber\\
			&\leq\Prob{\sup_{t\leq T\lK}\abs{\ee^{-rt}Z^K_t-Z^K_0}>(1-p)Z^K_0} \nonumber\\
			&=\Prob{\sup_{t\leq T\lK}\abs{\hat{M}^K_t}>(1-p)Z^K_0} \nonumber\\
			&\leq (1-p)^{-2}\left(Z^K_0\right)^{-2}\Exd{\langle\hat{M}^K\rangle_{T\lK}} \nonumber\\
			&=C\ee^{-\a\lK}K^{-1}\frac{\abs{\ee^{-rT\lK}-1}}{\abs{r}}\nonumber\\
			&\leq\tilde{C}\ee^{-\a\lK}K^{-1}\ee^{\abs{r}T\lK} \nonumber\\
			&=\tilde{C}\exp\left(-\a\lK-\ln K+\abs{r}T\lK\right) \nonumber\\
			&\leq\tilde{C}\exp\left(-\g\ln K\right)
			\overset{K\to\infty}{\longrightarrow}0.
		\end{align}
		The last inequality is true for every $\g\in(0,1)$ and $K$ large enough since, in the limit of large $K$, one has $(\abs{r}T-\a)\lK/\ln K<1-\g$ in the limit of large $K$. 
		
		If we now consider $r=0$, we see that in this case
		\begin{align}
			\Exd{\langle\hat{M}^K\rangle_t}=\left(b+d\right)Z^K_0 t.
		\end{align}
		Plugging this into the above estimate of the probability leads to an even better bound.
		
		For $p>1$, we adapt the first calculations as
		\begin{align}
			\Prob{Z^K_{T\lK}>p\ee^{rT\lK}Z^K_0}
			&=\Prob{\ee^{-rT\lK}Z^K_{T\lK}-Z^K_0>-(1-p)Z^K_0} \nonumber\\
			&\leq\Prob{\sup_{t\leq T\lK}\abs{\ee^{-rt}Z^K_t-Z^K_0}>(p-1)Z^K_0},
		\end{align}
		from where we can continue as before.
	\end{proof}

	\begin{remark}
		Note that we do not only prove the convergence to zero here but also determine the speed of convergence, which is  $o(K^{-\g})$, for all $\g\in(0,1)$.
	\end{remark}

	\begin{corollary}
		\label{cor:growth_in_lambda_time}
		Let $Z^K$ be birth death process with time-dependent rates $b^K, d^K$ as introduced in Appendix \ref{app:B.1_pureBDproc} and recall that $g^K(t)=\int_0^tr^K(s)\dd s$, where $r^K(s)=b^K(s)-d^K(s)$ is the net growth rate. For initial values $Z_0^K=\ve^2K$, for all $0<p_1<1<p_2$, all $T\geq 0$, and all $\g\in(0,1)$, we obtain
		\begin{align}
			\Prob{p_1\ee^{g^K(t)}Z^K_0<Z^K_t<p_2\ee^{g^K(t)}Z^K_0\quad\forall t\in[0,T\lK]}=1-o(K^{-\g})\overset{K\to\infty}{\longrightarrow}1.
		\end{align}
	\end{corollary}

	\begin{proof}
		The corollary can be deduced easily by iterative application of Lemma \ref{lem:growth_in_lambda_time}, combined with the Markov property at the times $T^\S_i$.
	\end{proof}

\bibliographystyle{abbrv}

\end{document}